\documentclass[12pt]{amsart}
\usepackage{amsfonts}
\usepackage{wasysym}
\usepackage{amsmath}
\usepackage{amssymb}
\usepackage{color}
\usepackage{graphicx}
\usepackage{xspace}
\usepackage{axodraw4j}
\usepackage{enumerate}
\usepackage{tikz}
\usepackage{pstricks}

 \font \eightrm=cmr8
 
\input xy
\xyoption{all}

 \newcommand{\nc}{\newcommand}

\nc{\surj}{\to\hskip -3mm \to}
\nc{\smallsq}{\mathop{\scriptscriptstyle\square}\nolimits}
\nc{\butcher}{{\scriptstyle\circleright}}
\nc{\rbutcher}{{\scriptstyle\circleleft}}
\nc\sqbullet{\mathop{\scriptscriptstyle\blacksquare}\nolimits}

\nc{\smop}[1]{\mathop{\hbox {\eightrm #1} }\nolimits}

\nc{\un}{{\mathbf 1}}
\nc{\Ab}{{\mathbb A}}
\nc{\Kb}{{\mathbb K}}
\nc{\Db}{{\mathbb D}}
\nc{\Rb}{{\mathbb R}}
\nc{\la}{{\langle}}
\nc{\ra}{{\rangle}}
\nc{\sh}{{\,\llcorner\!\llcorner\!\!\!\lrcorner\,}}
\nc{\sq}{{\,\llcorner\!\!\!\lrcorner\,}}
\nc{\Hc}{{\mathcal H}}
\nc{\Lc}{{\mathcal L}}
\nc{\Ec}{{E}}
\nc{\Gc}{{\mathrm g}}
\nc{\zsf}{{\mathrm z}}
\nc{\psf}{{\mathrm d}}
\nc{\nsf}{{\mathrm n}}
\nc{\unit}{{1}}
\nc{\R}{{\mathbb R}}
\nc{\CC}{{\mathbb{C}}}
\nc{\T}{{\mathbb{T}}}
\nc{\Hb}{{\mathbb H}}
\nc{\Ub}{{\mathbb U}}
\nc{\Vb}{{\mathbb V}}
\nc{\Sb}{{\mathbb S}}
\nc{\Nb}{{\mathbb N}}
\nc{\Zb}{{\mathbb Z}}
\nc{\Lb}{{\mathbb L}}
\nc{\Xc}{{\mathcal X}}
\nc{\conc}{{\mathrm{co}}}
\nc{\shuf}{{\mathrm{sh}}}
\nc{\Ap}{{\mathbb A^+}}
\nc{\ind}{{\pi}}

\def\baarr{
 \begin{tikzpicture}
\draw[ultra thick] (0.1,-0.4) -- (0.1,0); 
 \end{tikzpicture}}
 
\def\nci{
 \begin{tikzpicture}
\draw[thick] (0,-1) -- (0,0); 
 \end{tikzpicture}}

\def\ncii{
 \begin{tikzpicture}
\draw[thick]  (0,0) -- (1,0);
\draw[thick] (0,-1) -- (0,0); 
\draw[thick]  (1,-1) -- (1,0); 
 \end{tikzpicture}}

\def\nciiib{
 \begin{tikzpicture}
\draw[thick]  (0,0) -- (1,0);
\draw[thick] (0,-1) -- (0,0); 
\draw[thick]  (0.5,-1) -- (0.5,-0.2); 
\draw[thick] (1,-1) -- (1,0); 
\draw[thick]  (0.5,-0.2) -- (1.5,-0.2);
\draw[thick]  (1.5,-1) -- (1.5,-0.2); 
 \end{tikzpicture}}

\def\nciiic{
 \begin{tikzpicture}
\draw[thick]  (0,0) -- (2,0);
\draw[thick] (0,-1) -- (0,0); 
\draw[thick]  (0.5,-1) -- (0.5,-0.2); 
\draw[thick]  (0.5,-0.2) -- (1.5,-0.2);
\draw[thick]  (1.5,-1) -- (1.5,-0.2); 
\draw[thick] (2,-1) -- (2,0); 
 \end{tikzpicture}}

\def\nciiia{
 \begin{tikzpicture}
\draw[thick]  (0,0) -- (2,0);
\draw[thick] (0,-1) -- (0,0); 
\draw[thick]  (1,-1) -- (1,-0.2); 
\draw[thick]  (2,-1) -- (2,0); 
 \end{tikzpicture}}

\def\nciiiiia{
 \begin{tikzpicture}
\draw[thick]  (0,0) -- (3,0);
\draw[thick] (0,-1) -- (0,0); 
\draw[thick]  (0.75,-1) -- (0.75,-.2); 
\draw[thick]  (1.5,-.2) -- (2.5,-.2); 
\draw[thick]  (1.5,-1) -- (1.5,-.2); 
\draw[thick]  (2.5,-1) -- (2.5,-.2); 
\draw[thick]  (3,-1) -- (3,0); 
 \end{tikzpicture}}

\def\nciiiiiia{
 \begin{tikzpicture}
\draw[thick]  (0,0) -- (4,0);
\draw[thick] (0,-1) -- (0,0); 
\draw[thick]  (0.5,-1) -- (0.5,-.2); 
\draw[thick]  (1,-1) -- (1,0); 
\draw[thick]  (1.5,-.2) -- (3.5,-.2); 
\draw[thick]  (1.5,-1) -- (1.5,-.2); 
\draw[thick]  (3.5,-1) -- (3.5,-.2); 
\draw[thick]  (2,-1) -- (2,-.4); 
\draw[thick]  (2.5,-1) -- (2.5,-.2); 
\draw[thick]  (4,-1) -- (4,0); 
\draw[thick]  (4.5,0) -- (5.5,0); 
\draw[thick]  (4.5,-1) -- (4.5,0); 
\draw[thick]  (5.5,-1) -- (5.5,0); 

 \end{tikzpicture}}

\def\aFPxxyyFPb{\scalebox{0.3}{
\fcolorbox{white}{white}{
  \begin{picture}(224,66) (35,-63)
    \SetWidth{1.0}
    \SetColor{Black}
    \GOval(128,-30)(32,32)(0){0.882}
    \SetWidth{1.5}
    \Line(48,-30)(96,-30)
    \GOval(128,-30)(32,32)(0){0.882}
    \Line(160,-30)(208,-30)
    \Text(32,-30)[lb]{\Large{\Black{$1$}}}
    \Text(215,-30)[lb]{\Large{\Black{$2$}}}
  \end{picture}
}
}}

\def\aFPxxyyFPbAMPU{\scalebox{0.3}{
\fcolorbox{white}{white}{
  \begin{picture}(224,66) (35,-63)
    \SetWidth{1.0}
    \SetColor{Black}
    \GOval(128,-30)(32,32)(0){0.882}
    \SetWidth{1.5}
    \Line(48,-30)(96,-30)
    \SetWidth{1.5}
    \GOval(128,-30)(32,32)(0){0.882}
    \SetWidth{1.5} 
    \Line(160,-30)(178,-30)
    \Text(32,-30)[lb]{\Large{\Black{$1$}}}
    \Text(185,-30)[lb]{\Large{\Black{$2$}}}
  \end{picture}
}
}}

\def\aFPxxyyFPzzuuFPb{\scalebox{0.3}{
\fcolorbox{white}{white}{
  \begin{picture}(368,66) (35,-63)
    \SetWidth{1.0}
    \SetColor{Black}
    \GOval(128,-30)(32,32)(0){0.882}
    \SetWidth{1.5}
    \GOval(272,-30)(32,32)(0){0.882}
    \Line(160,-30)(240,-30)
    \Line(48,-30)(96,-30)
    \GOval(128,-30)(32,32)(0){0.882}
    \Line(304,-30)(352,-30)
    \Text(32,-30)[lb]{\Large{\Black{$1$}}}
    \Text(360,-30)[lb]{\Large{\Black{$2$}}}
  \end{picture}
}
}}

\def\gTgTgTg{\scalebox{0.3}{
\fcolorbox{white}{white}{
  \begin{picture}(512,66) (35,-63)
    \SetWidth{1.0}
    \SetColor{Black}
    \GOval(128,-30)(32,32)(0){0.882}
    \SetWidth{1.5}
    \GOval(416,-30)(32,32)(0){0.882}
    \Line(160,-30)(240,-30)
    \Line(48,-30)(96,-30)
    \GOval(128,-30)(32,32)(0){0.882}
    \Line(448,-30)(496,-30)
    \Text(32,-30)[lb]{\Large{\Black{$1$}}}
    \Text(512,-30)[lb]{\Large{\Black{$2$}}}
    \GOval(272,-30)(32,32)(0){0.882}
    \Line(304,-30)(384,-30)
  \end{picture}
}
}}

\def\aFPxxyyFPzzuvuFPbvFPc{\scalebox{0.3}{
\fcolorbox{white}{white}{
  \begin{picture}(368,124) (35,-43)
    \SetWidth{1.0}
    \SetColor{Black}
    \GOval(128,28)(32,32)(0){0.882}
    \SetWidth{1.5}
    \GOval(272,28)(32,32)(0){0.882}
    \Line(160,28)(240,28)
    \Line(48,28)(96,28)
    \GOval(128,28)(32,32)(0){0.882}
    \Line(272,-4)(272,-52)
    \Line(304,28)(352,28)
    \Text(32,28)[lb]{\Large{\Black{$1$}}}
    \Text(360,28)[lb]{\Large{\Black{$2$}}}
    \Text(272,-68)[lb]{\Large{\Black{$3$}}}
  \end{picture}
}
}}

\def\aFPxbFPyxyzzFPuuvvFPc{\scalebox{0.3}{
\fcolorbox{white}{white}{
  \begin{picture}(214,212) (35,80)
    \SetWidth{1.0}
    \SetColor{Black}
    \GOval(128,156)(32,32)(0){0.882}
    \SetWidth{1.5}
    \GOval(128,28)(32,32)(0){0.882}
    \Line(128,124)(128,60)
    \Line(48,156)(96,156)
    \GOval(128,156)(32,32)(0){0.882}
    \Line(208,156)(160,156)
    \Line(128,-4)(128,-52)
    \Text(32,156)[lb]{\Large{\Black{$1$}}}
    \Text(128,-68)[lb]{\Large{\Black{$3$}}}
    \Text(215,156)[lb]{\Large{\Black{$2$}}}
  \end{picture}
}
}}

\def\aFPxcFPzxyzyFPuuvvFPc{\scalebox{0.3}{
\fcolorbox{white}{white}{
  \begin{picture}(368,124) (35,-43)
    \SetWidth{1.0}
    \SetColor{Black}
    \GOval(128,28)(32,32)(0){0.882}
    \SetWidth{1.5}
    \GOval(272,28)(32,32)(0){0.882}
    \Line(160,28)(240,28)
    \Line(48,28)(96,28)
    \GOval(128,28)(32,32)(0){0.882}
    \Line(128,-4)(128,-52)
    \Line(304,28)(352,28)
    \Text(32,28)[lb]{\Large{\Black{$1$}}}
    \Text(360,28)[lb]{\Large{\Black{$2$}}}
    \Text(128,-68)[lb]{\Large{\Black{$3$}}}
  \end{picture}
}
}}

\def\kk{\scalebox{0.3}{
\fcolorbox{white}{white}{
  \begin{picture}(448,124) (35,0)
    \SetWidth{1.0}
    \SetColor{Black}
    \GOval(128,28)(32,32)(0){0.882}
    \SetWidth{1.5}
    \GOval(352,28)(32,32)(0){0.882}
    \Line(160,28)(208,28)
    \Line(48,28)(96,28)
    \GOval(128,28)(32,32)(0){0.882}
    \Line(352,-4)(352,-52)
    \Line(384,28)(432,28)
    \Text(32,28)[lb]{\Large{\Black{$1$}}}
    \Text(352,-68)[lb]{\Large{\Black{$3$}}}
    \Text(440,28)[lb]{\Large{\Black{$2$}}}
    \GOval(240,28)(32,32)(0){0.882}
    \Line(272,28)(320,28)
  \end{picture}
}
}}
  
\def\ll{\scalebox{0.3}{
\fcolorbox{white}{white}{
  \begin{picture}(224,48) (47,227)
    \SetWidth{1.0}
    \SetColor{Black}
    \GOval(128,252)(32,32)(0){0.882}
    \SetWidth{1.5}
    \GOval(128,28)(32,32)(0){0.882}
    \Line(128,220)(128,172)
    \Line(48,252)(96,252)
    \GOval(128,252)(32,32)(0){0.882}
    \Line(208,252)(160,252)
    \Line(128,-4)(128,-52)
    \Text(32,252)[lb]{\Large{\Black{$1$}}}
    \Text(128,-68)[lb]{\Large{\Black{$3$}}}
    \Text(215,252)[lb]{\Large{\Black{$2$}}}
    \GOval(128,140)(32,32)(0){0.882}
    \Line(128,108)(128,60)
  \end{picture}
}
}}
 
\def\mm{\scalebox{0.3}{
\fcolorbox{white}{white}{
  \begin{picture}(448,124) (35,0)
    \SetWidth{1.0}
    \SetColor{Black}
    \GOval(128,28)(32,32)(0){0.882}
    \SetWidth{1.5}
    \GOval(352,28)(32,32)(0){0.882}
    \Line(160,28)(208,28)
    \Line(48,28)(96,28)
    \GOval(128,28)(32,32)(0){0.882}
    \Line(128,-4)(128,-52)
    \Line(384,28)(432,28)
    \Text(32,28)[lb]{\Large{\Black{$1$}}}
    \Text(440,28)[lb]{\Large{\Black{$2$}}}
    \Text(128,-68)[lb]{\Large{\Black{$3$}}}
    \GOval(240,28)(32,32)(0){0.882}
    \Line(272,28)(320,28)
  \end{picture}
}
}}

\def\nn{\scalebox{0.3}{
\fcolorbox{white}{white}{
  \begin{picture}(336,236) (35,115)
    \SetWidth{1.0}
    \SetColor{Black}
    \GOval(128,140)(32,32)(0){0.882}
    \SetWidth{1.5}
    \GOval(240,28)(32,32)(0){0.882}
    \Line(160,140)(208,140)
    \Line(48,140)(96,140)
    \GOval(128,140)(32,32)(0){0.882}
    \Line(320,140)(272,140)
    \Line(240,-4)(240,-52)
    \Text(32,140)[lb]{\Large{\Black{$1$}}}
    \Text(240,-68)[lb]{\Large{\Black{$3$}}}
    \Text(330,140)[lb]{\Large{\Black{$2$}}}
    \GOval(240,140)(32,32)(0){0.882}
    \Line(240,108)(240,60)
  \end{picture}
}
}}
 
\def\oo{\scalebox{0.3}{
\fcolorbox{white}{white}{
  \begin{picture}(448,124) (35,0)
    \SetWidth{1.0}
    \SetColor{Black}
    \GOval(128,28)(32,32)(0){0.882}
    \SetWidth{1.5}
    \GOval(352,28)(32,32)(0){0.882}
    \Line(160,28)(208,28)
    \Line(48,28)(96,28)
    \GOval(128,28)(32,32)(0){0.882}
    \Line(240,-4)(240,-52)
    \Line(384,28)(432,28)
    \Text(32,28)[lb]{\Large{\Black{$1$}}}
    \Text(240,-68)[lb]{\Large{\Black{$3$}}}
    \Text(440,28)[lb]{\Large{\Black{$2$}}}
    \GOval(240,28)(32,32)(0){0.882}
    \Line(272,28)(320,28)
  \end{picture}
}
}}

\def\pp{\scalebox{0.3}{
\fcolorbox{white}{white}{
  \begin{picture}(336,36) (35,120)
    \SetWidth{1.0}
    \SetColor{Black}
    \GOval(128,140)(32,32)(0){0.882}
    \SetWidth{1.5}
    \GOval(128,28)(32,32)(0){0.882}
    \Line(160,140)(208,140)
    \Line(48,140)(96,140)
    \GOval(128,140)(32,32)(0){0.882}
    \Line(128,108)(128,60)
    \Line(128,-4)(128,-52)
    \Text(32,140)[lb]{\Large{\Black{$1$}}}
    \Text(128,-68)[lb]{\Large{\Black{$3$}}}
    \Text(330,140)[lb]{\Large{\Black{$2$}}}
    \GOval(240,140)(32,32)(0){0.882}
    \Line(272,140)(320,140)
  \end{picture}
}
}}

\def\abI{\scalebox{0.4}{
\fcolorbox{white}{white}{
  \begin{picture}(192,130) (19,-31)
    \SetWidth{1.5}
    \SetColor{Black}
    \COval(96,34)(64,64)(0){Black}{White}
    \SetWidth{1.0}
    \Vertex(160,34){5}
    \Vertex(32,34){5}
    \SetWidth{1.5}
    \Line(32,34)(160,34)
    \Text(16,34)[lb]{\Large{\Black{$1$}}}
    \Text(170,34)[lb]{\Large{\Black{$2$}}}
  \end{picture}
}
}}

\def\abII{\scalebox{0.4}{
\fcolorbox{white}{white}{
  \begin{picture}(192,130) (19,-31)
    \SetWidth{1.5}
    \SetColor{Black}
    \COval(96,34)(64,64)(0){Black}{White}
    \SetWidth{1.0}
    \Vertex(160,34){5}
    \Vertex(32,34){5}
    \SetWidth{1.5}
    \Line(32,34)(160,34)
    \SetWidth{1.0}
    \Vertex(96,34){5}
    \Text(16,34)[lb]{\Large{\Black{$1$}}}
    \Text(170,34)[lb]{\Large{\Black{$2$}}}
  \end{picture}
}
}}

\def\abIII{\scalebox{0.4}{
\fcolorbox{white}{white}{
  \begin{picture}(192,130) (19,-31)
    \SetWidth{1.5}
    \SetColor{Black}
    \COval(96,34)(64,64)(0){Black}{White}
    \SetWidth{1.0}
    \Vertex(160,34){5}
    \Vertex(32,34){5}
    \SetWidth{1.5}
    \Line(32,34)(160,34)
    \SetWidth{1.0}
    \Vertex(118,34){5}
    \Vertex(74,34){5}
    \Text(16,34)[lb]{\Large{\Black{$1$}}}
    \Text(170,34)[lb]{\Large{\Black{$2$}}}
    \Text(90,40)[lb]{\Large{\Black{$x_1$}}}
  \end{picture}
}
}}

\def\abIV{\scalebox{0.4}{
\fcolorbox{white}{white}{
  \begin{picture}(192,130) (19,-31)
    \SetWidth{1.5}
    \SetColor{Black}
    \COval(96,34)(64,64)(0){Black}{White}
    \SetWidth{1.0}
    \Vertex(160,34){5}
    \Vertex(32,34){5}
    \SetWidth{1.5}
    \Line(32,34)(160,34)
    \SetWidth{1.0}
    \Vertex(128,34){5}
    \Vertex(64,34){5}
    \Text(16,34)[lb]{\Large{\Black{$1$}}}
    \Text(170,34)[lb]{\Large{\Black{$2$}}}
    \Text(75,40)[lb]{\Large{\Black{$x_2$}}}
    \Text(109,40)[lb]{\Large{\Black{$x_1$}}}
    \Vertex(96,34){5}
  \end{picture}
}
}}

\def\abci{\scalebox{0.4}{
\fcolorbox{white}{white}{
  \begin{picture}(189,141) (22,15)
    \SetWidth{1.5}
    \SetColor{Black}
    \COval(96,45)(64,64)(0){Black}{White}
    \SetWidth{1.0}
    \Vertex(159,59){5}
    \Vertex(34,61){5}
    \Vertex(96,61){5}
    \Vertex(96,-19){5}
    \Text(19,60)[lb]{\Large{\Black{$1$}}}
    \Text(170,61)[lb]{\Large{\Black{$2$}}}
    \Vertex(64,61){5}
    \SetWidth{1.5}
    \Line(160,61)(32,61)
    \Line(96,61)(96,-19)
    \Text(95,-36)[lb]{\Large{\Black{$3$}}}
    \Text(75,67)[lb]{\Large{\Black{$x_1$}}}
  \end{picture}
}
}}

\def\abcii{\scalebox{0.4}{
\fcolorbox{white}{white}{
  \begin{picture}(189,141) (22,15)
    \SetWidth{1.5}
    \SetColor{Black}
    \COval(96,45)(64,64)(0){Black}{White}
    \SetWidth{1.0}
    \Vertex(159,59){5}
    \Vertex(34,61){5}
    \Vertex(96,61){5}
    \Vertex(96,-19){5}
    \Text(19,60)[lb]{\Large{\Black{$1$}}}
    \Text(170,61)[lb]{\Large{\Black{$2$}}}
    \Vertex(96,29){5}
    \SetWidth{1.5}
    \Line(160,61)(32,61)
    \Line(96,61)(96,-19)
    \Text(95,-36)[lb]{\Large{\Black{$3$}}}
    \Text(103,36)[lb]{\Large{\Black{$x_1$}}}
  \end{picture}
}
}}

\def\abciii{\scalebox{0.4}{
\fcolorbox{white}{white}{
  \begin{picture}(189,141) (22,15)
    \SetWidth{1.5}
    \SetColor{Black}
    \COval(96,45)(64,64)(0){Black}{White}
    \SetWidth{1.0}
    \Vertex(159,59){5}
    \Vertex(34,61){5}
    \Vertex(96,61){5}
    \Vertex(96,-19){5}
    \Text(19,60)[lb]{\Large{\Black{$1$}}}
    \Text(170,61)[lb]{\Large{\Black{$2$}}}
    \Vertex(128,61){5}
    \SetWidth{1.5}
    \Line(160,61)(32,61)
    \Line(96,61)(96,-19)
    \Text(95,-36)[lb]{\Large{\Black{$3$}}}
    \Text(108,65)[lb]{\Large{\Black{$x_1$}}}
  \end{picture}
}
}}

\def\abcI{\scalebox{0.4}{
\fcolorbox{white}{white}{
  \begin{picture}(192,140) (19,-31)
    \SetWidth{1.5}
    \SetColor{Black}
    \COval(96,44)(64,64)(0){Black}{White}
    \SetWidth{1.0}
    \Vertex(159,58){5}
    \Vertex(34,60){5}
    \Vertex(83,43){5}
    \Vertex(96,-20){5}
    \Text(16,60)[lb]{\Large{\Black{$1$}}}
    \Text(170,60)[lb]{\Large{\Black{$2$}}}
    \Vertex(107,33){5}
    \SetWidth{1.5}
    \Line(35,60)(104,34)
    \Line(159,58)(107,32)
    \Line(106,32)(97,-19)
    \Text(96,-36)[lb]{\Large{\Black{$3$}}}
    \SetWidth{1.0}
    \Vertex(60,51){5}
    \Text(70,50)[lb]{\Large{\Black{$x_2$}}}
    \Text(95,40)[lb]{\Large{\Black{$x_1$}}}
  \end{picture}
}
}}

\def\abcII{\scalebox{0.4}{
\fcolorbox{white}{white}{
  \begin{picture}(189,141) (22,-31)
    \SetWidth{1.5}
    \SetColor{Black}
    \COval(96,45)(64,64)(0){Black}{White}
    \SetWidth{1.0}
    \Vertex(34,61){5}
    \Vertex(96,-19){5}
    \Text(19,60)[lb]{\Large{\Black{$1$}}}
    \Text(170,61)[lb]{\Large{\Black{$2$}}}
    \Text(95,-36)[lb]{\Large{\Black{$3$}}}
    \Vertex(160,61){5}
    \Vertex(96,61){5}
    \Vertex(96,8){5}
    \Vertex(96,34){5}
    \SetWidth{1.5}
    \Line(32,61)(160,61)
    \Line(96,61)(96,-19)
    \Text(104,40)[lb]{\Large{\Black{$x_2$}}}
    \Text(104,13)[lb]{\Large{\Black{$x_1$}}}
  \end{picture}
}
}}

\def\abcIII{\scalebox{0.4}{
\fcolorbox{white}{white}{
  \begin{picture}(189,141) (22,-31)
    \SetWidth{1.5}
    \SetColor{Black}
    \COval(96,45)(64,64)(0){Black}{White}
    \SetWidth{1.0}
    \Vertex(159,59){5}
    \Vertex(34,61){5}
    \Vertex(64,61){5}
    \Vertex(96,-19){5}
    \Text(19,60)[lb]{\Large{\Black{$1$}}}
    \Text(170,61)[lb]{\Large{\Black{$2$}}}
    \Vertex(128,61){5}
    \SetWidth{1.5}
    \Line(32,61)(160,61)
    \Line(64,61)(96,-19)
    \Text(95,-36)[lb]{\Large{\Black{$3$}}}
    \SetWidth{1.0}
    \Vertex(96,61){5}
    \Text(75,67)[lb]{\Large{\Black{$x_2$}}}
    \Text(107,67)[lb]{\Large{\Black{$x_1$}}}
  \end{picture}
}
}}

\def\abcIV{\scalebox{0.4}{
\fcolorbox{white}{white}{
  \begin{picture}(189,141) (22,-31)
    \SetWidth{1.5}
    \SetColor{Black}
    \COval(96,45)(64,64)(0){Black}{White}
    \SetWidth{1.0}
    \Vertex(160,61){5}
    \Vertex(96,29){5}
    \Vertex(32,61){5}
    \Vertex(96,-19){5}
    \Text(19,60)[lb]{\Large{\Black{$1$}}}
    \Text(170,61)[lb]{\Large{\Black{$2$}}}
    \Vertex(96,61){5}
    \SetWidth{1.5}
    \Line(96,-19)(96,61)
    \Line(32,61)(160,61)
    \Text(95,-36)[lb]{\Large{\Black{$3$}}}
    \SetWidth{1.0}
    \Vertex(64,61){5}
    \Text(75,67)[lb]{\Large{\Black{$x_2$}}}
    \Text(102,40)[lb]{\Large{\Black{$x_1$}}}
  \end{picture}
}
}}

\def\abcV{\scalebox{0.4}{
\fcolorbox{white}{white}{
  \begin{picture}(189,141) (22,-31)
    \SetWidth{1.5}
    \SetColor{Black}
    \COval(96,45)(64,64)(0){Black}{White}
    \SetWidth{1.0}
    \Vertex(159,59){5}
    \Vertex(34,61){5}
    \Vertex(96,61){5}
    \Vertex(96,-19){5}
    \Text(19,60)[lb]{\Large{\Black{$1$}}}
    \Text(170,61)[lb]{\Large{\Black{$2$}}}
    \Vertex(128,61){5}
    \SetWidth{1.5}
    \Line(160,61)(32,61)
    \Line(96,61)(96,-19)
    \Text(95,-36)[lb]{\Large{\Black{$3$}}}
    \SetWidth{1.0}
    \Vertex(64,61){5}
    \Text(75,67)[lb]{\Large{\Black{$x_2$}}}
    \Text(107,67)[lb]{\Large{\Black{$x_1$}}}
  \end{picture}
}
}}

\def\abcVI{\scalebox{0.4}{
\fcolorbox{white}{white}{
  \begin{picture}(189,141) (22,-31)
    \SetWidth{1.5}
    \SetColor{Black}
    \COval(96,45)(64,64)(0){Black}{White}
    \SetWidth{1.0}
    \Vertex(160,61){5}
    \Vertex(128,61){5}
    \Vertex(32,61){5}
    \Vertex(96,-19){5}
    \Text(19,60)[lb]{\Large{\Black{$1$}}}
    \Text(170,61)[lb]{\Large{\Black{$2$}}}
    \Vertex(96,61){5}
    \SetWidth{1.5}
    \Line(96,-19)(96,61)
    \Line(32,61)(160,61)
    \Text(95,-36)[lb]{\Large{\Black{$3$}}}
    \SetWidth{1.0}
    \Vertex(96,29){5}
    \Text(102,35)[lb]{\Large{\Black{$x_1$}}}
    \Text(107,67)[lb]{\Large{\Black{$x_2$}}}
  \end{picture}
}
}}

\def\abcVII{\scalebox{0.4}{
\fcolorbox{white}{white}{
  \begin{picture}(189,141) (22,-31)
    \SetWidth{1.5}
    \SetColor{Black}
    \COval(96,45)(64,64)(0){Black}{White}
    \SetWidth{1.0}
    \Vertex(160,61){5}
    \Vertex(128,61){5}
    \Vertex(32,61){5}
    \Vertex(96,-19){5}
    \Text(19,60)[lb]{\Large{\Black{$1$}}}
    \Text(170,61)[lb]{\Large{\Black{$2$}}}
    \Vertex(96,61){5}
    \SetWidth{1.5}
    \Line(96,-19)(96,61)
    \Line(32,61)(160,61)
    \Text(95,-36)[lb]{\Large{\Black{$3$}}}
    \SetWidth{1.0}
    \Vertex(96,29){5}
    \Text(102,35)[lb]{\Large{\Black{$x_2$}}}
    \Text(107,67)[lb]{\Large{\Black{$x_1$}}}
  \end{picture}
}
}}

\def\treeB{\scalebox{0.4}{
\fcolorbox{white}{white}{
  \begin{picture}(42,80) (57,-43)
    \SetWidth{1.0}
    \SetColor{Black}
    \Vertex(64,-32){5.657}
    \SetWidth{1.5}
    \Line(64,-32)(64,0)
    \SetWidth{1.0}
    \Vertex(64,0){5.657}
    \Text(64,-52)[lb]{\Large{\Black{$1$}}}
    \Text(64,8)[lb]{\Large{\Black{$2$}}}
  \end{picture}
}
}}

\def\treeC{\scalebox{0.4}{
\fcolorbox{white}{white}{
  \begin{picture}(58,112) (57,-11)
    \SetWidth{1.0}
    \SetColor{Black}
    \Vertex(64,0){5.657}
    \SetWidth{1.5}
    \Line(64,0)(64,32)
    \SetWidth{1.0}
    \Vertex(64,32){5.657}
    \SetWidth{1.5}
    \Line(64,32)(64,64)
    \SetWidth{1.0}
    \Vertex(64,64){5.657}
    \Text(74,32)[lb]{\Large{\Black{$x_1$}}}
    \Text(64,72)[lb]{\Large{\Black{$2$}}}
    \Text(64,-20)[lb]{\Large{\Black{$1$}}}
  \end{picture}
}
}}

\def\treeD{\scalebox{0.4}{
\fcolorbox{white}{white}{
  \begin{picture}(58,144) (57,-11)
    \SetWidth{1.0}
    \SetColor{Black}
    \Vertex(64,0){5.657}
    \SetWidth{1.5}
    \Line(64,32)(64,64)
    \SetWidth{1.0}
    \Vertex(64,64){5.657}
    \SetWidth{1.5}
    \Line(64,64)(64,96)
    \SetWidth{1.0}
    \Vertex(64,96){5.657}
    \Text(74,64)[lb]{\Large{\Black{$x_1$}}}
    \Text(64,104)[lb]{\Large{\Black{$2$}}}
    \Text(64,-20)[lb]{\Large{\Black{$1$}}}
    \SetWidth{1.5}
    \Line(64,0)(64,32)
    \SetWidth{1.0}
    \Vertex(64,32){5.657}
    \Text(74,32)[lb]{\Large{\Black{$x_2$}}}
  \end{picture}
}
}}

\def\treeDa{\scalebox{0.4}{
\fcolorbox{white}{white}{
  \begin{picture}(74,112) (41,-11)
    \SetWidth{1.0}
    \SetColor{Black}
    \Vertex(64,0){5.657}
    \SetWidth{1.5}
    \Line(64,0)(64,32)
    \SetWidth{1.0}
    \Vertex(64,32){5.657}
    \SetWidth{1.5}
    \Line(64,32)(48,64)
    \SetWidth{1.0}
    \Vertex(48,64){5.657}
    \SetWidth{1.5}
    \Line(64,32)(80,64)
    \SetWidth{1.0}
    \Vertex(80,64){5.657}
    \Text(74,32)[lb]{\Large{\Black{$x_1$}}}
    \Text(64,-20)[lb]{\Large{\Black{$1$}}}
    \Text(48,74)[lb]{\Large{\Black{$2$}}}
    \Text(80,74)[lb]{\Large{\Black{$3$}}}
  \end{picture}
}
}}

\def\treeDb{\scalebox{0.4}{
\fcolorbox{white}{white}{
  \begin{picture}(90,112) (41,-11)
    \SetWidth{1.0}
    \SetColor{Black}
    \Vertex(64,0){5.657}
    \SetWidth{1.5}
    \Line(64,0)(80,32)
    \SetWidth{1.0}
    \Vertex(80,32){5.657}
    \SetWidth{1.5}
    \Line(64,0)(48,32)
    \SetWidth{1.0}
    \Vertex(48,32){5.657}
    \SetWidth{1.5}
    \Line(80,32)(80,64)
    \SetWidth{1.0}
    \Vertex(80,64){5.657}
    \Text(90,32)[lb]{\Large{\Black{$x_1$}}}
    \Text(64,-20)[lb]{\Large{\Black{$1$}}}
    \Text(48,40)[lb]{\Large{\Black{$2$}}}
    \Text(80,74)[lb]{\Large{\Black{$3$}}}
  \end{picture}
}
}}

\def\treeDc{\scalebox{0.4}{
\fcolorbox{white}{white}{
  \begin{picture}(80,112) (35,-11)
    \SetWidth{1.0}
    \SetColor{Black}
    \Vertex(64,0){5.657}
    \SetWidth{1.5}
    \Line(64,0)(80,32)
    \SetWidth{1.0}
    \Vertex(80,32){5.657}
    \SetWidth{1.5}
    \Line(64,0)(48,32)
    \SetWidth{1.0}
    \Vertex(48,32){5.657}
    \SetWidth{1.5}
    \Line(48,32)(48,64)
    \SetWidth{1.0}
    \Vertex(48,64){5.657}
    \Text(26,32)[lb]{\Large{\Black{$x_1$}}}
    \Text(64,-20)[lb]{\Large{\Black{$1$}}}
    \Text(48,74)[lb]{\Large{\Black{$2$}}}
    \Text(80,40)[lb]{\Large{\Black{$3$}}}
  \end{picture}
}
}}

\def\treeEa{\scalebox{0.4}{
\fcolorbox{white}{white}{
  \begin{picture}(74,144) (41,-11)
    \SetWidth{1.0}
    \SetColor{Black}
    \Vertex(64,32){5.657}
    \SetWidth{1.5}
    \Line(64,32)(64,64)
    \SetWidth{1.0}
    \Vertex(80,96){5.657}
    \SetWidth{1.5}
    \Line(64,32)(64,0)
    \SetWidth{1.0}
    \Vertex(64,64){5.657}
    \SetWidth{1.5}
    \Line(64,64)(48,96)
    \SetWidth{1.0}
    \Vertex(48,96){5.657}
    \Text(74,64)[lb]{\Large{\Black{$x_1$}}}
    \Text(64,-20)[lb]{\Large{\Black{$1$}}}
    \Text(48,104)[lb]{\Large{\Black{$2$}}}
    \Text(80,104)[lb]{\Large{\Black{$3$}}}
    \SetWidth{1.5}
    \Line(64,64)(80,96)
    \SetWidth{1.0}
    \Vertex(64,0){5.657}
    \Text(74,32)[lb]{\Large{\Black{$x_2$}}}
  \end{picture}
}
}}

\def\treeEb{\scalebox{0.4}{
\fcolorbox{white}{white}{
  \begin{picture}(90,144) (41,-11)
    \SetWidth{1.0}
    \SetColor{Black}
    \Vertex(64,32){5.657}
    \SetWidth{1.5}
    \Line(64,32)(80,64)
    \SetWidth{1.0}
    \Vertex(80,96){5.657}
    \SetWidth{1.5}
    \Line(64,32)(64,0)
    \SetWidth{1.0}
    \Vertex(80,64){5.657}
    \SetWidth{1.5}
    \Line(64,32)(48,64)
    \SetWidth{1.0}
    \Vertex(48,64){5.657}
    \Text(90,64)[lb]{\Large{\Black{$x_1$}}}
    \Text(64,-20)[lb]{\Large{\Black{$1$}}}
    \Text(48,72)[lb]{\Large{\Black{$2$}}}
    \Text(80,104)[lb]{\Large{\Black{$3$}}}
    \SetWidth{1.5}
    \Line(80,64)(80,96)
    \SetWidth{1.0}
    \Vertex(64,0){5.657}
    \Text(74,32)[lb]{\Large{\Black{$x_2$}}}
  \end{picture}
}
}}

\def\treeEc{\scalebox{0.4}{
\fcolorbox{white}{white}{
  \begin{picture}(80,144) (35,-11)
    \SetWidth{1.0}
    \SetColor{Black}
    \Vertex(64,32){5.657}
    \SetWidth{1.5}
    \Line(64,32)(80,64)
    \SetWidth{1.0}
    \Vertex(48,96){5.657}
    \SetWidth{1.5}
    \Line(64,32)(64,0)
    \SetWidth{1.0}
    \Vertex(80,64){5.657}
    \SetWidth{1.5}
    \Line(64,32)(48,64)
    \SetWidth{1.0}
    \Vertex(48,64){5.657}
    \Text(25,64)[lb]{\Large{\Black{$x_1$}}}
    \Text(64,-20)[lb]{\Large{\Black{$1$}}}
    \Text(48,104)[lb]{\Large{\Black{$2$}}}
    \Text(80,72)[lb]{\Large{\Black{$3$}}}
    \SetWidth{1.5}
    \Line(48,64)(48,96)
    \SetWidth{1.0}
    \Vertex(64,0){5.657}
    \Text(40,32)[lb]{\Large{\Black{$x_2$}}}
  \end{picture}
}
}}

\def\treeEdi{\scalebox{0.4}{
\fcolorbox{white}{white}{
  \begin{picture}(96,112) (35,-43)
    \SetWidth{1.0}
    \SetColor{Black}
    \Vertex(48,0){5.657}
    \SetWidth{1.5}
    \Line(48,0)(48,32)
    \SetWidth{1.0}
    \Vertex(48,32){5.657}
    \SetWidth{1.5}
    \Line(80,0)(64,-32)
    \SetWidth{1.0}
    \Vertex(80,0){5.657}
    \SetWidth{1.5}
    \Line(64,-32)(48,0)
    \SetWidth{1.0}
    \Vertex(80,32){5.657}
    \Text(90,0)[lb]{\Large{\Black{$x_2$}}}
    \Text(64,-52)[lb]{\Large{\Black{$1$}}}
    \Text(48,40)[lb]{\Large{\Black{$2$}}}
    \Text(80,40)[lb]{\Large{\Black{$3$}}}
    \SetWidth{1.5}
    \Line(80,0)(80,32)
    \SetWidth{1.0}
    \Vertex(64,-32){5.657}
    \Text(25,0)[lb]{\Large{\Black{$x_1$}}}
  \end{picture}
}
}}

\def\treeEdii{\scalebox{0.4}{
\fcolorbox{white}{white}{
  \begin{picture}(96,112) (35,-43)
    \SetWidth{1.0}
    \SetColor{Black}
    \Vertex(48,0){5.657}
    \SetWidth{1.5}
    \Line(48,0)(48,32)
    \SetWidth{1.0}
    \Vertex(48,32){5.657}
    \SetWidth{1.5}
    \Line(80,0)(64,-32)
    \SetWidth{1.0}
    \Vertex(80,0){5.657}
    \SetWidth{1.5}
    \Line(64,-32)(48,0)
    \SetWidth{1.0}
    \Vertex(80,32){5.657}
    \Text(25,0)[lb]{\Large{\Black{$x_2$}}}
    \Text(64,-52)[lb]{\Large{\Black{$1$}}}
    \Text(48,40)[lb]{\Large{\Black{$2$}}}
    \Text(80,40)[lb]{\Large{\Black{$3$}}}
    \SetWidth{1.5}
    \Line(80,0)(80,32)
    \SetWidth{1.0}
    \Vertex(64,-32){5.657}
    \Text(90,0)[lb]{\Large{\Black{$x_1$}}}
  \end{picture}
}
}}

\def\treeEe{\scalebox{0.4}{
\fcolorbox{white}{white}{
  \begin{picture}(90,144) (41,-11)
    \SetWidth{1.0}
    \SetColor{Black}
    \Vertex(80,32){5.657}
    \SetWidth{1.5}
    \Line(80,32)(80,64)
    \SetWidth{1.0}
    \Vertex(80,96){5.657}
    \SetWidth{1.5}
    \Line(80,32)(64,0)
    \SetWidth{1.0}
    \Vertex(80,64){5.657}
    \SetWidth{1.5}
    \Line(64,0)(48,32)
    \SetWidth{1.0}
    \Vertex(48,32){5.657}
    \Text(90,64)[lb]{\Large{\Black{$x_1$}}}
    \Text(64,-20)[lb]{\Large{\Black{$1$}}}
    \Text(48,40)[lb]{\Large{\Black{$2$}}}
    \Text(80,104)[lb]{\Large{\Black{$3$}}}
    \SetWidth{1.5}
    \Line(80,64)(80,96)
    \SetWidth{1.0}
    \Vertex(64,0){5.657}
    \Text(90,32)[lb]{\Large{\Black{$x_2$}}}
  \end{picture}
}
}}

\def\treeEf{\scalebox{0.4}{
\fcolorbox{white}{white}{
  \begin{picture}(80,144) (35,-11)
    \SetWidth{1.0}
    \SetColor{Black}
    \Vertex(80,32){5.657}
    \SetWidth{1.5}
    \Line(48,32)(48,64)
    \SetWidth{1.0}
    \Vertex(48,96){5.657}
    \SetWidth{1.5}
    \Line(80,32)(64,0)
    \SetWidth{1.0}
    \Vertex(48,64){5.657}
    \SetWidth{1.5}
    \Line(64,0)(48,32)
    \SetWidth{1.0}
    \Vertex(48,32){5.657}
    \Text(25,64)[lb]{\Large{\Black{$x_1$}}}
    \Text(64,-20)[lb]{\Large{\Black{$1$}}}
    \Text(48,104)[lb]{\Large{\Black{$2$}}}
    \Text(80,40)[lb]{\Large{\Black{$3$}}}
    \SetWidth{1.5}
    \Line(48,64)(48,96)
    \SetWidth{1.0}
    \Vertex(64,0){5.657}
    \Text(25,32)[lb]{\Large{\Black{$x_2$}}}
  \end{picture}
}
}}

\def\treeAsym{\scalebox{0.4}{
\fcolorbox{white}{white}{
  \begin{picture}(122,176) (57,30)
    \SetWidth{1.0}
    \SetColor{Black}
    \Vertex(128,0){5.657}
    \Vertex(144,32){5.657}
    \Vertex(112,32){5.657}
    \Vertex(96,64){5.657}
    \Vertex(80,96){5.657}
    \Vertex(64,128){5.657}
    \Vertex(128,64){5.657}
    \Vertex(112,96){5.657}
    \SetWidth{1.5}
    \Line(64,128)(128,0)
    \Line(112,96)(96,64)
    \Line(144,96)(112,32)
    \Line(144,32)(128,0)
    \SetWidth{1.0}
    \Vertex(144,96){5.657}
    \Text(128,-20)[lb]{\Large{\Black{$1$}}}
    \Text(144,40)[lb]{\Large{\Black{$5$}}}
    \Text(144,104)[lb]{\Large{\Black{$4$}}}
    \Text(112,104)[lb]{\Large{\Black{$3$}}}
    \Text(64,136)[lb]{\Large{\Black{$2$}}}
    \Text(55,96)[lb]{\Large{\Black{$x_1$}}}
    \Text(71,64)[lb]{\Large{\Black{$x_3$}}}
    \Text(85,32)[lb]{\Large{\Black{$x_4$}}}
    \Text(106,64)[lb]{\Large{\Black{$x_2$}}}
  \end{picture}
}
}}

\def\treeBsym{\scalebox{0.4}{
\fcolorbox{white}{white}{
  \begin{picture}(122,176) (57,30)
    \SetWidth{1.0}
    \SetColor{Black}
    \Vertex(128,0){5.657}
    \Vertex(144,32){5.657}
    \Vertex(112,32){5.657}
    \Vertex(96,64){5.657}
    \Vertex(80,96){5.657}
    \Vertex(64,128){5.657}
    \Vertex(128,64){5.657}
    \Vertex(112,96){5.657}
    \SetWidth{1.5}
    \Line(64,128)(128,0)
    \Line(112,96)(96,64)
    \Line(144,96)(112,32)
    \Line(144,32)(128,0)
    \SetWidth{1.0}
    \Vertex(144,96){5.657}
    \Text(128,-20)[lb]{\Large{\Black{$1$}}}
    \Text(144,40)[lb]{\Large{\Black{$5$}}}
    \Text(144,104)[lb]{\Large{\Black{$4$}}}
    \Text(112,104)[lb]{\Large{\Black{$3$}}}
    \Text(64,136)[lb]{\Large{\Black{$2$}}}
    \Text(55,96)[lb]{\Large{\Black{$x_2$}}}
    \Text(71,64)[lb]{\Large{\Black{$x_3$}}}
    \Text(85,32)[lb]{\Large{\Black{$x_4$}}}
    \Text(106,64)[lb]{\Large{\Black{$x_1$}}}
  \end{picture}
}
}}

\def\FeynGraphPart{\scalebox{0.4}{
\fcolorbox{white}{white}{
  \begin{picture}(402,120) (31,-9)
    \SetWidth{1.5}
    \SetColor{Black}
    \Line(32,-8)(32,88)
    \Line(32,88)(208,88)
    \Line(208,88)(208,-8)
    \Line(64,-8)(64,40)
    \Line(64,40)(176,40)
    \Line(176,40)(176,-8)
    \Arc[clock](120,48)(56.569,135,45)
    \Arc(120,128)(56.569,-135,-45)
    \Line(256,88)(256,-8)
    \Line(256,88)(432,88)
    \Arc(344,96)(40.792,-168.69,-11.31)
    \Arc[clock](344,48)(56.569,135,45)
    \Line(432,88)(432,-8)
    \Arc[clock](345,53.933)(26.026,157.245,22.755)
    \Arc[clock](345,-0.875)(68.235,110.593,69.407)
  \end{picture}
}
}}

\def\ExpTreeA{\scalebox{0.4}{
\fcolorbox{white}{white}{
  \begin{picture}(154,144) (73,-43)
    \SetWidth{1.0}
    \SetColor{Black}
    \Vertex(128,-32){5.657}
    \Vertex(160,0){5.657}
    \Vertex(96,0){5.657}
    \Vertex(192,64){5.657}
    \Vertex(80,32){5.657}
    \Vertex(112,32){5.657}
    \Vertex(128,0){5.657}
    \Vertex(176,32){5.657}
    \SetWidth{1.5}
    \Line(176,32)(160,0)
    \Line(128,0)(128,-32)
    \Line(128,-32)(96,0)
    \Line(160,0)(128,-32)
    \SetWidth{1.0}
    \Vertex(144,32){5.657}
    \Text(144,40)[lb]{\Large{\Black{$5$}}}
    \Text(80,40)[lb]{\Large{\Black{$2$}}}
    \Text(73,0)[lb]{\Large{\Black{$x_1$}}}
    \Text(128,8)[lb]{\Large{\Black{$4$}}}
    \Text(112,40)[lb]{\Large{\Black{$3$}}}
    \Text(192,72)[lb]{\Large{\Black{$7$}}}
    \Text(160,72)[lb]{\Large{\Black{$6$}}}
    \Text(170,0)[lb]{\Large{\Black{$x_3$}}}
    \Text(128,-54)[lb]{\Large{\Black{$1$}}}
    \Vertex(160,64){5.657}
    \SetWidth{1.5}
    \Line(96,0)(80,32)
    \Line(192,64)(176,32)
    \Line(96,0)(112,32)
    \Line(160,0)(144,32)
    \Line(176,32)(160,64)
    \Text(186,32)[lb]{\Large{\Black{$x_2$}}}
  \end{picture}
}
}}

\def\ExpTreeBa{\scalebox{0.4}{
\fcolorbox{white}{white}{
  \begin{picture}(154,144) (57,-8)
    \SetWidth{1.0}
    \SetColor{Black}
    \Vertex(128,-32){5.657}
    \Vertex(160,0){5.657}
    \Vertex(96,0){5.657}
    \Vertex(80,32){5.657}
    \Vertex(64,64){5.657}
    \Vertex(112,32){5.657}
    \SetWidth{1.5}
    \Line(112,32)(96,0)
    \Line(128,-32)(96,0)
    \Line(160,0)(128,-32)
    \SetWidth{1.0}
    \Vertex(176,32){5.657}
    \Text(64,74)[lb]{\Large{\Black{$2$}}}
    \Text(112,40)[lb]{\Large{\Black{$3$}}}
    \Text(128,-52)[lb]{\Large{\Black{$1$}}}
    \SetWidth{1.5}
    \Line(96,0)(80,32)
    \Line(80,32)(64,64)
    \Line(160,0)(176,32)
    \Text(176,40)[lb]{\Large{\Black{$4$}}}
  \end{picture}
}
}}

\def\ExpTreeBb{\scalebox{0.4}{
\fcolorbox{white}{white}{
  \begin{picture}(122,112) (73,-50)
    \SetWidth{1.0}
    \SetColor{Black}
    \Vertex(128,-64){5.657}
    \Vertex(160,-32){5.657}
    \Vertex(96,-32){5.657}
    \Vertex(80,0){5.657}
    \Vertex(112,0){5.657}
    \SetWidth{1.5}
    \Line(112,0)(96,-32)
    \Line(128,-64)(96,-32)
    \Line(160,-32)(128,-64)
    \Text(80,8)[lb]{\Large{\Black{$2$}}}
    \Text(112,8)[lb]{\Large{\Black{$3$}}}
    \Text(128,-84)[lb]{\Large{\Black{$1$}}}
    \Line(96,-32)(80,0)
    \Text(160,-24)[lb]{\Large{\Black{$4$}}}
  \end{picture}
}
}}

\def\ExpTreeC{\scalebox{0.4}{
\fcolorbox{white}{white}{
  \begin{picture}(138,112) (73,-75)
    \SetWidth{1.0}
    \SetColor{Black}
    \Vertex(128,-64){5.657}
    \Vertex(160,-32){5.657}
    \Vertex(96,-32){5.657}
    \Vertex(80,0){5.657}
    \Vertex(112,0){5.657}
    \Vertex(128,-32){5.657}
    \Vertex(176,0){5.657}
    \SetWidth{1.5}
    \Line(176,0)(160,-32)
    \Line(128,-32)(128,-64)
    \Line(128,-64)(96,-32)
    \Line(160,-32)(128,-64)
    \SetWidth{1.0}
    \Vertex(144,0){5.657}
    \Text(144,8)[lb]{\Large{\Black{$5$}}}
    \Text(80,8)[lb]{\Large{\Black{$2$}}}
    \Text(74,-32)[lb]{\Large{\Black{$x_1$}}}
    \Text(128,-24)[lb]{\Large{\Black{$4$}}}
    \Text(112,8)[lb]{\Large{\Black{$3$}}}
    \Text(176,8)[lb]{\Large{\Black{$6$}}}
    \Text(170,-32)[lb]{\Large{\Black{$x_2$}}}
    \Text(128,-83)[lb]{\Large{\Black{$1$}}}
    \SetWidth{1.5}
    \Line(96,-32)(80,0)
    \Line(96,-32)(112,0)
    \Line(160,-32)(144,0)
  \end{picture}
}
}}

\def\ExpTreeCa{\scalebox{0.4}{
\fcolorbox{white}{white}{
  \begin{picture}(122,104) (89,-79)
    \SetWidth{1.0}
    \SetColor{Black}
    \Vertex(128,-72){5.657}
    \Vertex(160,-40){5.657}
    \Text(170,-46)[lb]{\Large{\Black{$x_1$}}} 
    \Vertex(96,-40){5.657}
    \Vertex(144,-8){5.657}
    \Vertex(176,-8){5.657}
    \SetWidth{1.5}
    \Line(176,-8)(160,-40)
    \Line(128,-72)(96,-40)
    \Line(160,-40)(128,-72)
    \Text(96,-31)[lb]{\Large{\Black{$2$}}}
    \Text(144,1)[lb]{\Large{\Black{$3$}}}
    \Text(128,-91)[lb]{\Large{\Black{$1$}}}
    \Line(160,-40)(144,-8)
    \Text(176,1)[lb]{\Large{\Black{$4$}}}
  \end{picture}
}
}}

\def\ExpTreeCaa{\scalebox{0.4}{
\fcolorbox{white}{white}{
  \begin{picture}(122,104) (89,-79)
    \SetWidth{1.0}
    \SetColor{Black}
    \Vertex(128,-72){7.657}
    \Vertex(160,-40){7.657}
    \Text(170,-46)[lb]{\Large{\Black{$x_1$}}} 
    \Vertex(96,-40){5.657}
    \Vertex(144,-8){5.657}
    \Vertex(176,-8){7.657}
    \SetWidth{3.5}
    \Line(176,-8)(160,-40)
    \SetWidth{1.5}
    \Line(128,-72)(96,-40)
    \SetWidth{3.5}
    \Line(160,-40)(128,-72)
    \Text(96,-31)[lb]{\Large{\Black{$2$}}}
    \Text(144,1)[lb]{\Large{\Black{$3$}}}
    \Text(128,-94)[lb]{\Large{\Black{$1$}}}
    \SetWidth{1.5}
    \Line(160,-40)(144,-8)
    \Text(176,1)[lb]{\Large{\Black{$4$}}}
  \end{picture}
}
}}

\def\ExpTreeCba{\scalebox{0.4}{
\fcolorbox{white}{white}{
  \begin{picture}(130,104) (89,-79)
    \SetWidth{1.0}
    \SetColor{Black}
    \Vertex(128,-72){7.657}
    \Vertex(160,-40){7.657}
    \Vertex(96,-40){5.657}
    \Vertex(136,-8){5.657}
    \Vertex(160,-8){5.657}
        \Text(170,-46)[lb]{\Large{\Black{$x_1$}}} 
    \SetWidth{3.5}
    \Line(184,-8)(160,-40)
   \SetWidth{1.5}
    \Line(128,-72)(96,-40)
   \SetWidth{3.5}
    \Line(160,-40)(128,-72)
    \Text(96,-31)[lb]{\Large{\Black{$2$}}}
    \Text(136,1)[lb]{\Large{\Black{$3$}}}
    \Text(128,-95)[lb]{\Large{\Black{$1$}}}
   \SetWidth{1.5}
    \Line(160,-40)(140,-12)
    \Text(160,1)[lb]{\Large{\Black{$4$}}}
    \SetWidth{1.0}
    \Vertex(184,-8){7.657}
    \SetWidth{1.5}
    \Line(160,-44)(160,-12)
    \Text(184,1)[lb]{\Large{\Black{$5$}}}
  \end{picture}
}
}}

\def\ExpTreeCca{\scalebox{0.4}{
\fcolorbox{white}{white}{
  \begin{picture}(138,140) (89,-43)
    \SetWidth{1.0}
    \SetColor{Black}
    \Vertex(128,-36){7.657}
    \Vertex(160,-4){7.657}
    \Vertex(96,-4){5.657}
    \Vertex(160,28){5.657}
     \SetWidth{3.5}
    \Line(192,28)(160,-4)
     \SetWidth{1.5}
    \Line(128,-36)(96,-4)
      \SetWidth{3.5}
    \Line(160,-4)(128,-36)
    \Text(96,4)[lb]{\Large{\Black{$2$}}}
    \Text(144,70)[lb]{\Large{\Black{$3$}}}
    \Text(128,-59)[lb]{\Large{\Black{$1$}}}
      \SetWidth{1.5}
    \Line(160,28)(144,60)
    \Text(176,70)[lb]{\Large{\Black{$4$}}}
      \SetWidth{1.5}
    \Line(160,-8)(160,24)
    \Text(192,38)[lb]{\Large{\Black{$5$}}}
      \SetWidth{1.5}
    \Line(176,60)(160,28)
    \SetWidth{1.0}
    \Vertex(144,60){5.657}
    \Vertex(176,60){5.657}
    \Vertex(192,28){7.657}
    \Text(170,-10)[lb]{\Large{\Black{$x_2$}}}
    \Text(135,24)[lb]{\Large{\Black{$x_1$}}}
  \end{picture}
}
}}

\def\ExpTreeCda{\scalebox{0.4}{
\fcolorbox{white}{white}{
  \begin{picture}(154,140) (89,-43)
    \SetWidth{1.0}
    \SetColor{Black}
    \Vertex(128,-36){7.657}
    \Vertex(160,-4){7.657}
    \Vertex(96,-4){5.657}
    \Vertex(128,28){5.657}
     \SetWidth{3.5}
    \Line(192,28)(160,-4)
     \SetWidth{1.5}
    \Line(128,-36)(96,-4)
     \SetWidth{3.5}
    \Line(160,-4)(128,-36)
    \Text(96,4)[lb]{\Large{\Black{$2$}}}
    \Text(128,36)[lb]{\Large{\Black{$3$}}}
    \Text(128,-60)[lb]{\Large{\Black{$1$}}}
     \SetWidth{1.5}
    \Line(192,28)(176,60)
    \Text(176,70)[lb]{\Large{\Black{$4$}}}
     \SetWidth{1.5}
    \Line(160,-8)(128,28)
    \Text(208,70)[lb]{\Large{\Black{$5$}}}
     \SetWidth{3.5}
    \Line(208,60)(192,28)
    \SetWidth{1.0}
    \Vertex(176,60){5.657}
    \Vertex(208,60){7.657}
    \Vertex(192,28){7.657}
    \Text(170,-10)[lb]{\Large{\Black{$x_2$}}}
    \Text(202,24)[lb]{\Large{\Black{$x_1$}}}
  \end{picture}
}
}}

\def\ExpTreeCea{\scalebox{0.4}{
\fcolorbox{white}{white}{
  \begin{picture}(106,108) (89,-75)
    \SetWidth{1.0}
    \SetColor{Black}
    \Vertex(128,-68){7.657}
    \Vertex(128,-36){5.657}
    \Vertex(96,-36){5.657}
    \Vertex(144,-4){5.657}
        \Text(136,-38)[lb]{\Large{\Black{$x_1$}}}
    \Vertex(112,-4){5.657}
    \SetWidth{1.5}
    \Line(144,-4)(128,-36)
    \Line(128,-68)(96,-36)
        \SetWidth{3.5}
    \Line(160,-36)(128,-68)
    \Text(96,-28)[lb]{\Large{\Black{$2$}}}
    \Text(112,4)[lb]{\Large{\Black{$3$}}}
    \Text(128,-92)[lb]{\Large{\Black{$1$}}}
     \SetWidth{1.5}
    \Line(128,-36)(112,-4)
    \Text(144,4)[lb]{\Large{\Black{$4$}}}
    \SetWidth{1.0}
    \Vertex(160,-36){7.657}
    \SetWidth{1.5}
    \Line(128,-36)(128,-68)
    \Text(160,-28)[lb]{\Large{\Black{$5$}}}
  \end{picture}
}
}}

\def\ExpTreeCfa{\scalebox{0.4}{
\fcolorbox{white}{white}{
  \begin{picture}(138,140) (89,-79)
    \SetWidth{1.0}
    \SetColor{Black}
    \Vertex(128,-36){5.657}
    \Vertex(160,-4){5.657}
    \Vertex(96,-4){5.657}
    \Vertex(144,28){5.657}
    \Vertex(176,28){5.657}
    \SetWidth{1.5}
    \Line(176,28)(160,-4)
    \SetWidth{1.5}
    \Line(128,-36)(96,-4)
     \SetWidth{1.5}
    \Line(160,-4)(128,-36)
    \Text(96,4)[lb]{\Large{\Black{$2$}}}
    \Text(144,35)[lb]{\Large{\Black{$3$}}}
    \Text(160,-89)[lb]{\Large{\Black{$1$}}}
     \SetWidth{1.5}
    \Line(160,-4)(144,28)
    \Text(176,35)[lb]{\Large{\Black{$4$}}}
     \SetWidth{3.5}
    \Line(192,-36)(160,-68)
    \SetWidth{1.0}
    \Vertex(192,-36){7.657}
    \SetWidth{1.5}
    \Line(160,-68)(128,-36)
    \SetWidth{1.0}
    \Vertex(160,-68){7.657}
    \Text(192,-26)[lb]{\Large{\Black{$5$}}}
    \Text(104,-39)[lb]{\Large{\Black{$x_2$}}}
    \Text(170,-10)[lb]{\Large{\Black{$x_1$}}}
  \end{picture}
}
}}

\def\ExampleTreeA{\scalebox{0.4}{
\fcolorbox{white}{white}{
  \begin{picture}(74,91) (73,-75)
    \SetWidth{1.0}
    \SetColor{Black}
    \Vertex(96,-53){5.657}
    \Vertex(80,-21){5.657}
    \Vertex(112,-21){5.657}
    \Text(80,-10)[lb]{\Large{\Black{$2$}}}
    \Text(112,-10)[lb]{\Large{\Black{$3$}}}
    \SetWidth{1.5}
    \Line(96,-53)(80,-21)
    \Line(96,-53)(112,-21)
    \Text(95,-74)[lb]{\Large{\Black{$1$}}}
  \end{picture}
}
}}

\def\ExampleTreeAa{\scalebox{0.4}{
\fcolorbox{white}{white}{
  \begin{picture}(74,91) (73,-75)
    \SetWidth{1.0}
    \SetColor{Black}
    \Vertex(96,-53){7.657}
    \Vertex(80,-21){5.657}
    \Vertex(112,-21){7.657}
    \Text(80,-10)[lb]{\Large{\Black{$2$}}}
    \Text(112,-10)[lb]{\Large{\Black{$3$}}}
    \SetWidth{1.5}
    \Line(96,-53)(80,-21)
   \SetWidth{3.5}
    \Line(96,-53)(112,-21)
    \Text(95,-74)[lb]{\Large{\Black{$1$}}}
  \end{picture}
}
}}

\def\ExampleTreeB{\scalebox{0.4}{
\fcolorbox{white}{white}{
  \begin{picture}(106,91) (57,-75)
    \SetWidth{1.0}
    \SetColor{Black}
    \Vertex(96,-53){5.657}
    \Vertex(64,-21){5.657}
    \Vertex(128,-21){5.657}
    \Text(64,-12)[lb]{\Large{\Black{$2$}}}
    \Text(96,-12)[lb]{\Large{\Black{$3$}}}
    \SetWidth{1.5}
    \Line(96,-53)(64,-21)
    \Line(96,-53)(128,-21)
    \Text(95,-74)[lb]{\Large{\Black{$1$}}}
    \SetWidth{1.0}
    \Vertex(96,-21){5.657}
    \SetWidth{1.5}
    \Line(96,-53)(96,-21)
    \Text(128,-12)[lb]{\Large{\Black{$4$}}}
  \end{picture}
}
}}

\def\ExampleTreeBb{\scalebox{0.4}{
\fcolorbox{white}{white}{
  \begin{picture}(106,91) (57,-75)
    \SetWidth{1.0}
    \SetColor{Black}
    \Vertex(96,-53){7.657}
    \Vertex(64,-21){5.657}
    \Vertex(128,-21){7.657}
    \Text(64,-12)[lb]{\Large{\Black{$2$}}}
    \Text(96,-12)[lb]{\Large{\Black{$3$}}}
    \SetWidth{1.5}
    \Line(96,-53)(64,-21)
    \SetWidth{3.5}
    \Line(96,-53)(128,-21)
    \Text(95,-74)[lb]{\Large{\Black{$1$}}}
    \SetWidth{1.0}
    \Vertex(96,-21){5.657}
    \SetWidth{1.5}
    \Line(96,-53)(96,-21)
    \Text(128,-12)[lb]{\Large{\Black{$4$}}}
  \end{picture}
}
}}

\def\ExampleTreeC{\scalebox{0.4}{
\fcolorbox{white}{white}{
  \begin{picture}(90,123) (73,-43)
    \SetWidth{1.0}
    \SetColor{Black}
    \Vertex(96,-21){5.657}
    \Vertex(80,11){5.657}
    \Vertex(112,11){5.657}
    \Text(80,20)[lb]{\Large{\Black{$2$}}}
    \Text(96,52)[lb]{\Large{\Black{$3$}}}
    \SetWidth{1.5}
    \Line(96,-21)(80,11)
    \Line(96,-21)(112,11)
    \Text(95,-42)[lb]{\Large{\Black{$1$}}}
    \Line(112,11)(128,43)
    \Line(112,11)(96,43)
    \SetWidth{1.0}
    \Vertex(128,43){5.657}
    \Vertex(96,43){5.657}
    \Text(128,52)[lb]{\Large{\Black{$4$}}}
    \Text(120,11)[lb]{\Large{\Black{$x_1$}}}
  \end{picture}
}
}}

\def\ExampleTreeCc{\scalebox{0.4}{
\fcolorbox{white}{white}{
  \begin{picture}(90,123) (73,-43)
    \SetWidth{1.0}
    \SetColor{Black}
    \Vertex(96,-21){7.657}
    \Vertex(80,11){5.657}
    \Vertex(112,11){7.657}
    \Text(80,20)[lb]{\Large{\Black{$2$}}}
    \Text(96,52)[lb]{\Large{\Black{$3$}}}
    \SetWidth{1.5}
    \Line(96,-21)(80,11)
        \SetWidth{3.5}
    \Line(96,-21)(112,11)
    \Text(95,-42)[lb]{\Large{\Black{$1$}}}
    \SetWidth{3.5}
    \Line(112,11)(128,43)
      \SetWidth{1.5}
    \Line(112,11)(96,43)
    \SetWidth{1.0}
    \Vertex(128,43){7.657}
    \Vertex(96,43){5.657}
    \Text(128,52)[lb]{\Large{\Black{$4$}}}
    \Text(122,11)[lb]{\Large{\Black{$x_1$}}}
  \end{picture}
}
}}

\def\ExampleTreeD{\scalebox{0.4}{
\fcolorbox{white}{white}{
  \begin{picture}(90,123) (57,-43)
    \SetWidth{1.0}
    \SetColor{Black}
    \Vertex(96,-21){5.657}
    \Vertex(80,11){5.657}
    \Vertex(112,11){5.657}
    \Text(64,52)[lb]{\Large{\Black{$2$}}}
    \Text(96,52)[lb]{\Large{\Black{$3$}}}
    \SetWidth{1.5}
    \Line(96,-21)(80,11)
    \Line(96,-21)(112,11)
    \Text(95,-42)[lb]{\Large{\Black{$1$}}}
    \Line(80,11)(96,43)
    \Line(80,11)(64,43)
    \SetWidth{1.0}
    \Vertex(96,43){5.657}
    \Vertex(64,43){5.657}
    \Text(112,20)[lb]{\Large{\Black{$4$}}}
    \Text(58,11)[lb]{\Large{\Black{$x_1$}}}
  \end{picture}
}
}}

\def\ExampleTreeDd{\scalebox{0.4}{
\fcolorbox{white}{white}{
  \begin{picture}(90,123) (57,-43)
    \SetWidth{1.0}
    \SetColor{Black}
    \Vertex(96,-21){7.657}
    \Vertex(80,11){5.657}
    \Vertex(112,11){7.657}
    \Text(64,52)[lb]{\Large{\Black{$2$}}}
    \Text(96,52)[lb]{\Large{\Black{$3$}}}
    \SetWidth{1.5}
    \Line(96,-21)(80,11)
     \SetWidth{3.5}
    \Line(96,-21)(112,11)
    \Text(95,-42)[lb]{\Large{\Black{$1$}}}
     \SetWidth{1.5}
    \Line(80,11)(96,43)
    \Line(80,11)(64,43)
    \SetWidth{1.0}
    \Vertex(96,43){5.657}
    \Vertex(64,43){5.657}
    \Text(112,20)[lb]{\Large{\Black{$4$}}}
    \Text(58,11)[lb]{\Large{\Black{$x_1$}}}
  \end{picture}
}
}}

\newtheorem{thm}{Theorem}

\newtheorem{cor}[thm]{Corollary}

\newtheorem{lem}[thm]{Lemma}
\newtheorem{prop}[thm]{Proposition}
\newtheorem{defn}{Definition}
\newtheorem{rmk}[thm]{Remark}

\nc{\id}{{\mathrm{id}}}

\def\shuff#1#2{\mathbin{
      \hbox{\vbox{\hbox{\vrule \hskip#2 \vrule height#1 width 0pt}\hrule}\vbox{\hbox{\vrule \hskip#2 \vrule height#1 width 0pt\vrule }\hrule}}}}
\def\shuffl{{\mathchoice{\shuff{5pt}{3.5pt}}{\shuff{5pt}{3.5pt}}{\shuff{3pt}{2.6pt}}{\shuff{3pt}{2.6pt}}}}
\def\shuffle{{\, \shuffl \,}}

\begin{document}
%%%%%%%%%%%%%%%%%%%%%%%%%%%%%%%%%%%%%%%%%%%%%

\title[The combinatorics of Green's functions in planar field theories]{The combinatorics of Green's functions\\ in planar field theories}{\let\thefootnote\relax\footnote{{{These notes are based on talks given at the workshop ``Dyson--Schwinger Equations in Modern Mathematics \& Physics", European Centre for Theoretical Studies in Nuclear Physics and Related Areas (ECT*), Trento, Italy, September 22-26, 2014.}}}}
\vspace{1cm}

\author{Kurusch Ebrahimi-Fard}
\address{{ICMAT,
		C/Nicol\'as Cabrera, no.~13-15, 28049 Madrid, Spain}.
		{\tiny{On leave from UHA, Mulhouse, France.}}}
         \email{kurusch@icmat.es, kurusch.ebrahimi-fard@uha.fr}         
         \urladdr{www.icmat.es/kurusch}

\author{Fr\'ed\'eric Patras}
\address{Univ.~de Nice,
			Labo.~J.-A.~Dieudonn\'e,
         		UMR 7351, CNRS,
         		Parc Valrose,
         		06108 Nice Cedex 02, France.}
\email{patras@math.unice.fr}
\urladdr{www-math.unice.fr/$\sim$patras}

%%%%%%%%%%%%%%%%%%%%%%%%%%%%%%%%%%%%%%%
\date{\today}
%%%%%%%%%%%%%%%%%%%%%%%%%%%%%%%%%%%%%%%

\begin{abstract}
The aim of this work is to outline in some detail the use of combinatorial algebra in planar quantum field theory. Particular emphasis is given to the relations between the different types of planar Green's functions. The key object is a Hopf algebra which is naturally defined on non-commuting sources, and the fact that its genuine unshuffle coproduct splits into left- and right unshuffle half-coproduts. The latter give rise to the notion of unshuffle bialgebra. This setting allows to describe the relation between planar full and connected Green's functions by solving a simple linear fixed point equation. A modification of this linear fixed point equation gives rise to the relation between planar connected and one-particle irreducible Green's functions. The graphical calculus that arises from this approach also leads to a new understanding of functional calculus in planar QFT, whose rules for differentiation with respect to non-commuting sources can be translated into the language of growth operations on planar rooted trees. We also include a brief outline of our approach in the framework of non-planar theories.      
\end{abstract}

%%%%%%%%%%%%%%%%%%%%%%%%%%%%%%%%%%%%%%%

\maketitle
\tableofcontents

%%%%%%%%%%%%%%%%%%%%%%%%%%%%%%%%%%%%%%%

\section{Introduction}
\label{sect:intro}

The intent of this work is to present a purely Hopf algebraic description of the well-known relations between full, connected and one-particle irreducible (1PI) Green's functions in planar quantum field theory (QFT). 
Since this approach may as well shed new light on the classical setting, we also include a short outline addressing the framework of non-planar theories. Although emphasis is put on planar theories.  

\smallskip

In the early 1980's Cvitanovic et al.~\cite{cvitanovic1,cvitanovic2} proposed an approach to planar quantum field theories, which was largely motivated as a way of properly encoding the behaviour of the planar sector of quantum chromodynamics (QCD). An interesting feature of planar theories is the manner in which the calculus of symmetry factors changes (and simplifies) with respect to classical theories. In fact, planarity is reflected in a strictly non-commutative nature of the theory, which results in a rather substantial deviation from the classical description of the relations between different types of Green's functions. Indeed, Cvitanovic et.~al.~observed that the functional relation between the generating functionals of the planar full and planar connected Green's functions is encoded by a fixed point equation, which is solved by those generating functionals. This fixed point equation replaces the common exponential map relating the generating functionals of the classical, i.e., non-planar full and connected Green's functions. The exponential relation between those classical generating functionals is analog to the moments-cumulants relations in classical probability theory \cite{speed}. Several years later it was realised that the description in \cite{cvitanovic1,cvitanovic2} of the relations between planar Green's functions is closely related to Speicher's approach to the relations between free moments and free cumulants \cite{biane,nicaspeicher,speicher,novaksniady} in the context of Voiculescu's theory of free probability \cite{voiculescu1,voiculescu2}. The resulting link between free probability and planar QFT has been explored in several works, see e.g. \cite{douglas,gopagross}.

The approach to the relations between planar Green's functions presented in these notes is based on our recent work on the algebraic and combinatorial structures underlying the relations between free and classical moments and cumulants in probability theory \cite{EP1,EP2}. In those references, it was shown that these relations, both classical and free, can be understood algebraically in terms of solutions of linear fixed point equations in (co)commutative, corresponding to classical probability, and non-(co)commutative, corresponding to free probability, shuffle Hopf algebras. It turns out that these linear fixed point equations have proper exponential solutions. In the classical case this exponential solution coincides with the standard exponential relating classical moments and cumulants. In the non-classical, i.e., planar setting the relation between free cumulants and moments is displayed by an exponential as well, which is defined with respect to a non-commutative product. The difference between these two exponentials is analogous to the difference between exponential solutions of scalar- and matrix-valued non-autonomous linear differential equations. We propose here a similar approach to the -- Hopf -- algebraic understanding of the relations between Green's functions in planar QFT.

\medskip
 
We start by recalling the classical relation between full and connected Green's functions \cite{itzyksonzuber}. The generating series of the full, or complete, Green's function is denoted 
$$
	Z(j):= 1 + \sum_{k>0} \frac{1}{k!} Z^{(k)}_{j_1 \cdots j_k} j_1 \cdots j_k, 
$$
with $j$ denoting the set of  commuting external sources $\{j_i\}_{i>0}$. We follow notations and conventions used in \cite{cvitanovic1,cvitanovic2}, where indices $k,l, m, n, \cdots$ may represent discrete as well as continuous variables, and repeated indices indicate summations and integration over discrete respectively continuous variables that characterise the actual Green's functions. The generating series for connected Green's functions is denoted
$$
	W(j):= \sum_{m>0} \frac{1}{m!} W^{(m)}_{j_1 \cdots j_m} j_1 \cdots j_m. 
$$
In both cases the Green's functions are symmetric in the $j_i$s, and follow from functional derivations with respect to the external sources, e.g., the complete Green's function of order $n$, $Z^{(n)}_{j_{i_1} \cdots j_{i_n}}$, is given by
$$
	\left. \frac{\partial^n}{\partial_{j_{i_1}} \cdots \partial_{j_{i_n}}}\right |_{j=0} Z(j) 
		= Z^{(n)}_{j_{i_1} \cdots j_{i_n}} 
$$  
It turns out that the above generating series are related through the exponential map
$$
	Z(j)=\exp\big( W(j) \big).
$$
Taking functional derivations, we find up to order four:
\allowdisplaybreaks{
\begin{eqnarray}
	Z^{(0)} &=& 1  											\nonumber \\
	Z_{j_1}^{(1)} &=& W_{j_1}^{(1)}  							\nonumber \\ 
	Z_{j_1j_2}^{(2)} &=& W_{j_1}^{(1)}W_{j_2}^{(1)} + W_{j_1j_2}^{(2)}
				  = W_{j_1}^{(1)}Z_{j_2}^{(1)} + W_{j_1j_2}^{(2)}Z^{(0)}	\nonumber \\ 
	Z_{j_1j_2j_3}^{(3)} &=& W^{(1)}_{j_1}W^{(1)}_{j_2}W^{(1)}_{j_3} 
							+ W^{(1)}_{j_1}W^{(2)}_{j_2j_3} 
							+ W^{(1)}_{j_2}W^{(2)}_{j_1j_3}
							+ W^{(1)}_{j_3}W^{(2)}_{j_1j_2} 
							+ W_{j_1j_2j_3}^{(3)} 			\nonumber\\
				      &=& W^{(1)}_{j_1}Z_{j_2j_3}^{(2)} 
				      		+ W_{j_1j_2}^{(2)}Z_{j_3}^{(1)}
						+ W_{j_1j_3}^{(2)}Z_{j_2}^{(1)}
						+ W_{j_1j_2j_3}^{(3)}Z^{(0)}			\nonumber\\
	Z_{j_1j_2j_3j_4}^{(4)} &=&  W_{j_1}^{(1)}Z_{j_2j_3j_4}^{(3)}
					     	+ W_{j_1j_2}^{(2)}Z_{j_3j_4}^{(2)}	
					     	+ W_{j_1j_3}^{(2)}Z_{j_2j_4}^{(2)}	
					     	+ W_{j_1j_4}^{(2)}Z_{j_2j_3}^{(2)} 		\label{nonplanar1} \\
					& & 	+ W_{j_1j_2j_3}^{(3)}Z_{j_4}^{(1)}	
					      	+ W_{j_1j_2j_4}^{(3)}Z_{j_3}^{(1)}	
					     	+ W_{j_1j_3j_4}^{(3)}Z_{j_2}^{(1)}	
					     	+ W_{j_1j_2j_3j_4}^{(4)}Z^{(0)} . 		\nonumber 
\end{eqnarray}}
Abstractly, these polynomial expressions for full Green's functions given in terms of connected ones are a multivariate generalization of the classical Bell polynomials relating, among others, moments and cumulants in classical probability \cite{ELundMan}. The recursive structure featured in the order four case in place of the complete expansion displays the full Green's function in terms of lower order connected ones, and will be explained further below in the context of the tensor algebra approach. The third class of Green's functions are the one-particle irreducible (1PI) ones. They are related to the connected Green's functions in a somewhat more involved manner. We will come back to this with more details in the planar setting. 

\smallskip 

The generating functionals for planar full and planar connected Green's functions are given by
$$
	\mathrm{Z}[j]:= 1 + \sum_{k>0} \mathrm{Z}^{(k)}_{j_1 \cdots j_k} j_1 \cdots j_k 
	\quad\ {\mathrm{resp.}} \quad\ 
	\mathrm{W}[j]:= \sum_{m>0} \mathrm{W}^{(m)}_{j_1 \cdots j_m} j_1 \cdots j_m.
$$
Here the external sources $j=\{j_i\}_{i>0}$ are strictly non-commutative, so that neither $\mathrm{Z}^{(k)}_{j_1 \cdots j_k}$ nor $\mathrm{W}^{(m)}_{j_1 \cdots j_m}$ are symmetric functions with respect to the sources any more. The corresponding non-commutative functional calculus is explained in detail in reference \cite{cvitanovic2}. Cvitanovic noted in \cite{cvitanovic1} that the planar nature of the problem yields a different functional relation between the two planar generating functionals $\mathrm{Z}[j]$ and $\mathrm{W}[j]$, which is given in terms of the fixed point equation
\begin{equation}
\label{cvitanovich}
	\mathrm{Z}[j]:= 1 + \mathrm{W}[j\mathrm{Z}[j]].
\end{equation}
Due to the non-commutative nature of the external sources, a different but equivalent form is given by $\mathrm{Z}[j]:= 1 + \mathrm{W}[\mathrm{Z}[j]j]$. We will work foremost with equality \eqref{cvitanovich}, which yields up to order four:
\allowdisplaybreaks{
\begin{eqnarray}
	\mathrm{Z}^{(0)} &=& 1  																	\nonumber \\
	\mathrm{Z}_{j_1}^{(1)} &=& \mathrm{W}_{j_1}^{(1)} 												\nonumber \\ 
	\mathrm{Z}_{j_1j_2}^{(2)} &=&\mathrm{W}_{j_1}^{(1)}\mathrm{Z}_{j_2}^{(1)} 
							+ \mathrm{W}_{j_1j_2}^{(2)}\mathrm{Z}^{(0)}
						= \mathrm{W}_{j_1}^{(1)}\mathrm{W}_{j_2}^{(1)} 
							+ \mathrm{W}_{j_1j_2}^{(2)}	\nonumber \\ 
	\mathrm{Z}_{j_1j_2j_3}^{(3)} &=& 	 {\mathrm{W}_{j_1}^{(1)}}\mathrm{Z}_{j_2j_3}^{(2)}
						     		+  {\mathrm{W}_{j_1j_2}^{(2)}}\mathrm{Z}_{j_3}^{(1)}
								+  {\mathrm{W}_{j_1j_3}^{(2)}}\mathrm{Z}_{j_2}^{(1)}
								+  \mathrm{W}_{j_1j_2j_3}^{(3)}	\mathrm{Z}^{(0)} 												\nonumber \\
						   &=& {\mathrm{W}_{j_1}^{(1)}}{\mathrm{W}_{j_2}^{(1)}}{\mathrm{W}_{j_3}^{(1)}} 
							+ \mathrm{W}_{j_1}^{(1)} \mathrm{W}_{j_2j_3}^{(2)}
							+ \mathrm{W}_{j_1j_2}^{(2)} \mathrm{W}_{j_3}^{(1)} 
							+ \mathrm{W}_{j_1j_3}^{(2)} \mathrm{W}_{j_2}^{(1)}  
							+ \mathrm{W}_{j_1j_2j_3}^{(3)}	
\nonumber \\
	\mathrm{Z}_{j_1j_2j_3j_4}^{(4)} &=& {\mathrm{W}_{j_1}^{(1)}}\mathrm{Z}_{j_2j_3j_4}^{(3)} 
								+ {\mathrm{W}_{j_1j_2}^{(2)}}\mathrm{Z}_{j_3j_4}^{(2)} 
								+ {\mathrm{W}_{j_1j_3}^{(2)}}\mathrm{Z}_{j_2}^{(1)}\mathrm{Z}_{j_4}^{(1)} 
								+ {\mathrm{W}_{j_1j_4}^{(2)}}\mathrm{Z}_{j_2j_3}^{(2)} 				\label{planar1}\\
							 & &  + {\mathrm{W}_{j_1j_2j_3}^{(3)}}\mathrm{Z}_{j_4}^{(1)} 
							 	+ {\mathrm{W}_{j_1j_2j_4}^{(3)}}\mathrm{Z}_{j_3}^{(1)} 
								+ {\mathrm{W}_{j_1j_3j_4}^{(3)}}\mathrm{Z}_{j_2}^{(1)}  
					 			+ \mathrm{W}_{j_1j_2j_3j_4}^{(4)}\mathrm{Z}^{(0)}. 					\nonumber 
\end{eqnarray}}
Note that the difference with the analogous polynomials in the non-planar case starts at order four, i.e., compare the terms $W_{j_1j_3}^{(2)}Z_{j_2j_4}^{(2)}$ and ${\mathrm{W}_{j_1j_3}^{(2)}}\mathrm{Z}_{j_2}^{(1)}\mathrm{Z}_{j_4}^{(1)}$ in lines \eqref{nonplanar1} and  \eqref{planar1}, respectively. A precise understanding of the combinatorial nature of the recursive structure which is on display here will be given further below in terms of the double tensor Hopf algebra of non-commuting sources, and its natural non-cocommuting unshuffle coproduct. The combinatorial description of the relations between planar connected and planar 1PI Green's functions, which derives from this Hopf algebra approach will be explained in section \ref{sect:1PIplGreen}.

\smallskip

Our approach may be summarised by saying that it captures in a purely algebraic way the functional calculus employed to describe the relations between planar Green's functions. More precisely, we show that the double tensor algebra $\bar T(T(J))$ and the tensor algebra $\bar T(J)$ together with their non-cocommutative respectively cocommutative unshuffle coproducts provide the appropriate Hopf algebraic setting to algebraize the relations between the generating series of planar respectively non-planar full and connected Green's functions. To this end, the planar and non-planar generating functionals are considered as linear maps over the two aforementioned tensor algebras, $\bar T(T(J))$ respectively $\bar T(J)$. In both cases, the usual functional relations between generating functionals are described in Hopf algebraic terms through linear fixed point equations of the same type, which are solved by the corresponding linear maps. In the planar case the linear map $\tau_\mathrm{Z}$ that represents $\mathrm{Z}[j]$, is a multiplicative unital map, i.e., a Hopf algebra character over $\bar T(T(J))$, whereas $\tau_\mathrm{W}$, representing $\mathrm{W}[j]$, is an infinitesimal character. They are related through the fundamental linear fixed point equation defined in terms of the left half-shuffle product
$$
		\tau_\mathrm{Z} = \varepsilon + \tau_\mathrm{W} \prec \tau_\mathrm{Z}.
$$
Its solution is given by the exponential map 
$$
	\tau_\mathrm{Z}=\exp^{\star}\big(\Omega'(\tau_\mathrm{W})\big).	
$$
The map $\Omega'$ reflects the non-commutative nature of the product $\star$ used to define the exponential. In the non-planar case we find that the linear maps $\tau_Z ,\tau_W \in  \bar T^*(J)$ solve the analogue left half-shuffle relation, this time, however, defined over $ \bar T^*(J)$
$$
	\tau_Z = \varepsilon + \tau_W \prec \tau_Z \quad \mathrm{and} \quad \tau_Z=\exp^{\shuffle}\big(\tau_W\big).
$$ 
In the non-planar case the exponential map is defined with respect to the shuffle product $\shuffle$, and the non-appearance of the map $\Omega'$ is a consequence of the commutative nature of the shuffle algebra. We see that in both the planar and non-planar setting the -- generating series of -- full and connected Green's functions are related through an exponential map, defined with respect to a non-commutative respectively commutative product. This picture allows us, in particular, to formulate the description of planar connected Green's functions in terms of planar 1PI Green's functions, using the language of noncrossing partitions and planar de- or increasing trees \cite{drmota}. In this framework, functional derivations have a diagrammatical description in terms of right-grafting single edges to planar rooted trees.

\medskip 

We briefly outline the organization of the article. In the next section we introduce various notions, ranging from classical Hopf algebras to the lesser known unshuffle bialgebras. We explain several key properties of the latter and illustrate these notions on the tensor and double tensor algebras generated by families of external sources. The second part investigates the relations between full and connected Green's functions. We also introduce noncrossing Green's functions and show how the free probability analysis of the relations between free moments and free cumulants can be transferred to planar QFT using the unshuffle calculus underlying our approach. The third part investigates the links between planar connected and planar 1PI Green's functions. This leads to a graphical interpretation of the calculus of derivatives with respect to external sources in the functional approach to planar QFT which is the object of the last section.

\medskip 

In the following $\mathbb{K}$ denotes a ground field of characteristic zero, e.g., $\mathbb{K}=\mathbb{C}$ or $\mathbb{R}$. We also assume any $\mathbb{K}$-algebra $A$ to be associative and unital, if not stated otherwise. The unit in $A$ is denoted $\un_A$. Identity morphisms are written $\id$.

\vspace{0.4cm}

\noindent {\bf{Acknowledgements}}: First and foremost we would like to thank the organisers of the workshop ``Dyson--Schwinger Equations in Modern Mathematics \& Physics", ECT*, Trento, Italy, September 22-26, 2014, for inviting us to present our work at this meeting. The authors are indebted to T.~Krajewski  who pointed out to them the relevance of their work on free probabilities for planar QCD during the workshop in Trento. The first author is supported by a Ram\'on y Cajal research grant from the Spanish government. The second author acknowledges support from the grant ANR-12-BS01-0017, Combinatoire Alg\'ebrique, R\'esurgence, Moules et Applications and from the ESI Vienna, where this work started. Support by the CNRS GDR Renormalisation is also acknowledged.

%%%%%%%%%%%%%%%%%%%%%%%%%%%%%%%%%%%%%%%

\section{Green's functions and connected graded Hopf algebras}
\label{sect:GHA}

%%%%%%%%%%%%%%%%%%%%%%%%%%%%%%%%%%%%%%%

\subsection{Hopf algebras}
\label{sect:Hopf}

We start this section by recalling a few basic facts on Hopf algebras, which shall also serve to fix notation. For details the reader is referred to \cite{cartier1,reutenauer,sweedler}. However, a remark is in order. Our work focuses exclusively on a particular connected graded non-commutative non-cocommutative Hopf algebra with some extra structure defined on the double tensor algebra over a set of external sources as the basic object of interest in our approach to planar Green's functions. These objects are related to (non-commutative) Fock spaces and can be thought of as a kind of generalization thereof. In particular, both their structure and combinatorics are different from the ones involved in the modelling of the Bogoliubov recursion and the BPHZ renormalization process by means of Hopf algebras of Feynman diagrams and Rota--Baxter algebra structures. We refer the reader interested in learning about Connes' and Kreimer's Hopf algebraic approach to renormalization to the original works \cite{ck0,ck1,ck2}. See also \cite{EGP,GBFV,manchon2}.

\smallskip

A coalgebra over $\mathbb{K}$ consist of a $\mathbb{K}$-vector space $C$ and two maps, the coproduct $\Delta: C \to C \otimes C$, which is coassociative
\begin{equation}
\label{coasso}
	(\Delta \otimes \id)\circ \Delta=(\id \otimes \Delta)\circ \Delta,
\end{equation}
and the counit $\varepsilon: C \to \mathbb{K}$, such that $(\varepsilon\otimes \id)\circ \Delta = (\id\otimes \varepsilon)\circ\Delta= \id$. A coalgebra is cocommutative if $\Delta=\tau \circ \Delta$, where $\tau$ is the switch map $\tau(x\otimes y):=y\otimes x$. Iterated coproducts are denoted as follows $\Delta^{0}:=\id$ and $\Delta^{n}:C\to C^{\otimes n+1}$
$$
	\Delta^{n}:=(\id \otimes \Delta^{n-1}) \circ \Delta.
$$
A  $\mathbb{K}$-bialgebra is a  $\mathbb{K}$-vector space $B$, which is both a $\mathbb{K}$-algebra and a $\mathbb{K}$-coalgebra together with certain compatibility relations, such as, for instance, both the algebra product, $m:B \otimes B \to B $, and unit map, $e: \mathbb{K} \to B$, are coalgebra morphisms \cite{sweedler}. The unit of $B$ is denoted by $\un = e(1)$. A bialgebra is called graded if  there exist $\mathbb{K}$-vector spaces $B_{m}$ such that $B = \bigoplus_{n \geq 0} B_{n}$,  and $m(B_{p} \otimes B_{q}) \subseteq B_{p+q}$ as well as $\Delta(B_{n}) \subseteq \bigoplus_{p+q=n} B_{p}\otimes B_{q}.$ Elements $x \in B_{n}$ are said to be of degree $|x|=n$. Define $B^+=\bigoplus_{n > 0} B_{n}$. A graded bialgebra $B$ is called connected if the degree zero component is one dimensional $B_{0} = \mathbb{K}\un$. In this case,  the coproduct for an element $x \in B^+$ of degree $|x|=n$ is of the form
\begin{equation*}
	\Delta(x) = x \otimes \un + \un \otimes x + \Delta'(x) \in \bigoplus_{k+l=n} B_{k} \otimes B_{l},
\end{equation*}
where $\Delta'(x) := \Delta(x) -  x \otimes \un - \un \otimes x \in B^+ \otimes B^+$ is the reduced coproduct. Furthermore, $B^+=Ker(\varepsilon)$ is the augmentation ideal of $B$. By definition an element $x \in B$ is called primitive if $\Delta'(x) = 0.$

Suppose that $A$ is an algebra with product $m_A$ and unit map $e_A(1)=\un_A$, e.g., $A=\mathbb{K}$ or $A=B$, where $B$ is a bialgebra. The vector space  $L(B, A)$ of linear maps from $B$ to $A$ together with the convolution product 
\begin{equation}
\label{convolutionProd}
	\Phi \star \Psi := m_{A} \circ (\Phi \otimes \Psi) \circ \Delta : B \to A,
\end{equation}
$\Phi,\Psi \in L(B,A)$, is an associative algebra with unit $\iota := e_{A} \circ \varepsilon$.

A $\mathbb{K}$-Hopf algebra is a $\mathbb{K}$-bialgebra $H$ equiped with a $\mathbb{K}$-linear map $S: H \to H$ called the antipode  which (for the kind of Hopf algebras we are interested in) is characterised as the inverse of the identity map $\id \in L(H,H)$ with respect to the convolution product, that is, $\id \star S = S \star \id = e \circ \varepsilon$, where
\begin{equation*}
	S \star \id = m \circ (S \otimes \id) \circ \Delta.
\end{equation*}
It is a well-known fact that any connected graded bialgebra is automatically a connected graded Hopf algebra, see e.g. ~\cite{cartier1,sweedler}. 

Let $H=\bigoplus_{n \ge 0} H_{n}$ be a connected graded Hopf algebra. Suppose $A$ is a commutative unital algebra. The subset $g_0 \subset  L(H, A)$ of linear maps, $\alpha$, that send the unit to zero, $\alpha(\un)=0$, forms a Lie algebra in $L(H, A)$. The exponential $ \exp^\star(\alpha) = \sum_{j\ge 0} \frac{1}{j!}\alpha^{\star j}$ defines a bijection from $g_0$ onto the group $G_0 = \iota + g_0$ of linear maps, $\gamma$, that send the unit of $H$ to the algebra unit, $\gamma(\un)=\un_{A}$. The neutral element is $\iota:=e_{A} \circ \epsilon$, given by $\iota(\un)=\un_{ A}$ and $\iota(x) = 0$ for $x \in H^+$. An infinitesimal character with values in $A$ is a linear map $\xi \in g_0$ such that for $x, y \in  H^+$, $\xi(m(x \otimes y)) = 0$. The linear space of infinitesimal characters is a Lie subalgebra of $g_0$ denoted $g_{A}$. An element $\Phi$ in $ G_0$ is called a character if, for $x,y \in  H$, $\Phi(m(x \otimes y)) = m_A(\Phi(x) \otimes \Phi(y))$. The set of characters is denoted by $G_{A} \subset G_0$. It forms a pro-unipotent group for the convolution product with (pro-nilpotent) Lie algebra $g_{ A}$. The exponential map $\exp^{\star}$ restricts to a bijection between $g_{ A}$ and $G_{ A}$. . The inverse of $\Phi \in G_{A}$ is given by composition with the Hopf algebra antipode $S$, i.e., $\Phi^{\star -1} = \Phi \circ S$. See \cite{EGP,GBFV} for more results and details.

%%%%%%%%%%%%%%%%%%%%%%%%%%%%%%%%%%%%%%%

\subsection{Tensor Hopf algebras}
\label{ssect:tensorHA}

Next we present briefly two examples of connected graded Hopf algebra, which play a key role from our point of view, i.e., the tensor and double tensor algebra over an arbitrary set. Let $J:=\{j_1,j_2, j_3,\ldots\}$ be a set (also called an alphabet). The set of words $j_{i_1}\cdots j_{i_l}$ is denoted $J^*$ and the linear span of $J^*$ is denoted $\mathcal J$. We start by defining $T(J):=\oplus_{n > 0} {\mathcal J}^{\otimes n}$ to be the nonunital tensor algebra over $J$. The full tensor algebra is denoted $\bar T(J):=\oplus_{n \ge 0} {\mathcal J}^{\otimes n}$, with ${\mathcal J}^{\otimes 0}=\mathbb{K}\un$. Elements in $T(J)$ are written as linear combinations of words $j_{i_1}\cdots j_{i_l} \in T(J)$. The natural degree of a word $w = j_{i_1}\cdots j_{i_n}$ is its lenght $n=:|w|$, and we write $w \in T_n(J)$. The space $T(J)$ is a graded algebra with the natural non-commutative product defined by concatenating words $w=j_{i_1}\cdots j_{i_n} \in T_n(J)$ and $w'=j_{k_1}\cdots j_{k_m} \in T_m(J)$
$$
	w \cdot w':= j_{i_1}\cdots j_{i_n} \cdot j_{k_1}\cdots j_{k_m}
			= j_{i_1}\cdots j_{i_n}j_{k_1}\cdots j_{k_m}  \in T_{n+m}(J).
$$ 
The tensor algebra $\bar T(J)$ becomes a unital connected non-commutative but cocommutative Hopf algebra if it is equipped with the unshuffle coproduct, which is defined by declaring the elements in degree one $J  \hookrightarrow \bar T(J)$ to be primitive, i.e., $\Delta^{\!\!\shuffle}(j_{i}):= j_i \otimes \un + \un \otimes j_i$. This definition extends multiplicatively to all of $\bar T(J)$, e.g.
$$
	\Delta^{\!\!\shuffle}(j_{i}j_{l}) = \Delta^{\!\!\shuffle}(j_{i})\Delta^{\!\!\shuffle}(j_{l}) 
	= j_{i}j_{l} \otimes \un + \un \otimes j_{i}j_{l} + j_{i} \otimes j_{l} + j_{l} \otimes j_{i}.
$$ 
The general form is given as follows
\begin{align}
\label{unshuffle}
	\Delta^{\!\!\shuffle}(w)=\sum_{u, v \in J^*} \langle u \shuffle v, w \rangle u \otimes v.
\end{align}
Here the coefficient $\langle u , v \rangle : = 1$ if $u=v$, and zero else. The product $\shuffle : \bar T(J) \otimes \bar T(J) \to \bar T(J)$ displaying in \eqref{unshuffle} is the shuffle product of words \cite{reutenauer}, which is defined iteratively for any word $w$ by $\un \shuffle w := w \shuffle \un := w$, and 
\begin{align}
\label{shuffleProd}
	j_i v \shuffle j_ku := j_i(v \shuffle j_ku) + j_k(j_i v \shuffle u),
\end{align}
for words $u,v \in J^*$. In low degrees,
\begin{align*}
	j_i \shuffle j_k &= j_i j_k + j_kj_i \\
	j_i \shuffle j_kj_l &=  j_i j_kj_l  + j_k(j_i \shuffle j_l ) = j_i j_kj_l + j_k j_i j_l + j_kj_lj_i\\
	j_ij_n \shuffle j_kj_m &= j_i(j_n  \shuffle j_kj_m)  + j_k( j_i j_n  \shuffle j_m).
\end{align*}
The shuffle product of words is associative and commutative. 

Later we will see that \eqref{shuffleProd} is just the commutative version of a more general shuffle product. In fact, we will be interested mainly in the non-commutative case (the one meaningful for planar QFT). The name ``shuffle algebra'' refers to general, possibly non-commutative, shuffle algebras, and we refer explicitly to ``commutative shuffle algebras'' in the commutative case.

\smallskip

Next we augment the complexity of our word algebra $T(J)$ by defining the double tensor algebra $T(T(J)) := \oplus_{n > 0} T(J)^{\otimes n}$, and use the bar-notation to denote elements $w_1 | \cdots | w_n \in T(T(J))$, $w_i \in T(J)$, $i=1,\ldots,n$. The algebra $T(T(J))$ is equipped with the concatenation product. For $a = w_1 | \cdots | w_n$ and $b = w_1' | \cdots | w_m'$ we denote the concatenation product in $T(T(J))$ by $a|b$, that is, $a|b := w_1 | \cdots | w_n | w_1' | \cdots | w_m'$. This algebra is multi-graded, i.e., $T(T(J))_{n_1,\ldots ,n_k}:=T_{n_1}(J) \otimes \cdots \otimes T_{n_k}(J)$, as well as graded
$$
	T(T(J))_n := \bigoplus\limits_{n_1 + \cdots + n_k=n}T(T(J))_{n_1,\ldots ,n_k}.
$$ 
Similar observations hold for the unital case, $\bar T(T(J))=\oplus_{n \ge 0} T(J)^{\otimes n}$, and we will identify without further comments a bar symbol such as $w_1|\un|w_2$ with $w_1|w_2$ (formally, using the canonical map from $\bar T(\bar T(J))$ to $\bar T(T(J))$). 

The double tensor algebra becomes a Hopf algebra by defining another unshuffle-type coproduct, which is a refinment of the unshuffling coproduct in \eqref{unshuffle}. Given two (canonically ordered) subsets $S \subseteq U$ of the set of integers $\mathbb{N}^\ast$, we call connected component of $S$ relative to $U$ a maximal sequence $s_1, \ldots , s_n$ in $S$ such that there are no $ 1\leq i < n$ and $u \in U$, such that $s_i < u < s_{i+1}$. In particular, a connected component of $S$ in $\mathbb{N}^\ast$ is simply a maximal sequence of successive elements $s,s+1,\ldots ,s+n$ in $S$.

Consider a word $j_{i_1}\cdots j_{i_n} \in T(J)$. For $S:=\{s_1,\ldots, s_p\} \subseteq [n]$, we set $j_S:= j_{i_{s_1}} \cdots j_{i_{s_p}}$ (resp. $j_\emptyset:=1$). Denoting $J_1,\ldots,J_k $ the connected components of $[n] - S$, we also set $j_{J^S_{[n]}}:= j_{J_1} | \cdots | j_{J_k}$. More generally, for $S \subseteq U \subseteq [n]$, set  $j_{J^S_U}:= j_{J_1} | \cdots | j_{J_k}$, where the ${J_i}$ are now the connected components of $U-S$ in $U$.

\begin{defn} \label{def:coproduct}
The map $\delta : T(J) \to \bar T(J) \otimes  \bar T(T(J))$ is defined by
\begin{equation}
\label{HA}
	\delta(j_{i_1}\cdots j_{i_n}) := \sum_{S \subseteq [n]} j_S \otimes  j_{J_1} | \cdots | j_{J_k}
					   =\sum_{S \subseteq [n]} j_S \otimes j_{J^S_{[n]}}.
\end{equation} 
The coproduct is then extended multiplicatively to all of $\bar T(T(J))$
$$
	\delta(w_1 | \cdots | w_m) := \delta(w_1) \cdots \delta(w_m),
$$
with $\delta(\un):= \un \otimes \un$.
\end{defn}

\begin{thm} \label{thm:HA} \cite{EP1,EP2}
The graded algebra $\bar T(T(J))$ equipped with the coproduct \eqref{HA} is a connected graded non-commutative and non-cocommutative Hopf algebra. 
\end{thm}

%%%%%%%%%%%%%%%%%%%%%%%%%%%%%%%%%%%%%%%

\subsection{Splitting of unshuffle coproducts}
\label{ssect:unshuffle}

As we already indicated, both coproducts, \eqref{unshuffle} and \eqref{HA}, are considered to be of unshuffle-type. In the latter case we keep track of the subsets of letters that have been extracted from a word $w$ in $T(J)$ by ``filling-in the holes" by bars in the respective places. We will show in the rest of this paper, that this simple operation, that is, going from $T(J)$ to $T(T(J))$, is enough to understand the different natures of the relations between full, connected and 1PI Green's functions in the context of planar field theories. In fact, we will see that these relations, whether in the planar or non-planar case, are encoded by a particular linear fixed point equation defined on either $T(J)$ or $T(T(J))$, and which is derived from a rather natural splitting of the coproducts \eqref{unshuffle} and \eqref{HA}.

Indeed, the unshuffle-type coproducts, \eqref{unshuffle} and \eqref{HA}, share a particular coalgebraic property, i.e., both can be split into a sum of  left and right unshuffle half-coproducts. We will first introduce this splitting for the coproduct \eqref{unshuffle}, where it is easily defined by realising, that \eqref{unshuffle} can be written in more set-theoretic terms
\begin{equation}
\label{unshuffle1}
	\Delta^{\!\!\shuffle}(j_{i_1}\cdots j_{i_n}) = \sum_{I \subseteq [n]} j_I \otimes j_{[n] -I}.
\end{equation}
As before, for subsets $S=\{s_1,\ldots ,s_k\}\subset [n]$, $j_S$ stands for the word $j_{i_{s_1}}\cdots j_{i_{s_k}}$. 
Now we define the left unshuffle half-coproduct $\Delta^{\!\!\shuffle}_\prec: T(J) \to T(J) \otimes \bar T(J)$
\begin{equation}
\label{deshufleft}
	\Delta^{\!\!\shuffle}_\prec(j_{i_1} \cdots j_{i_n})
		=\sum_{I \subseteq [n] \atop 1 \in I} j_I \otimes j_{[n] -I}.
\end{equation}
The right unshuffle half-coproduct is defined through $\Delta^{\!\!\shuffle}_\succ:=\Delta^{\!\!\shuffle}-\Delta^{\!\!\shuffle}_\prec$. In explicit form it writes
\begin{equation}
\label{deshufright}
	\Delta^{\!\!\shuffle}_\succ(j_{i_1} \cdots j_{i_n})
		=\sum_{I \subset [n] \atop 1 \notin I} j_I \otimes j_{[n] -I},
\end{equation}
so that the splitting of the unshuffle coproduct \eqref{unshuffle} in terms of these two operations writes
\begin{equation}
\label{deshufsplit}
	\Delta^{\!\!\shuffle} = \Delta^{\!\!\shuffle}_\succ + \Delta^{\!\!\shuffle}_\prec.
\end{equation}
This gives rise to an unshuffle bialgebra structure on $\bar T(J)$, whose formal definition will be given further below. The fine structure of this mathematical notion is studied in \cite{foipat}. Note that in the following we shall use both appellations, (left) right unshuffle half-coproduct or (left) right half-unshuffle. 

Before we move on to the double tensor algebra, $T(T(J))$, we shall give a few examples.
\allowdisplaybreaks{
\begin{align*}
	\Delta^{\!\!\shuffle}_\prec(j_{1}) &= j_{1} \otimes \un \\
	\Delta^{\!\!\shuffle}_\prec(j_{1} j_{2}) &= j_{1} \otimes j_{2} + j_{1} j_{2} \otimes \un \\
	\Delta^{\!\!\shuffle}_\prec(j_{1} j_{2}j_{3}) &= j_{1} \otimes j_{2} j_{3} 
										+ j_{1} j_{2} \otimes  j_{3} 
										+ j_{1} j_{3} \otimes  j_{2}
										+ j_{1} j_{2} j_{3} \otimes \un\\
	\Delta^{\!\!\shuffle}_\prec(j_{1} j_{2}j_{3}j_{4}) 	&= j_{1} \otimes j_{2} j_{3}j_{4} 
											+ j_{1} j_{2} \otimes  j_{3}j_{4} 
											+ j_{1} j_{3} \otimes  j_{2}j_{4}		
											+ j_{1} j_{4} \otimes  j_{2}j_{3}\\
										&\quad	+ j_{1} j_{2} j_{3} \otimes j_{4}
											+ j_{1} j_{2} j_{4} \otimes j_{3}
											+ j_{1} j_{3} j_{4} \otimes j_{2}
											+ j_{1} j_{2} j_{3}j_{4} \otimes \un,							
\end{align*}}
which should be compared with \eqref{nonplanar1}. For the right half-unshuffle $\Delta^{\!\!\shuffle}_\succ$ we find up to order three
\begin{align*}
	\Delta^{\!\!\shuffle}_\succ(j_{1}) &=  \un \otimes j_{1}\\
	\Delta^{\!\!\shuffle}_\succ(j_{1} j_{2}) &= j_{2} \otimes j_{1} + \un \otimes j_{1} j_{2}\\
	\Delta^{\!\!\shuffle}_\succ(j_{1} j_{2}j_{3}) &= j_{2} \otimes j_{1} j_{3}
										+ j_{3} \otimes j_{1} j_{2} 
										+ j_{2} j_{3} \otimes  j_{1}
										+ \un \otimes j_{1} j_{2} j_{3}.
\end{align*}

\begin{rmk}\label{rmk:cocom}{\rm{
Observe that $\tau \circ \Delta^{\!\!\shuffle}_\succ = \Delta^{\!\!\shuffle}_\prec$, which reflects cocommutativity of \eqref{unshuffle}. 
}}
\end{rmk}

Let us turn to the coproduct \eqref{HA}, which we would like to split into left and right unshuffle half-coproducts
\begin{equation}
\label{HAsplit}
	\delta = \delta_\succ + \delta_\prec,
\end{equation}
analogous to what we did for \eqref{deshufsplit}. The left unshuffle half-coproduct $\delta_\prec: T(J) \to  T(J) \otimes \bar T(T(J))$ is defined through
\begin{equation}
\label{HAleft}
	\delta_\prec(j_{i_1}\cdots j_{i_n}) := \sum_{S \subseteq [n] \atop 1 \in S} j_S \otimes  j_{J_1} | \cdots | j_{J_k}
					   =\sum_{S \subseteq [n] \atop 1 \in S} j_S \otimes j_{J^S_{[n]}},
\end{equation} 
and the right unshuffle half-coproduct $\delta_\succ: T(J) \to \bar T(J) \otimes T(T(J))$ is defined by
\begin{equation}
\label{HAright}
	\delta_\succ(j_{i_1}\cdots j_{i_n}) := \sum_{S \subset [n] \atop 1 \notin S} j_S \otimes  j_{J_1} | \cdots | j_{J_k}
					   =\sum_{S \subset [n] \atop 1 \notin S} j_S \otimes j_{J^S_{[n]}}.
\end{equation}

\begin{rmk}\label{n-unshuffles}{\rm{Note that for symmetry reasons, one can define a companion left half-unshuffle 
\begin{equation}
\label{HAleft-end}
	\delta_{\bar\prec}(j_{i_1}\cdots j_{i_n}) 
				:= \sum_{S \subseteq [n] \atop n \in S} j_S \otimes  j_{J_1} | \cdots | j_{J_k}
				 =\sum_{S \subseteq [n] \atop n \in S} j_S \otimes j_{J^S_{[n]}},
\end{equation}
and correspondingly another right half-unshuffle $\delta_{\bar\succ}(j_{i_1}\cdots j_{i_n}) :=\sum\limits_{S \subset [n] \atop n \notin S} j_S \otimes j_{J^S_{[n]}}.$}}
\end{rmk}

We shall give a few examples
\begin{align*}
	\delta_\prec(j_{1}) &= j_{1} \otimes \un \\
	\delta_\prec(j_{1} j_{2}) &= j_{1} \otimes j_{2} + j_{1} j_{2} \otimes \un \\
	\delta_\prec(j_{1} j_{2}j_{3}) &= j_{1} \otimes j_{2} j_{3} 
										+ j_{1} j_{2} \otimes  j_{3} 
										+ j_{1} j_{3} \otimes  j_{2}
										+ j_{1} j_{2} j_{3} \otimes \un\\
	\delta_\prec(j_{1} j_{2}j_{3}j_{4}) 	&= j_{1} \otimes j_{2} j_{3}j_{4} 
											+ j_{1} j_{2} \otimes  j_{3}j_{4} 
											+ j_{1} j_{3} \otimes  j_{2} | j_{4}
											+ j_{1} j_{4} \otimes  j_{2}j_{3}\\
										&\quad	
											+ j_{1} j_{2} j_{3} \otimes j_{4}
											+ j_{1} j_{2} j_{4} \otimes j_{3}
											+ j_{1} j_{3} j_{4} \otimes j_{2}
											+ j_{1} j_{2} j_{3}j_{4} \otimes \un.						
\end{align*}
Note the difference between the terms $j_{1} j_{3} \otimes  j_{2} | j_{4} \in T(J) \otimes T(T(J))$ and $j_{1} j_{3} \otimes  j_{2} j_{4}  \in T(J) \otimes T(J)$, which distinguishes the two left unshuffles half-coproduct $\delta_\prec$ and $\Delta^{\!\!\shuffle}_\prec$ at order four. This distinction reflects the one we already observed between the expansions at order four of the full Green's functions in terms of the connected ones in the planar and non-planar cases.

For the right unshuffle half-coproduct up to order three we find
\begin{align*}
	\delta_\succ(j_{1}) &= \un \otimes j_{1} \\
	\delta_\succ(j_{1} j_{2}) &= j_{2} \otimes  j_{1} + \un \otimes  j_{1} j_{2} \\
	\delta_\succ(j_{1} j_{2}j_{3}) &=  j_{2} \otimes j_{1} | j_{3} 
										+ j_{3} \otimes j_{1}  j_{2} 
										+ j_{2} j_{3} \otimes  j_{1}
										+ \un \otimes  j_{1} j_{2} j_{3}\\
	\delta_\succ(j_{1} j_{2}j_{3}j_{4}) 	&= j_{2} \otimes j_{1} | j_{3}j_{4} 
									+  j_{3} \otimes j_{1} j_{2} | j_{4}
									+  j_{4} \otimes j_{1} j_{2} j_{3} \\
								&\quad	+ j_{2} j_{3} \otimes  j_{1} | j_{4} 
										+ j_{2} j_{4} \otimes  j_{1} | j_{3}
										+ j_{3} j_{4} \otimes  j_{1} j_{2}
										+ j_{2} j_{3} j_{4} \otimes j_{1} 
										+ \un \otimes  j_{1} j_{2} j_{3} j_{4}											
\end{align*}

\begin{rmk}\label{rmk:noncocom}{\rm{Note that contrary to Remark \ref{rmk:cocom}, we observe that already at order three we find that $\tau \circ \delta_\succ \neq \delta_\prec$. We would like to emphasize that this difference, that is, the fact that \eqref{unshuffle} is cocommutative but \eqref{HA} is not, distinguishes the non-planar from the planar setting of QFTs. 
}}
\end{rmk}

In the light of coassociativity \eqref{coasso} of the coproducts  \eqref{unshuffle} and \eqref{HA}, the respective splittings in \eqref{deshufsplit} and \eqref{HAsplit} imply general properties to be satisfied by the left and right unshuffle half-coproducts, which will be stated now in the next two definitions.

\begin{defn}
\label{def:unshufCoalg}
A counital unshuffle coalgebra is a coaugmented coassociative coalgebra $\bar C = C \oplus \Kb \un$ with coproduct
\begin{equation}
\label{codend}
	\Delta(c) := \bar\Delta(c) + c \otimes \un + \un \otimes c,
\end{equation}
such that, on $C$, the reduced coproduct splits, $\bar\Delta = \bar\Delta_{\prec} + \bar\Delta_{\succ}$ with 
\begin{eqnarray}
	(\bar\Delta_{\prec} \otimes \id) \circ \bar\Delta_{\prec}   	
		&=& (\id \otimes \bar\Delta)\circ \bar\Delta_{\prec}        	\label{C1}\\
  	(\bar\Delta_{\succ} \otimes \id) \circ \bar\Delta_{\prec}   	
		&=& (\id \otimes \bar\Delta_{\prec})\circ \bar\Delta_{\succ} 	\label{C2}\\
   	(\bar\Delta \otimes  \id) \circ \bar\Delta_{\succ}         	
		&=& (\id \otimes \bar\Delta_{\succ})\circ \bar\Delta_{\succ}     \label{C3}.
\end{eqnarray}
The maps $\bar\Delta_{\prec}$ and $\bar\Delta_{\succ}$ are called respectively augmented left and right unshuffle half-coproducts. A cocommutative unshuffle coalgebra satisfies $\bar\Delta_\prec = \tau \circ \bar\Delta_\succ$, where $\tau$ denotes the twist map, $\tau(x \otimes y):=y\otimes x$. 
\end{defn}

\begin{defn}
An unshuffle bialgebra is a unital and counital bialgebra $\bar B = B \oplus \Kb \un$ with product $m_B(x \otimes y)=: x \cdot_B y$ and coproduct $\Delta$, as well as a counital unshuffle coalgebra $\bar\Delta = \bar\Delta_{\prec} + \bar\Delta_{\succ}$. Moreover, the following compatibility relations hold 
\begin{eqnarray}
	\Delta_{\prec}(a \cdot_B b)  &=& \Delta_{\prec}(a)  \cdot_B \Delta(b)      	\label{D1}\\
  	\Delta_{\succ}(a \cdot_B b)  &=& \Delta_{\succ}(a)  \cdot_B \Delta(b),     \label{D2}
\end{eqnarray}
where
\begin{eqnarray}
	\Delta_{\prec}(a)  &:=& \bar\Delta_{\prec}(a) + a \otimes \un     	\label{D3}\\
  	\Delta_{\succ}(a)  &:=& \bar\Delta_{\succ}(a) + \un \otimes a ,   	\label{D4}
\end{eqnarray}
and $\Delta = \Delta_{\prec} + \Delta_{\succ}$.
\end{defn}

For example, \eqref{unshuffle1} together with the concatenation product defines the structure of a cocommutative unshuffle bialgebra on $\bar T(J)$.

\begin{thm}  \cite{EP1} \label{thm:bialg}
The bialgebra $\bar T(T(J))$ with its coproduct \eqref{HA} split into left and right unshuffle half-coproducts, $\delta = \delta_{\prec} + \delta_{\succ}$, defined in \eqref{HAleft} and \eqref{HAright}, is an unshuffle bialgebra. 
\end{thm}
   
\begin{rmk} \label{Cor:singleInterval}
{\rm{
The left unshuffle half-coproduct \eqref{HAleft} can be further split 
\begin{equation}
\label{HAleft1}
	\bar \delta_{\prec}(j_{i_1}\cdots j_{i_n}) = \hat{\delta}_{\prec}(j_{i_1}\cdots j_{i_n}) + \tilde{\delta}_{\prec}(j_{i_1}\cdots j_{i_n}),
\end{equation} 	
where $\hat{\delta}_{\prec}(j_{i_1}\cdots j_{i_n}) \in T(J) \otimes T(J)$ is the linearized part, and the rest $ \tilde{{\delta}}_{\prec}(j_{i_1}\cdots j_{i_n}) \in T(J) \otimes\bigoplus_{n>1} T(T(J))_n$. The linearized part is best described in terms of intervals
\begin{equation}
\label{HAleft2}
	\hat{\delta}_\prec (j_{i_1}\cdots j_{i_n}) = \sum_{I_1 \coprod I_2 \coprod I_3 = [n] \atop 1\in I_1, I_2 \neq \emptyset} j_{I_1 \coprod I_3} \otimes  j_{I_2}.
\end{equation} 	
Here $I_1,I_2,I_3$ are three disjoint intervals, such that $I_1 \coprod I_2 \coprod I_3 = [n]$ and $I_3$ is possibly empty. Moreover, the minimal elements of each interval satisfy $\min(I_1)=1 < \min(I_2) < \min(I_3)$.}}
\end{rmk}

%%%%%%%%%%%%%%%%%%%%%%%%%%%%%%%%%%%%%%%

\subsection{Green's functions revisited}
\label{ssect:Gfcts}

Recall that the main mathematical purpose of planar and non-planar QFTs is the calculation of Green's functions, from which, ultimately, physical quantities can be computed. Green's functions are also called $n$-point correlation functions, and in real space they are defined as (functional) expectation values of products of $n$ field operators at different positions \cite{itzyksonzuber,zavialov}. 

In this work we will consider the generating series of Green's functions for planar (resp.~non-planar) theories as elements in the dual space $Lin(\bar{T}(T(J)),\mathbb{C})=:\bar{T}^*(T(J))$ of linear maps from the Hopf algebra $\bar{T}(T(J))$ to the complex numbers $\mathbb{C}$ (resp.~$Lin(\bar{T}(J),\mathbb{C})=:\bar{T}^*(J)$). 
This approach allows to handle Green's functions using the full machinery of Hopf algebras and unshuffle bialgebras.

Note that the problem of (ultraviolet) divergencies and its solution in terms of the process of renormalization \cite{caswellkennedy,collins} demands for regularization procedures, which involves replacing the field of complex numbers as target space of linear functions, by some commutative and unital $\mathbb{C}$-algebra $A$, for example the algebra of Laurent series $\mathbb{C}[\varepsilon^{-1},\varepsilon]]$ in a dimensional regularization parameter $\varepsilon$. We point out that changing the target algebra from $\mathbb{C}$ to such an algebra $A$ would not change the underlying algebraic and combinatorial framework. Indeed, our forthcoming developments apply as well in this more general setting.

To make this more precise we consider first $Lin(T(J),\mathbb{C})=T^*(J)$, and introduce for each word $w=j_{i_1}\cdots j_{i_k} \in J^*$ its dual element $d_w \in T^*(J)$, which is defined as a linear form over $T(J)$, such that for any word $w' \in J^*$ one has that $d_w(w'):=(w,w')=1$ if $w=w'$, and zero else. General linear forms over $T(J)$ are then defined as formal series $F:=\sum_{w \in J^*} (F,w)d_w $.

From this perspective generating functionals of non-planar Green's functions are viewed now as formal series, with the full, connected and 1PI Green's functions as coefficients:
\begin{align*}
	 \tau_{{Z}}(j_{i_1}\cdots j_{i_l})&:=Z^{(l)}_{j_{i_1}\cdots j_{i_l}}\\
	 \tau_{{W}}(j_{i_1}\cdots j_{i_l})&:=W^{(l)}_{j_{i_1}\cdots j_{i_l}}\\
	 \tau_{{\Gamma}}(j_{i_1}\cdots j_{i_l})&:=\Gamma^{(l)}_{j_{i_1}\cdots j_{i_l}}.   
\end{align*}
We call these linear maps respectively the non-planar full, connected and 1PI Green's functions. 

Let us focus now on the planar case. Here, one may apply the above interpretation of generating functionals of Green's functions as linear forms, when restricted to degree one,  that is, for words in $T(J) \hookrightarrow T(T(J))$ one writes
\begin{align*}
	 \tau_{\mathrm{Z}}(j_{i_1}\cdots j_{i_m})&:=\mathrm{Z}^{(m)}_{j_{i_1}\cdots j_{i_m}}\\
	 \tau_{\mathrm{W}}(j_{i_1}\cdots j_{i_m})&:=\mathrm{W}^{(m)}_{j_{i_1}\cdots j_{i_m}}\\
	 \tau_{\mathrm{\Gamma}}(j_{i_1}\cdots j_{i_m})&:=\mathrm{\Gamma}^{(m)}_{j_{i_1}\cdots j_{i_m}}.   
\end{align*}
The critical step is the extension to all of $\bar{T}(T(J))$. Indeed, the first step in clarifying the relations between these linear maps from a Hopf algebra point of view is taken by considering full Green's functions as a multiplicative map on $\bar T(T(J))$, that is, by extending the linear map of full Green's functions, $\tau_{\mathrm{Z}}$, multiplicatively to all of $\bar T(T(J))$
$$
	\tau({\bf 1})=1,\quad \tau_{\mathrm{Z}}(w_1 | \cdots | w_n):=\tau_{\mathrm{Z}}(w_1) \cdots \tau_{\mathrm{Z}}(w_n).
$$  
In other terms, $\tau_{\mathrm{Z}}\in G_\mathbb{C}$, the group of characters on $\bar T(T(J))$ -- we shall see later that it is natural to require that both $\tau_{\mathrm{W}}$ and $\tau_{\mathrm{\Gamma}}$ are infinitesimal characters.

Next we remind the reader of the convolution product \eqref{convolutionProd} introduced above in the context of the space of linear maps on a Hopf algebra with values in a commutative unital algebra, say, for instance, the complex numbers. This turns $\bar{T}^*(T(J))$ into a non-commutative unital algebra, where the product is defined in terms of the coproduct (\ref{HA}), i.e., for $\alpha, \beta \in Lin(\bar T(T(A)),\mathbb{C})$
\begin{equation}
\label{ConvProd1}
	\alpha \star \beta := m_{\mathbb{C}} \circ (\alpha \otimes \beta) \circ \delta.
\end{equation}
Here $m_{\mathbb{C}}$ stands for the product map in $\mathbb{C}$. 

The splitting of the coproduct (\ref{HA}) into left and right unshuffle half-coproducts \eqref{HAsplit} can be lifted to the algebra $\bar{T}^*(T(J))$. To this extend we define the left and right half-shuffle convolution products 
\begin{eqnarray}
	\alpha \prec \beta &:=& m_{\mathbb{C}} \circ (\alpha \otimes \beta)\circ \delta_\prec 		\label{dend1}\\
	\alpha \succ \beta &:=& m_{\mathbb{C}} \circ (\alpha \otimes \beta)\circ \delta_\succ ,	\label{dend2}
\end{eqnarray}
such that \eqref{HAsplit} implies that
\begin{equation}
\label{ConvProd2}
	\alpha \star \beta =  \alpha \succ \beta + \alpha \prec \beta.
\end{equation}

\begin{rmk}{\rm{
The same construction would hold for the dual space of an arbitrary unshuffle bialgebra.
}}
\end{rmk}

Definition \ref{def:unshufCoalg} implies the following relations for the binary operations $\succ$ and $\prec$. Note that with the aim of emphasising the shuffle-type behaviour, we replace the product $\star$ by the classical notation for the shuffle product $\shuffle$. 
\begin{eqnarray}
	(a\prec b)\prec c   &=& a\prec(b \shuffle c)        		\label{A1}\\
  	(a\succ b)\prec c   &=& a\succ(b\prec c)   			\label{A2}\\
   	a\succ(b\succ c)   &=& (a \shuffle b)\succ c        	\label{A3},
\end{eqnarray}
where
\begin{equation}
	a \shuffle b := \ a \prec b + a \succ b. 			\label{dendassoc}       
\end{equation}

These are the axioms defining the abstract notion of a {\it{shuffle}}, or {\it{dendrimorphic}} algebra \cite{EP1,EP2}, which is a $\mathbb{K}$-vector space $D$ together with two bilinear compositions $\prec$ and $\succ$, the so-called left and right half-shuffle products, satisfying \eqref{A1}, \eqref{A2}, and \eqref{A3}. In fact, these axioms imply that any shuffle algebra is an associative algebra for the shuffle product defined in (\ref{dendassoc}), and we call $\shuffle$ the shuffle product on $D$.

A {\it{commutative shuffle algebra}} is a shuffle algebra, where the left and right half-shuffles are identified as follows
$$	
	x \succ y = y \prec x ,
$$
so that in particular the shuffle product $\shuffle$ is commutative: $x \shuffle y = x \prec y + x \succ y =y \shuffle x$. The standard example of a commutative shuffle algebra structure was introduced on $\bar{T}(J)$ in \eqref{shuffleProd}. We remark that the same structure is obtained as the dual to the one on $\bar T(J)$, when equipped with the concatenation product and the unshuffle coproduct. The duality is obtained by requiring the words in $J^\ast$ to form an orthonormal basis.

Shuffle algebras are not naturally unital. This is because it is impossible to ``split'' the unit equation, $\un \shuffle a=a\shuffle \un=a$, into two equations involving the left and right half-shuffle products $\succ$ and $\prec$. This issue is circumvented by using the ``Sch\"utzenberger trick'', that is, for $D$ a shuffle algebra, $\bar D := D \oplus \mathbb{K}\un$ denotes the shuffle algebra augmented by a unit $\un$, such that
\begin{equation}
\label{unit-dend}
    a \prec \un := a =: \un \succ a
    \hskip 12mm
    \un \prec a := 0 =: a \succ \un,
\end{equation}
implying $a \shuffle \un = \un \shuffle a = a$. By convention, $\un \shuffle \un=\un$, but $\un \prec \un$ and $\un \succ \un$ cannot be defined consistently in the context of the axioms of shuffle algebras.

In the light of \eqref{dend1} and \eqref{dend2} we arrive at the next result.

\begin{prop}
\label{shufflealgTTJ}
	The space $(Lin(\bar{T}(T(J)),\mathbb{C}), \prec, \succ)$ is a shuffle algebra.
\end{prop}

\begin{rmk}{\rm{We state the proposition for $\bar{T}(T(J))$, since we always try to put the emphasis on planar QFT. However, the proof depends only on the fact that $\bar T(T(J))$ is an unshuffle bialgebra and therefore would hold for an arbitrary unshuffle bialgebra $B$ with $\bar B=B\oplus\Kb\un$. In particular, the property holds for $Lin(\bar{T}(J),\mathbb{C})$.
}}
\end{rmk}

\begin{proof}
For arbitrary $\alpha , \beta , \gamma \in T^*(T(J))$,
$$
	(\alpha \prec \beta)\prec \gamma	
			=m_{\mathbb{C}^{[3]}}\circ ((\alpha \prec \beta)\otimes \gamma)\circ\delta_\prec 
			=m_{\mathbb{C}^{[3]}}\circ (\alpha \otimes \beta\otimes \gamma)\circ(\delta_\prec\otimes \id)\circ \delta_\prec,
$$
where ${\mathbb{C}^{[3]}}$ stands for the product map from $\mathbb{C}^{\otimes 3}$ to $\mathbb{C}$. Similarly 
\begin{eqnarray*}
	\alpha\prec ( \beta\shuffle \gamma)
		&=&m_{\mathbb{C}}\circ (\alpha\otimes (\beta\shuffle \gamma))\circ \delta_\prec \\
		&=&m_{\mathbb{C}^{[3]}}\circ (\alpha \otimes \beta\otimes \gamma)\circ (\id \otimes \overline\delta)\circ\delta_\prec,
\end{eqnarray*}
where $\overline\delta(u) = \delta(u) - u \otimes \un - \un \otimes u$ is the reduced coproduct. So that the identity $(\alpha \prec \beta)\prec \gamma=\alpha \prec (\beta\shuffle \gamma)$ follows from $(\delta_\prec\otimes \id)\otimes\delta_\prec =(\id \otimes \overline\delta)\circ\delta_\prec$, and similarly for the other identities characterizing shuffle algebras.

We equip the shuffle algebra $(T^*(T(J)), \prec, \succ)$ with the unit $\varepsilon$ -- recall that $\varepsilon$ is the null map on $T(T(J))$, and the identity map on $T(J)^{\otimes 0}\cong \mathbb{C}\un$. That is, for an arbitrary $\alpha$ in $\bar{T}^*(T(J))$,
$$
	\alpha \prec \varepsilon 
	= \alpha 
	= \varepsilon \succ \alpha,\ \ \varepsilon \prec \alpha 
	= 0 
	= \alpha \succ \varepsilon.
$$
\end{proof}

Let us introduce some useful notations. Let $L_{a \succ} \left( b \right) := a \succ b =; R_{\succ b} \left( a \right) $. The shuffle axioms yield
$$	
	L_{a \succ} L_{b \succ} = L_{a \shuffle b \succ}, 
	\qquad\ 
	R_{\prec a} R_{\prec b} = R_{\prec b \shuffle a}.
$$

Recall that a left pre-Lie algebra \cite{cartier2, manchon1} is a $\mathbb{K}$-vector space $V$ equipped with a bilinear product $\vartriangleright : V \otimes V \to V$, such that for arbitrary $a,b,c \in V$ 
\begin{equation}
\label{pLrel}
	a \vartriangleright (b\vartriangleright c) - (a\vartriangleright b) \vartriangleright c
		= b \vartriangleright (a\vartriangleright c) - (b\vartriangleright a) \vartriangleright c.
\end{equation}
It implies that the bracket $[a,b]:= a \vartriangleright b - b \vartriangleright a$ satisfies the Jacobi identity. A right pre-Lie algebra is defined appropriately. Note that any associative algebra is pre-Lie. For several reasons pre-Lie algebras play a key role, e.g., in the understanding of recursive equations such as Bogoliubov's counterterm formula in perturbative quantum field theory \cite{EGP2,EMP}. The next lemma follows directly from the axioms (\ref{A1})-(\ref{A3}) of shuffle products.

\begin{lem}
Let $D$ be a shuffle algebra. The product $\vartriangleright : D \otimes D \to D$ 
$$
	a \vartriangleright b := a \succ b - b \prec a 
$$
is left pre-Lie. We write its left action $L_{a \vartriangleright} \left( b \right) = a \vartriangleright b=L_{a \succ} - R_{\prec a}$.
\end{lem}

\noindent Note that $[a,b]= a \vartriangleright b - b \vartriangleright a = a \shuffle b - b \shuffle a$ for all $a,b \in D$. The pre-Lie product is trivial (null) on commutative shuffle algebras, since we then have $a\succ b=b\prec a$.

The following set of left and right half-shuffle words in $\bar{D}$ are defined recursively for fixed elements $x_1,\ldots, x_n \in D$, $n \in \mathbb{N}$
 \allowdisplaybreaks{
\begin{eqnarray*}
    w^{(0)}_{\prec}(x_1,\ldots, x_n) &:=& \un =:w^{(0)}_{\succ}(x_1,\ldots, x_n) \\
    w^{(n)}_{\prec}(x_1,\ldots, x_n) &:=& x_1 \prec \bigl(w^{(n-1)}_\prec(x_2,\ldots, x_n)\bigr)\\
    w^{(n)}_{\succ}(x_1,\ldots, x_n) &:=& \bigl(w^{(n-1)}_\succ(x_1,\ldots, x_{n-1})\bigr)\succ x_n.
\end{eqnarray*}}
In case that $x_1=\cdots = x_n=x$ we simply write $x^{\prec{n}}:=w^{(n)}_{\prec}(x,\ldots, x)$ and $x^{\succ{n}} := w^{(n)}_{\succ}(x,\ldots, x)$.

In the unital algebra $\bar D$ both the exponential and logarithm maps are defined in terms of the associative product~(\ref{dendassoc})
\begin{equation}
\label{ExpLog}
	\exp^\shuffle(x):=\un + \sum_{n > 0} \frac{x^{\shuffle n}}{n!}  
	\quad\ {\rm{resp.}} \quad\ 
	\log^\shuffle(\un+x):=-\sum_{n>0}(-1)^n\frac{x^{\shuffle n}}{n}. 
\end{equation}
Notice that we do not consider convergence issues: in practice we will apply such formal power series computations either in a purely algebraic setting (formal convergence arguments would then apply), or when dealing with graded algebras (then the series will reduce to a finite number of nonzero terms when restricted to a given graded component).

It is also convenient to introduce the  ``(time-)ordered'' exponential
$$
	\exp^{\prec}( x) := \un + \sum_{n > 0} x^{\prec{n}}.
$$
Similarly, we also define $\exp^{\succ}( x):=\sum_{n \ge 0} x^{\succ n}.$ It corresponds to the usual time-ordered exponential in physics, when the shuffle product is defined with respect to products of, say, matrix- or operator-valued iterated integrals. Notice that $X=\exp^{\prec}( x)$ and $Z=\exp^{\succ}( x)$ are respectively the formal solutions of the two linear fixed point equations 
\begin{eqnarray}
	X&=&\un + x\prec X			\label{KeyEquation1}\\
	Z&=& \un + Z \succ x	.		\nonumber
\end{eqnarray}
The solution of equation \eqref{KeyEquation1} can also be written in terms of the proper exponential map (\ref{ExpLog}). Further below we will see that \eqref{KeyEquation1} is the key ingredient in our approach to a Hopf algebraic description of the functional relations among (non-)planar Green's functions. This point of view paves the way to new formal results on the combinatorics of Green's functions. 

\begin{lem}\label{inverse-shuffle}
Let $D$ be a shuffle algebra, and $\bar D$ its augmentation by a unit $\un$. For $x \in D$ we have
$$
	\exp^{\succ}(-x) \shuffle \exp^{\prec}(x) = \un.
$$
\end{lem}

\begin{proof}
Indeed, we see that
\begin{eqnarray*}
	\exp^{\succ}(-x) \shuffle \exp^{\prec}(x)- \un
	&=& \sum\limits_{n+m\geq 1}(-1)^n\big\{(x^{\succ n})\prec (x^{\prec m}) 
					+ (x^{\succ n})\succ (x^{\prec m})\big\}\\
	&=& \sum\limits_{n>0,m\geq 0}(-1)^n(x^{\succ n})\prec (x^{\prec m}) 
				+ \sum\limits_{n\geq 0,m> 0}(-1)^n(x^{\succ n})\succ (x^{\prec m}).
\end{eqnarray*}
Now, since $(-1)^n(x^{\succ n})\prec (x^{\prec m})=(-1)^n((x^{\succ n-1})\succ x)\prec (x^{\prec m})=(-1)^n(x^{\succ n-1})\succ (x^{\prec m+1})$, the proof follows.
\end{proof}

Another useful result follows from the computation of the composition inverse of the time-ordered exponential.

\begin{lem}\label{inverse}
Let $D$ be a shuffle algebra, and $\bar D$ its augmentation by a unit $\un$. For $x\in D$ and $X := \un +Y :=\exp^\prec (x)$, then
$$
	x=Y\prec \big(\sum\limits_{n\geq 0}(-1)^nY^{\shuffle n}\big).
$$
\end{lem}

\begin{proof}
We follow \cite{foipat}. From $X=\un + \sum_{n > 0}x^{\prec n}$, we get $X-\un=Y=x \prec X$. On the other hand, the (formal) inverse of $X$ {\it{for the shuffle product}} is given by $X^{-1}=\frac{1}{1+Y}=\sum_{k\geq 0}(-1)^kY^{\shuffle k}$. We finally obtain 
$$
	x = x\prec \un 
	  = x\prec (X\shuffle X^{ -1})
	  = (x\prec X)\prec X^{-1}
	  =Y\prec \big(\sum\limits_{n\geq 0}(-1)^nY^{\shuffle n}\big).
$$
\end{proof}

%%%%%%%%%%%%%%%%%%%%%%%%%%%%%%%%%%%%%%%

\subsection{Towards planar group theory}
\label{ssect:planargroup}

Recall that $\bar{T}^*(T(J)):=Lin(\bar T(T(J)),{\mathbb C})$ and that a linear form $\phi \in \bar{T}^*(T(J))$ is called a character if it is unital, $\phi(\un)=1$, and multiplicative, i.e., for all $a,b \in \bar{T}(T(J))$, $\phi(a|b)=\phi(a)\phi(b).$ A linear form $\kappa \in \bar{T}^*(T(J))$ is called infinitesimal character, if $\kappa(\un)=0$, and if $\kappa(a|b)=0$ for all $a,b\in T(T(J))$. Characters and infinitesimal characters are bijectively related through the exponential map defined with respect to the convolution product. We write $Ch(\kappa)$ for the obvious extension of a linear form on $T(J)$ (e.g.~the restriction to $T(J)$ of an infinitesimal character) to a character, defined by $Ch(\kappa)(\un):=1$, and $Ch(\kappa)(w_1| \cdots|w_k):=\kappa(w_1) \cdots \kappa(w_k)$. Conversely, for an arbitrary $F \in \bar{T}^*(T(J))$, let us write $Res(F)$ for the infinitesimal character, which is defined as the restriction of $F$ to $T(J)$, and the null map on other tensor powers of $T(J)$ in $\bar{T}(T(J))$. 

\smallskip

The linear fixed point equation \eqref{KeyEquation1} is characterised in the next theorem.

\begin{thm}\label{thm:Gg}
There exists another natural bijection between the group of characters and the Lie algebra of infinitesimal characters on $\bar{T}(T(J))$. Indeed, for a character $\phi$ there exists an unique infinitesimal character $\kappa$ such that 
\begin{equation}
\label{dendeq}
	\phi = \varepsilon + \kappa \prec \phi = \exp^{\prec}(\kappa),
\end{equation}
and conversely, for an infinitesimal character $\kappa$
$$
	\phi:=\exp^{\prec}( \kappa)
$$
is a character.
\end{thm}

Let us use in the following the shortcut ``Hopf- or Sweedler-type'' notation for $\delta_\prec(w)=:w^{1,\prec}\otimes w^{2,\prec}$.
\begin{proof}
We know from Lemma~\ref{inverse} that the implicit equation $\phi = \varepsilon + \kappa \prec \phi = \exp^{\prec}(\kappa)$ has a unique solution $\kappa$ in $\bar{T}^*(T(J))$. Let us consider the infinitesimal character $\mu:=Res(\kappa)$, and let us show that $\mu$ also solves $\phi = \varepsilon + \mu \prec \phi$; the first part of the Theorem will follow.

Indeed, for an arbitrary $w \in T(T(J)), w=w_1|\cdots|w_n$, notice first that by definition of the product $\prec$, and due to the vanishing of $\mu$ on any tensor power $T(J)^{\otimes k}$, for $k \not= 1$, we have:
$$
	(\mu \prec \phi)(w) 	= \mu(w_1^{1,\prec})\phi(w_1^{2,\prec}|w_2|\cdots|w_n)
					= \kappa(w_1^{1,\prec})\phi(w_1^{2,\prec}|w_2|\cdots|w_n).
$$
We immediately obtain, since 
$$
	\phi(w_1)	=(\varepsilon + \kappa\prec \phi)(w_1)
			=\kappa(w_1^{1,\prec})\phi(w_1^{2,\prec})
			=\mu(w_1^{1,\prec})\phi(w_1^{2,\prec})
$$ 
that, for any $i>1$
$$
	\phi (w_1|\cdots|w_n)	=\phi(w_1)\phi(w_2|\cdots|w_n)
						=\mu(w_1^{1,\prec})\phi(w_1^{2,\prec}|w_2|\cdots|w_n))
						=(\varepsilon + \mu\prec \phi)(w_1|\cdots|w_n),
$$
from which the property follows. Conversely:
$$
	\exp^{\prec}(\kappa)(w_1|\cdots|w_n)
							=(\varepsilon + \kappa\prec\exp^{\prec}(\kappa))(w_1|\cdots|w_n)
							=\kappa(w_1^{1,\prec})\exp^{\prec}(\kappa)(w_1^{2,\prec}|\cdots|w_n).
$$
Assuming by induction that the property $\exp^{\prec}(\kappa)(w_1'|\cdots|w_k')=\exp^{\prec}(\kappa)(w_1') \cdots\exp^{\prec}(\kappa)(w_k')$ holds for elements $w_1'|\cdots|w_k'\in T(T(J))$ of total degree less than the degree of $w_1|\cdots|w_n$, yields
\begin{eqnarray*}
	\exp^{\prec}(\kappa)(w_1|\cdots|w_n)
				&=&\kappa(w_1^{1,\prec})\exp^{\prec}(\kappa)(w_1^{2,\prec})
								\exp^{\prec}(\kappa)(w_2)\cdots\exp^{\prec}(\kappa)(w_n)\\					
				&=&\exp^{\prec}(\kappa)(w_1)\exp^{\prec}(\kappa)(w_2) \cdots \exp^{\prec}(\kappa)(w_n).
\end{eqnarray*}
\end{proof}

The next result shows that equation \eqref{KeyEquation1} has a solution in terms of the proper exponential map defined with respect to the convolution \eqref{ConvProd1}, which splits as a shuffle product \eqref{ConvProd2} in the sense of \eqref{dendassoc}. We recall that $L_{a \rhd}(b):= a \rhd b = a \succ b - b \prec a$, where the product $a \rhd b$ satisfies the pre-Lie relation (\ref{pLrel}). See \cite{cartier2,manchon1} for details and more results on pre-Lie algebras.

\begin{thm}\cite{EM1,EM2}\label{thm:ExpSolution}
Equation \eqref{dendeq} has the exponential solution  
$$
	\phi = \exp^\star\big(\Omega'(\kappa)\big),
$$ 
where $\Omega'(\kappa)$ is the pre-Lie Magnus expansion, which obeys the following recursive equation
$$
	\Omega'(\kappa) = \frac{L_{\Omega' \rhd}}{\exp(L_{\Omega' \rhd})-1}(\kappa)
            =\sum\limits_{m\ge 0} \frac{B_m}{m!}\ L^{m}_{\Omega' \rhd}(\kappa).
$$
Here, the $B_l$'s are the Bernoulli numbers.
\end{thm}

For a proof of this theorem we refer the reader to \cite{EM1,EM2}. Note that for a commutative shuffle algebra the pre-Lie product is the null product, and $\Omega'(\kappa)$ reduces to the identity map. In this case the solution to \eqref{dendeq} is given by $\phi =  \exp^\star(\kappa)$. Let us mention that the pre-Lie Magnus expansion $\Omega'(\kappa)$ can also be understood from the point of view of enveloping algebras of pre-Lie algebras \cite{chappat}.

%%%%%%%%%%%%%%%%%%%%%%%%%%%%%%%%%%%%%%%

\section{From full to connected and noncrossing planar Green's functions}
\label{sect:FullConNoncross}

%%%%%%%%%%%%%%%%%%%%%%%%%%%%%%%%%%%%%%%

\subsection{From full to connected Green's functions}
\label{ssect:FullCon}

Recall the fixed point equation relating the generating functionals $\mathrm{Z}[j]$ and $\mathrm{W}[j]$ in the planar context
\begin{equation}
\label{ftoc}
	\mathrm{Z}[j]=1+\mathrm{W}[j\mathrm{Z}[j]].
\end{equation}
When expanded to compute the planar $n$-point function $\mathrm{Z}^{(n)}_{j_{i_1}\cdots j_{i_n}}$, the equation reads
\begin{equation}
\label{ftocExpanded}
	\mathrm{Z}^{(n)}_{j_{i_1}\ldots j_{i_n}}=\sum_{A=\{1=a_1,\ldots,a_k\}\subset [n]}
	\mathrm{W}^{(k)}_{j_{i_{a_1}}j_{i_{a_2}}\cdots j_{i_{a_k}}}
		\mathrm{Z}^{(a_2-a_1-1)}_{j_{i_{a_1+1}} \cdots j_{i_{a_2-1}}}
		\cdots  \mathrm{Z}^{(n-a_k)}_{j_{i_{a_k+1}}\cdots j_{i_{n}}}.
\end{equation}

We use now the half-(un)shuffle machinery to rewrite, in a rather natural way, this convoluted relation. Recall that $\bar{T}^*(T(J))$ is a unital shuffle algebra for the left and right half-shuffle products $\prec ,\succ$ defined in terms of the unshuffle half-coproducts \eqref{HAleft} respectively \eqref{HAright}. 

We define the linear map $\tau_{\mathrm{W}} \in \bar{T}^*(T(J))$ associated to planar connected $n$-point functions, i.e., $\tau_{\mathrm{W}}(w):=\mathrm{W}^{(|w|)}_w$ for any word $w=j_{i_1}\cdots j_{i_n} \in T_n(J)$,  $\tau_{\mathrm{W}}(\un)=0$, and $\tau_{\mathrm{W}}$ is zero on the components $T_n(T(J))$ for $n\geq 2$. The latter requirement, in particular, is natural in the light of desired connectedness. Hence, the map $\tau_{\mathrm{W}}$ defines an infinitesimal character with respect to the Hopf algebra $\bar T(T(J))$. On the other hand, recall that the planar full Green's function $\tau_{\mathrm{Z}}(w):=\mathrm{Z}^{(|w|)}_w$ is supposed to be multiplicative. Then we find that, by the very definition of the left half-shuffle product $\prec$, the relation between planar full and connected Green's functions (\ref{ftoc}) is encoded equivalently by the linear fixed point equation
\begin{equation}
\label{pFullConGF}
	\tau_{\mathrm{Z}} = \varepsilon + \tau_{\mathrm{W}} \prec \tau_{\mathrm{Z}}.
\end{equation}
This claim follows immediately from unfolding \eqref{pFullConGF} when evaluated on the word $w=j_{i_1}\cdots j_{i_n} \in T_n(J)$. Indeed, the very definition of the left half-unshuffle \eqref{HAleft} on $\bar{T}(T(J))$, together with the multiplicativity of $\tau_{\mathrm{Z}}$ imply 
\begin{align}
	\mathrm{Z}^{(n)}_{j_{i_1}\ldots j_{i_n}}	=\tau_{\mathrm{Z}}(w)
									&=(\tau_{\mathrm{W}} \prec \tau_{\mathrm{Z}}) (w) \nonumber\\
									&=\sum_{S=\{1=s_1,\ldots,s_k\}\subset [n]}
	\mathrm{W}^{(k)}_{j_{i_{s_1}}j_{i_{s_2}}\cdots j_{i_{s_k}}}
		\mathrm{Z}^{(s_2-s_1-1)}_{j_{i_{s_1+1}} \cdots j_{i_{s_2-1}}}
		\cdots  \mathrm{Z}^{(n-s_k)}_{j_{i_{s_k+1}}\cdots j_{i_{n}}}. \label{unshuffleRelation1}						
\end{align}
This exemplifies the naturalness of the (half-)unshuffle structure in the context of the relation between planar full and connected Green's functions. 

\begin{rmk}{\rm{
\begin{enumerate}

\item 
The companion equation $\mathrm{Z}[j]=1+\mathrm{W}[\mathrm{Z}[j]j]$ is encoded in terms of the second left half-unshuffle defined in  Remark \ref{n-unshuffles}, equation \eqref{HAleft-end}.

\item Theorem \ref{thm:ExpSolution} implies that the planar full Green's function can be written as a proper exponential in terms of the planar connected Green's function
$$
	\tau_{\mathrm{Z}} = \exp^\star\big(\Omega'(\tau_{\mathrm{W}})\big). 
$$ 
Recall that the associative shuffle product $a \star b = a \prec b + a \succ b$ is non-commutative.  

\item The last item should be seen in the context of the non-planar setting, in which the full and connected Green's functions are related through the exponential map. We refer the reader to  subsection \ref{non-planar-case}.
\end{enumerate}
}}
\end{rmk}

Let us illustrate equation \eqref{unshuffleRelation1} by expanding its solution in terms of the ordered exponential of the infinitesimal character of planar connected Green's functions, $\tau_{\mathrm{Z}} = \exp^{\prec}(\tau_{\mathrm{W}} )$, up to order four (compare with the identities obtained at the beginning of the article using relation (\ref{ftoc})). At order one we have 
\begin{align*}
	\tau_{\mathrm{Z}}(j_1) &=  (\varepsilon + \tau_{\mathrm{W}} + \tau_{\mathrm{W}} \prec  \tau_{\mathrm{W}} 
	+  \tau_{\mathrm{W}} \prec (\tau_{\mathrm{W}} \prec  \tau_{\mathrm{W}} ) + \cdots )(j_1)\\
	&= \tau_{\mathrm{W}}(j_1),
\end{align*}
since $\varepsilon(j_1)=0$ and powers of left half-shuffle products beyond order one applied to the letter $j_1$ are zero as well due to $\tau_{\mathrm{W}}(\un)=0$. For the word $w=j_1j_2$ the left half-unshuffle yields $\delta_{\prec}(j_1j_2)=j_1j_2 \otimes \un + j_1 \otimes j_2$, such that
\begin{align*}
	\tau_{\mathrm{Z}}(j_1j_2) &= \tau_{\mathrm{W}}(j_1j_2) +  \tau_{\mathrm{W}}(j_1) \tau_{\mathrm{W}}(j_2).
\end{align*}
For the order three word $w=j_1j_2j_3$ the left half-unshuffle yields
\begin{align*}
	 \tau_{\mathrm{Z}}(j_1j_2j_3) 	&=  \tau_{\mathrm{W}}(j_1j_2j_3) 
								+ \tau_{\mathrm{W}}(j_1)\tau_{\mathrm{W}}(j_2j_3)
								+ \tau_{\mathrm{W}}(j_1j_3)\tau_{\mathrm{W}}(j_2)  \\
							&\quad\quad
								+ \tau_{\mathrm{W}}(j_1j_2)\tau_{\mathrm{W}}(j_3) 
								+ \tau_{\mathrm{W}}(j_1)\tau_{\mathrm{W}}(j_2) \tau_{\mathrm{W}}(j_3).
\end{align*}
And for the order four word $w=j_1j_2j_3j_4$, the left half-unshuffle results in the lengthy expansion
\begin{align}
	\tau_{\mathrm{Z}}(j_1j_2j_3j_4) &=  \tau_{\mathrm{W}}( j_{1} j_{2} j_{3}j_{4}) 
					+ \tau_{\mathrm{W}}(j_{1} j_{2}) \tau_{\mathrm{Z}}(j_{3}j_{4})
					+ \tau_{\mathrm{W}}(j_{1} j_{3}) \tau_{\mathrm{Z}}(j_{2})\tau_{\mathrm{Z}}(j_{4})
					+ \tau_{\mathrm{W}}(j_{1} j_{4}) \tau_{\mathrm{Z}}(j_{2}j_{3}) 				\label{recOrd4}\\
				   &\quad + \tau_{\mathrm{W}}(j_1)\tau_{\mathrm{Z}}(j_{2}j_{3}j_4)
				   	+ \tau_{\mathrm{W}}(j_{1} j_{2}j_3) \tau_{\mathrm{Z}}(j_{4})
					+ \tau_{\mathrm{W}}(j_{1} j_{2}j_4) \tau_{\mathrm{Z}}(j_{3})
					+ \tau_{\mathrm{W}}(j_{1} j_{3}j_4) \tau_{\mathrm{Z}}(j_{2})					\nonumber\\
				   &= \tau_{\mathrm{W}}( j_{1} j_{2} j_{3}j_{4}) 
					+ \tau_{\mathrm{W}}(j_{1} j_{2}) \tau_{\mathrm{W}}(j_{3}j_{4}) 
					+ \tau_{\mathrm{W}}(j_{1} j_{4}) \tau_{\mathrm{W}}(j_{2}j_{3}) 
					+ \tau_{\mathrm{W}}(j_{1} j_{3}) \tau_{\mathrm{W}}(j_{2})\tau_{\mathrm{W}}(j_{4}) 	\nonumber\\ 
				   &\quad   
				   	+ \tau_{\mathrm{W}}(j_{1} j_{2}) \tau_{\mathrm{W}}(j_{3})\tau_{\mathrm{W}}(j_{4}) 
					+ \tau_{\mathrm{W}}(j_{1} j_{4}) \tau_{\mathrm{W}}(j_{2})\tau_{\mathrm{W}}(j_{3})  
					+ \tau_{\mathrm{W}}(j_3j_4)\tau_{\mathrm{W}}(j_1)\tau_{\mathrm{W}}(j_2)		\nonumber\\
				   &\quad
						+ \tau_{\mathrm{W}}(j_2j_3)\tau_{\mathrm{W}}(j_1)\tau_{\mathrm{W}}(j_4)  
						+\tau_{\mathrm{W}}(j_2j_4) \tau_{\mathrm{W}}(j_1)\tau_{\mathrm{W}}(j_3) 
						+ \tau_{\mathrm{W}}(j_1)\tau_{\mathrm{W}}(j_2)\tau_{\mathrm{W}}(j_3) \tau_{\mathrm{W}}(j_4) \nonumber\\
				  &\quad	
				  	+ \tau_{\mathrm{W}}(j_1)\tau_{\mathrm{W}}(j_{2}j_{3}j_4)
					+ \tau_{\mathrm{W}}(j_{1} j_{2}j_3) \tau_{\mathrm{W}}(j_{4})
					+ \tau_{\mathrm{W}}(j_{1} j_{2}j_4) \tau_{\mathrm{W}}(j_{3})
					+ \tau_{\mathrm{W}}(j_{1} j_{3}j_4) \tau_{\mathrm{W}}(j_{2}). \nonumber
\end{align}
Recall the definitions of $\tau_{\mathrm{W}}$ and $\tau_{\mathrm{Z}}$, and compare lines \eqref{recOrd4} and \eqref{planar1}. Note that multiplicativity of $ \tau_{\mathrm{Z}}$ enters at order four in the term $ \tau_{\mathrm{W}}(j_{1} j_{3}) \tau_{\mathrm{Z}}(j_{2})\tau_{\mathrm{Z}}(j_{4})$ corresponding to the tensor product $j_1j_3 \otimes j_2 | j_4 \in T(J) \otimes T_2(T(J))$.

%%%%%%%%%%%%%%%%%%%%%%%%%%%%%%%%%%%%%%%

\subsection{A bialgebra of non-crossing partitions}
\label{subsect:ncpart}

The following sections aim at developing further the  combinatorics underlying planar theories through the notion of non-crossing partitions. Recall that a partition $L$ of a (finite) set $[n]:=\{1,\ldots,n\}$ consists of a collection of (non-empty) subsets $L=\{L_1,\ldots,L_b\}$ of $[n]$, called blocks, which are mutually disjoint, i.e., $L_i \cap L_j=\emptyset$ for all $i\neq j$, and whose union $\cup_{i=1}^b L_i=[n]$ \cite{beissinger,speed}. By $|L|:=b$ the number of blocks of the partition $L$ is denoted, and $|L_i|$ is the number of elements in the $i$th block $L_i$. Given $p,q \in [n]$ we will write that $p \sim_{L} q$ if and only if they belong to the same block. The lattice of set partitions of $[n]$ is denoted by $P_n$. It has a partial order of refinement: $L \leq K$ if $L$ is a finer partition than $K$. The partition $\hat{1}_n = \{L_1\}$ consists of a single block, i.e., $|L_1|=n$, and is the maximum element in $P_n$. The partition $\hat{0}_n=\{L_1,\ldots,L_n\}$ has $n$ singleton blocks, and is the minimum partition in $P_n$. A set partition $L=\{L_1,\ldots,L_k\}$ of $[n]$ ($L_1\coprod \cdots \coprod L_k=[n]$) is called non-crossing if for $p_1,p_2,q_1,q_2 \in [n]$ the following property does not occur
$$
	1\leq p_1<q_1<p_2<q_2\leq n
$$
and 
$$
	p_1\sim_{L}p_2\nsim_{L}q_1\sim_{L}q_2.
$$

The set of non-crossing partitions of $[n]$ will be denoted by $NC_n$, we also set $NC:=\cup_{n\in\Nb}NC_n$. The reader is referred to the standard reference \cite{nicaspeicher} for more details. See also \cite{beissinger,novaksniady}. The common pictorial representation of (non-crossing) partitions is envoked. For example, 
$$
		{\scalebox{0.7}{	\nci\ \,} \atop 1}
		\qquad\ 
		{\scalebox{0.7}{	\ncii\ \,} \atop {1\hspace{0.4cm}\; 2} }
		\qquad\ 
		{\scalebox{0.7}{	\nci \quad \nci \quad \nci\ \,} \atop  1\;\; 2\;\; 3 }
$$
the first represents the singleton $\hat{0}_1=\hat{1}_1=\{1\}$ in $P_1$. The second is the single block partition, i.e., the maximal element $\hat{1}_2=\{1,2\} \in P_2$. Then follows the minimal element in $P_3$, i.e., the partition of the set $[3]$ into singletons, $\hat{0}_3=\{\{1\},\{2\},\{3\}\}$. The partition $\{\{1,3\},\{2, 4\}\}$ is represented 
$$
	\scalebox{0.7}{ \nciiib }
$$ 
and is not a non-crossing partition, whereas  $\{\{1,9\},\{2, 6,8\},\{3,5\},\{4\},\{7\}\}$ and $\{\{1,3,7\},\{2\},$ $\{4,5,6\}\}$ are proper partitions without crossings.

Non-crossing partitions of arbitrary subsets of the integers are defined similarly. For example, $\{\{1,6,10\},\{2\}, \{7,9\}\}$ is a non-crossing partition of $\{1,2,6,7,9,10\}$. We will use implicitly various elementary properties of non-crossing partitions \cite{nicaspeicher}. In particular, we will use the fact that, if $L$ is a non-crossing partition of $[n]$, then its restriction to an arbitrary subset $S$ of $[n]$ (by intersecting the blocks of $L$ with $S$) defines a non-crossing partition of $S$.

In planar quantum field theories the planarity constraint translates into the property that the various propagators joining external sources in the diagrammatic expansion of Green's functions never cross. In particular, given a Feynman graph $\gamma$ in the expansion of a complete $n$-point Green's function $\mathrm{Z}^{(n)}_{j_1,\ldots,j_n}$ with $k$ connected components, the partition of the external sources $j_1,\ldots,j_n$ according to their common belonging  to one connected component of the graph induces a non-crossing partition $L(\gamma)=(L_1,\ldots,L_k)$ of $[n]$. For example, the following Feynman diagram in planar $\Phi^4$ theory 
$$
	\FeynGraphPart
$$
is associated to the non-crossing partition
$$
	\scalebox{0.7}{\nciiic}\scalebox{0.7}{\ncii}
$$

Let us write $\gamma^{(n)}_{j_1 \cdots j_n}$ for the amplitude associated to a given Feynman diagram $\gamma$, normalized by the proper symmetry factor, so that the sum of all amplitudes when $\gamma$ runs over all diagrams with $n$ external legs gives the full Green's function $\mathrm{Z}^{(n)}_{j_1 \cdots j_n}$ (resp.~the connected Green's function $\mathrm{W}^{(n)}_{j_1 \cdots j_n}$ when the sum runs over connected diagrams).

Summing up all these amplitudes $\gamma^{(n)}_{j_1 \cdots j_n}$ of the Feynman diagrams associated to a given non-crossing partition $L=\{L_1,\ldots,L_k\}$ of $[n]$ defines a new Green's function 
$$
	\mathrm{L}^{(n)}_{j_1 \cdots j_n}:=\sum\limits_{L(\gamma)=L}\gamma^{(n)}_{j_1 \cdots j_n}
$$ 
such that, setting for $S=\{s_1,\ldots,s_p\} \subset [n]$, $\mathrm{W}^S:=\mathrm{W}^{(p)}_{j_{s_1} \cdots j_{s_p}}$, one gets 
$$
	\mathrm{L}^{(n)}_{j_1 \cdots j_n} = \prod_{i=1}^k \mathrm{W}^{L_i}.
$$
One step further, full Green's functions split according to non-crossing partitions:
$$
	\mathrm{Z}^{(n)}_{j_1 \cdots j_n}=\sum_{L\in NC_n}\mathrm{L}^{(n)}_{j_1 \cdots j_n}.
$$ 
This is the phenomenon this section aims at investigating from a combinatorial point of view.

Let $L=\{L_1,\ldots ,L_k\}$ be an arbitrary non-crossing partition of $[n]:=\{1,\ldots ,n\}$ with $\inf(L_i)<\inf(L_{i+1})$ for $i=1,\ldots ,n-1$. Let us write $L_i < L_j$ if $\forall a \in L_i$ and $\forall b \in L_j$ we have $a<b$. We define a partial order $<_{L}$ on the blocks $L_i$ as follows: $L_i<_{L} L_j$ if and only if, for all $m \in L_i$, $\inf (L_j) < m < \sup(L_j)$. The very definition of non-crossing partitions shows that this partial order is well-defined. Moreover, given two distinct blocks $L_i, L_j\in L$, then one and only one of the following inequalities holds
$$
	L_i < L_j,\quad L_j < L_i,\quad L_i<_L L_j,\quad L_j<_L L_i.
$$

As an example we consider the particular non-crossing partition $L \in P_{10}$ with five blocks $L=\{L_1,L_2,L_3,L_4,L_5\} =\{ \{1,3,8\},\{2\},\{4,6,7\},\{5\},\{9,10\} \}$
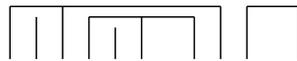
$$
	\scalebox{0.7}{\nciiiiiia}
$$
The block $L_5 > L_i$, $i=1,2,3,4$, and $L_2 <_L L_1$, $L_4 <_L L_3 <_LL_1$.

A partition of the blocks of $L$ into two (possibly empty) subsets
$$
	L = Q \coprod T = \{Q_1,\ldots ,Q_i\} \coprod \{T_1,\ldots ,T_{k-i}\}
$$ 
will be said {\it{admissible}} if and only if for all $p \leq i$, $q \leq k-i$, $Q_p \not<_L T_q$, that is, $T_q <_L Q_p $ or the two subsets of $[n]$ are incomparable for the partial order. We write then $L=Q {\coprod \atop {\rm{adm}}} T$. Admissible partitions of non-crossing partitions of arbitrary finite subsets $S$ of the integers are defined accordingly. Returning to the above example, we have (note that the list is not exhaustive)
$$
	L= \{L_1,L_2,L_3,L_4\} {\coprod \atop {\rm{\smop{adm}}}} \{L_5 \}
	  = \{L_1,L_2,L_5\} {\coprod \atop {\mathrm{\smop{adm}}}} \{L_3,L_4\}
	  = \{L_1,L_5\} {\coprod \atop {\mathrm{\smop{adm}}}} \{L_2, L_3,L_4\}.
$$

%
%Similarly, a partition 
%$$
%	L = Q \coprod T\coprod U
%	    = \{Q_1,\ldots ,Q_i\}\coprod \{T_1,\ldots ,T_{p-i}\}\coprod \{U_1,\ldots ,U_{k-p}\}
%$$
%of $L$ is admissible if and only if for all $v \leq i,\ w\leq p-i,\ z\leq k-p$, $Q_v \not<_L T_w,\ Q_v\not<_L U_z,\ T_w\not<_L U_z$. Notice that there is a bijection between admissible partitions $L=Q\coprod T\coprod U$, pairs of admissible partitions $L=Q \coprod W,\ W=T\coprod U$ and pairs of admissible partitions $L=V\coprod U,\ V=Q\coprod T$. We will refer to this property as the {\it{coassociativity of admissibility}} and write $L=Q\!{\coprod\atop {\rm{adm}}}\!T\!{\coprod\atop {\rm{adm}}}\!U$.

Given two (canonically ordered) subsets $S \subseteq U$ of the set of integers $\bf N$, recall that a connected component of $S$ relative to $U$ is a maximal sequence $s_1, \ldots , s_n$ in $S$, such that there are no $ 1 \leq i < n$ and $t\in U$, such that $s_i < t <s_{i+1}$. In particular, a connected component of $S$ in $\bf N$ is simply a maximal sequence of successive elements $s,s+1,\ldots ,s+n$ in $S$.

For an admissible partition $L=Q {\coprod \atop {\rm{adm}}} U$ as above, we consider the connected components $J_1,\ldots , J_{k(L ,Q)}$ of $[n]-(Q_1\cup \cdots \cup Q_i)$, that we will call slightly abusively from now on the {\it{connected components}} of $[n]-Q$ . The definition of $<_L$ implies that $J_i\cap U_j$ is empty or equals $U_j$. We write $J_i^{L,Q}$ for the set of all non-empty intersections $J_i \cap U_j,\ j=1, \ldots , k-i$ and notice that, since $L$ is a non-crossing partition of $[n]$, $J_i^{L,Q}$ is, by restriction, a non-crossing partition of the component $J_i$. For the same reason, $Q$ is a non-crossing partition of $Q_1\cup \cdots \cup Q_i$.

Let us recall now that, given a finite subset $S$ of cardinality $n$ of the integers, the standardization map $st$ is the (necessarily unique) increasing bijection between $S$ and $[n]$. By extension, we write also $st$ for the induced map on the various objects associated to $[n]$ (such as partitions). For example, the standardization of the non-crossing partition $L:=\{\{3,6,10\},\{4,5\},\{8\}\}$ of the set $\{3,4,5,6,8,10\}$ is the non-crossing partition $st(L):=\{\{1,4,6\},\{2,3\},\{5\}\}$ of $[6]=st(\{3,4,5,6,8,10\})$.

The linear span $\mathcal{NC}$ of all non-crossing partitions can then be equipped with a coproduct map $\Delta$ from $\mathcal{NC}$ to $\mathcal{NC}\otimes T(\mathcal{NC})$ defined by (using our previous notations as well as the bar-$|$ notation for elements in $T(\mathcal{NC})$)
$$
	\Delta(L)=\sum\limits_{Q\!\coprod\limits_{\rm{adm}}\!U = L} 
	st(Q) \otimes \big(st(J_1^{L,Q})| \cdots | st(J_{k(L, Q)}^{L,Q})\big).
$$
A few examples are in order at this stage. 
%\{\{1,4\},\{2, 3\}\}  \{\{1,5\},\{2\}, \{3,4\}\}
$$
	\Delta(\{\{1,4\},\{2, 3\}\}) = \{\{1,4\},\{2, 3\}\} \otimes \un + \un \otimes \{\{1,4\},\{2, 3\}\} 
						+ \{1,2\} \otimes \{1,2\}
$$
\begin{eqnarray*}
	\Delta(\{\{1,5\},\{2\}, \{3,4\}\}) &=&\{\{1,5\},\{2\}, \{3,4\}\} \otimes \un + \un \otimes \{\{1,5\},\{2\}, \{3,4\}\}\\
						& &
						+  \{\{1,3\},\{2\}\} \otimes \{1,2\} 
						+ \{\{1,4\},\{2, 3\}\} \otimes \{1\}\\
						& &
						+ \{1,2\} \otimes \{\{1\}, \{2,3\}\} 
\end{eqnarray*}
\begin{eqnarray*}
	\Delta(\{\{1,2\}, \{3\},\{4\}\}) &=& \{\{1,2\}, \{3\},\{4\}\} \otimes \un + \un \otimes \{\{1,2\}, \{3\},\{4\}\}
						+ \{1,2\} \otimes \{\{1\}\{2\}\}\\
					       & & + 2 \{\{1,2\},\{3\}\} \otimes \{1\} 
					       		+ \{1\}  \otimes \{\{1,2\},\{3\}\}
							+ \{1\}  \otimes \{1,2\} | \{1\}\\
						& & 	 \{\{1\},\{2\}\}  \otimes \{1,2\}
\end{eqnarray*}
The graphical notation of the coproduct simplifies the calculus since it is automatically standardized (note that the bar of the bar notation is written in bold to distinguish it from the single element partition)
$$
	\Delta(\scalebox{0.5}{\nci}\ ) = \scalebox{0.5}{\nci} \otimes \un + \un \otimes \scalebox{0.5}{\nci} 
	\qquad\
	\Delta(\scalebox{0.5}{\ncii}\ ) = \scalebox{0.5}{\ncii} \otimes \un + \un \otimes \scalebox{0.5}{\ncii}
$$
$$
	\Delta(\scalebox{0.5}{\nciiia}\ ) = \scalebox{0.5}{\nciiia} \otimes \un + \un \otimes \scalebox{0.5}{\nciiia} 
							\ +\  \scalebox{0.5}{\ncii} \otimes \scalebox{0.5}{\nci}
$$
$$
	\Delta(\scalebox{0.5}{\nciiic}\ ) =  \scalebox{0.5}{\nciiic} \otimes \un + \un \otimes \scalebox{0.5}{\nciiic} 
							\ +\  \scalebox{0.5}{\ncii} \otimes \scalebox{0.5}{\ncii}
$$
\begin{eqnarray*}
	\Delta(\scalebox{0.5}{\nciiiiia}\ ) &=& \scalebox{0.5}{\nciiiiia} \otimes \un + \un \otimes \scalebox{0.5}{\nciiiiia}
								\ + \ \scalebox{0.5}{\nciiia} \otimes \scalebox{0.5}{\ncii} 
								\ + \ \scalebox{0.5}{\nciiic} \otimes \scalebox{0.5}{\nci}
								\ +\ \scalebox{0.5}{\ncii} \otimes \scalebox{0.5}{\nci \ncii}
\end{eqnarray*}
\begin{eqnarray*}
	\Delta(\scalebox{0.5}{\ncii \nci \nci}\ ) &=& \scalebox{0.5}{\ncii \nci \nci} \otimes \un 
									+ \un \otimes \scalebox{0.5}{\ncii \nci \nci}
								\ + \ \scalebox{0.5}{\ncii} \otimes \scalebox{0.5}{\nci \nci} 
								\ + \ \scalebox{0.5}{\nci} \otimes \scalebox{0.5}{\ncii \nci}
								\ + \ \scalebox{0.5}{\nci} \otimes \scalebox{0.5}{\ncii} \;\; \baarr\ \scalebox{0.5}{\nci}
								\ +\ 2 \scalebox{0.5}{\ncii \nci} \otimes \scalebox{0.5}{\nci}
								\ +\  \scalebox{0.5}{\nci \nci} \otimes \scalebox{0.5}{\ncii}
\end{eqnarray*}

The map $\Delta$ is then extended multiplicatively to a coproduct on $\bar T(\mathcal{NC})$ 
$$
	\Delta (L_1|\cdots |L_n):=\Delta(L_1)\cdots \Delta(L_n), \quad \Delta(\un) = \un \otimes \un.
$$
Here, $\bar{T}(\mathcal{NC})$ is equipped with the structure of a free associative algebra over $\mathcal{NC}$ by the concatenation map, $(L_1|\cdots |L_k)\cdot (L_{k+1}|\cdots |L_n):=(L_{1}|\cdots |L_k | L_{k+1}| \cdots |L_n)$.

\begin{thm} 
\label{thm:HANC}\cite{EP2}
The graded algebra $\bar T(\mathcal{NC})$ equipped with the coproduct $\Delta$ is a connected graded non-commutative and non-cocommutative Hopf algebra. 
\end{thm}

This coproduct can be split into two parts as follows. On $\mathcal{NC}$ define the {\it{left half-coproduct}} by 
\begin{equation}
\label{HANCprec+}
	\Delta_{\prec}(L)=\sum\limits_{Q\!\coprod\limits_{\rm{adm}}\!U=L \atop 1 \in Q_1}
					st(Q)\otimes \big(st(J_1^{L,Q})| \cdots |st(J_{k(L , Q)}^{L,Q})\big),
\end{equation}
and
\begin{equation}
\label{HANCprec}
	\bar\Delta_{\prec}(L) := \Delta_{\prec}(L) - L \otimes \un. 
\end{equation}
The {\it{right half-coproduct}} is defined by
\begin{equation}
\label{HANCsucc+}
	\Delta_{\succ}(L)=\sum\limits_{Q\!\coprod\limits_{\rm{adm}}\!U=L \atop 1\notin Q_1}
					st(Q)\otimes \big(st(J_1^{L,Q})| \cdots |st(J_{k(L , Q)}^{L,Q})\big),
\end{equation}
and
\begin{equation}
\label{HANCsucc}
	\bar\Delta_{\succ}(L) := \Delta_{\succ}(L) -  \un \otimes L .
\end{equation}
Which yields $\Delta = \Delta_{\prec} + \Delta_{\succ}$, and for $L \in \mathcal{NC}$
$$
	\Delta(L) = \bar\Delta_{\prec}(L) + \bar\Delta_{\succ}(L) + L \otimes \un + \un \otimes  L.
$$
This is extended to $\bar T(\mathcal{NC})$ by defining
\begin{eqnarray}
\label{extens}
	\Delta_{\prec}(L_1 | \cdots | L_m) &:=& \Delta_{\prec}(L_1)\Delta(L_2) \cdots \Delta(L_m) \\\label{extens2}
	\Delta_{\succ}(L_1 | \cdots | L_m) &:=& \Delta_{\succ}(L_1)\Delta(L_2) \cdots \Delta(L_m). 
\end{eqnarray}

\begin{thm} \label{thm:bialgNC}\cite{EP2}
The bialgebra $\bar T(\mathcal{NC})$ equipped with $\Delta_{\succ}$ and $\Delta_{\prec}$ is an unshuffle bialgebra. 
\end{thm}

%%%%%%%%%%%%%%%%%%%%%%%%%%%%%%%%%%%%%%%

\subsection{Non-crossing Green's functions}
\label{sect:ncGreen}

We have introduced, at the begining of this section, the notion of Green's functions parametrized by non-crossing partitions,  $\mathrm{L}^{(n)}_{ j_1 \ldots j_n}$. We will call them simply non-crossing Green's functions. These Green's functions are parametrized by the pair of a non-crossing partition of $[n]$ and $n$ external sources, denoted from now on by $L( j_1\cdots  j_n)$ and called decorated non-crossing partitions. The latter are represented graphically by adding decorations to the graphical representations of non-crossing partitions. For example, for $L=\{\{1,3\},\{2,4\}\}$, $L(j_1j_2j_3j_4)$ is represented by
 $$
 	\begin{array}{c}
			\scalebox{0.5}{\nciiic}\\[-0.25cm] 
			{\hspace{0.1cm}\scriptstyle{j_1 j_2 \ j_3  j_4}}
	\end{array}.
$$
 
The present section shows how to mimick, for these data, the analog of the calculus that has been developed previously in order to understand from a Hopf and group-theoretical point of view the link between full and connected Green's functions.
 
For a set $J=\{j_i\}_{i>0}$ of non-commuting external sources, let us write ${\mathcal{NC}}(J)$ for the graded vector space spanned by $J$-decorated (or simply decorated) non-crossing partitions. For $S=\{s_1,\ldots,s_k\}\subset[n]$ and $j:=j_{i_1}\cdots j_{i_n}$, we set $j_S:=j_{i_{s_1}} \cdots  j_{i_{s_k}}$. Similarly, for $Q=\{Q_1,\ldots,Q_i\}$ a non-crossing partition of the set $S$, we write $j_Q:=j_S$.

The definitions and results in the previous paragraph carry over to decorated non-crossing partitions in a straightforward way. For example, the coproduct map $\Delta$ is defined on ${\mathcal{NC}}(J)$ by
\begin{eqnarray*}
	\Delta(L(j_{i_1}\cdots j_{i_n}))=
	\sum\limits_{Q\!\coprod\limits_{\rm{adm}}\!U=L}
	\big(st(Q)\otimes j_Q\big) \otimes \big(st(J_1^{L,Q})\otimes (j_{J_1^{L,Q}})| \cdots |st(J_{k(L, Q)}^{L,Q})
	\otimes (j_{J_{k(L, Q)}^{L,Q}})\big),
\end{eqnarray*}
where $L$ is a non-crossing partition of $[n]$. As an example we calculate
$$
	\Delta\Big(\!\!\!	\begin{array}{c}
				\scalebox{0.5}{\nciiic}\\[-0.25cm] 
				{\hspace{0.1cm}\scriptstyle{j_1 j_2 \ j_3  j_4}}
				\end{array}\!\!\!\Big) 
			=  	\begin{array}{c}
				\scalebox{0.5}{\nciiic}\\[-0.25cm] 
				{\hspace{0.1cm}\scriptstyle{j_1 j_2 \ j_3  j_4}}
				\end{array} \!\!\otimes \un 
				+ \un \otimes\!\! 
				\begin{array}{c}
				\scalebox{0.5}{\nciiic}\\[-0.25cm] 
				{\hspace{0.1cm}\scriptstyle{j_1 j_2 \ j_3  j_4}}
				\end{array} 
							+
				\begin{array}{c}
				\scalebox{0.5}{\ncii}\\[-0.25cm] 
				{\hspace{0.1cm}\scriptstyle{j_1 \ j_4}}
				\end{array} 
				\!\! \otimes\!\! 
				\begin{array}{c}
				\scalebox{0.5}{\ncii}\\[-0.25cm] 
				{\hspace{0.1cm}\scriptstyle{j_2 \ j_3}}
				\end{array}
$$
It is then extended to $T({\mathcal{NC}}(J))$ multiplicatively as in the previous sections; the other structural maps on $T({\mathcal{NC}})$ are extended similarly to $T({\mathcal{NC}}(J))$, and are written using the same symbols as before. We obtain finally

\begin{thm} \label{thm:bialgNCA}
The bialgebra $\bar T(\mathcal{NC}(J))$ equipped with $\Delta_{\succ}$ and $\Delta_{\prec}$ is an unshuffle bialgebra. 
\end{thm}

Since $\bar T(\mathcal{NC}(J))$ is in particular a Hopf algebra, the set of linear maps, $Lin(\bar{T}(\mathcal{NC}(J)),\Kb)$, is a $\mathbb{K}$-algebra with respect to the convolution product, which is defined in terms of the coproduct $\Delta$, i.e., for $f,g\in Lin(\bar{T}(\mathcal{NC}(J)),\Kb)$
$$ 
	f \star g := m_\mathbb{K}\circ (f\otimes g)\circ \Delta,
$$ 
where $m_\mathbb{K}$ stands for the product map in $\mathbb{K}$. Notice that, motivated by the next proposition, we will also use later a shuffle notation for this product: $f\star g=:f\shuffle g$. We define accordingly the left and right convolution half-products on $Lin(T(\mathcal{NC}(J)),\Kb)$:
$$
	f\prec g:=m_\mathbb{K}\circ (f\otimes g)\circ \Delta_\prec ,
$$
$$
	f\succ g:=m_\mathbb{K}\circ (f\otimes g)\circ \Delta_\succ .
$$

\begin{prop}
The space $T^\ast(\mathcal{NC}(J)):=Lin(T(\mathcal{NC}(J)),\mathbb{K})$ equipped with $(\prec, \succ)$ is a shuffle algebra.
\end{prop}

For completeness, and in view of the importance of this proposition for forthcoming developments, we recall briefly its proof: for arbitrary $f,g,h\in T^\ast(\mathcal{NC}(J))$,
$$
	(f\prec g)\prec h	
			=m_\mathbb{K}\circ ((f\prec g)\otimes h)\circ\Delta_\prec 
			=m_\mathbb{K}^{[3]}\circ (f\otimes g\otimes h)\circ(\Delta_\prec\otimes I)\circ \Delta_\prec,
$$
where $m_\mathbb{K}^{[3]}$ stands for the product map from $\mathbb{K}^{\otimes 3}$ to $\mathbb{K}$. Similarly 
\begin{eqnarray*}
	f\prec (g\shuffle h)
			&=&m_\mathbb{K}\circ (f\otimes (g\shuffle h))\circ \Delta_\prec \\
			&=&m_\mathbb{K}^{[3]}\circ (f\otimes g\otimes h)\circ (I\otimes \overline\Delta)\circ\Delta_\prec,
\end{eqnarray*}
so that the identity $(f\prec g)\prec h=f\prec (g\shuffle h)$ follows from $(\Delta_\prec\otimes I)\otimes\Delta_\prec =(I\otimes \overline\Delta)\circ\Delta_\prec$, and similarly for the other identities characterizing shuffle algebras.

As usual, we equip the shuffle algebras $T^\ast(\mathcal{NC}(J))$ with a unit. That is, we set $\bar{T}^\ast(\mathcal{NC}(J)):=T^\ast(\mathcal{NC}(J))\oplus \mathbb{K} \un\cong Lin(\bar T(\mathcal{NC}(J)),\mathbb{K})$, where in the last isomorphism the unit $\un \in \bar{T}^\ast(\mathcal{NC}(J))$ is identified with the augmentation map $\varepsilon \in Lin(\bar T(\mathcal{NC}(J)),\mathbb{K})$ -- the null map on $T(\mathcal{NC}(J))$, and the identity map on $\mathcal{NC}(J)^{\otimes 0}\cong \mathbb{K}$. Moreover, for an arbitrary $f$ in ${T}^\ast(\mathcal{NC}(J))$,
$$
	f\prec \varepsilon=f=\varepsilon\succ f,\ \ \varepsilon\prec f=0=f\succ \varepsilon.
$$

Let now $\phi$ be a linear form on $\mathcal{NC}(J)$, for example the one written $\tau_{\mathrm{nc}}$ associated with non-crossing Green's functions, and defined by
$$
	\tau_{\mathrm{nc}}(L(j_{i_1}\cdots j_{i_n})):=\mathrm{L}^{(n)}_{j_{i_1} \cdots  j_{i_n}}.
$$
It extends uniquely to a multiplicative linear form $\Phi$ on $T(\mathcal{NC}(J))$ (still written $\tau_{\mathrm{nc}}$ when $\phi=\tau_{\mathrm{nc}}$) by setting
$$
	\Phi(w_1| \cdots |w_n):=\phi(w_1) \cdots \phi(w_n),
$$
(and to a unital and multiplicative linear form on $\bar T(\mathcal{NC}(J))$ by setting $\Phi(\un):=1$). Conversely any such multiplicative map $\Phi$ gives rise to a linear form on $\mathcal{NC}(J)$ by restriction of its domain.

\begin{defn}
A linear form $\Phi \in \bar T(\mathcal{NC}(J))$ is called a character if it is unital, $\Phi(\un)=1$, and multiplicative, i.e., for all $a,b \in T(\mathcal{NC}(J))$, $\Phi(a|b)=\Phi(a)\Phi(b).$ A linear form $\kappa\in \bar T(\mathcal{NC}(J))$ is called an infinitesimal character, if $\kappa(\un)=0$, and if for all $a,b\in T(\mathcal{NC}(J))$, $\kappa(a|b)=0.$
\end{defn}

We write ${G}_{NC}$ for the set of characters in $\bar T(\mathcal{NC}(J))$ and ${g}_{NC}$ for the corresponding set of infinitesimal characters. Notice that, by its very definition, an infinitesimal character is entirely determined by its restriction to $\mathcal{NC}(J)\subset T(\mathcal{NC}(J))$.

%We also recall for later use that the same notions hold in $\bar{{\mathcal L}(J)}$, i.e., a linear form $\phi$ on $T(J)$ extends multiplicatively to a linear form on $\bar{T}(T(J))$; a linear form $\Phi \in \overline{{\mathcal L}(J)}$ is called a character if it is unital and multiplicative; a linear form $\kappa\in \overline{{\mathcal L}(J)}$ is called an infinitesimal character, if $\kappa(\un)=0$, and $\kappa(a|b)=0$ for all $a,b\in T(T(J))$. We write ${\mathcal G}(J)$ for the set of characters in $\overline{{\mathcal L}(J)}$ and $g(L)$ for the corresponding set of infinitesimal characters.

\begin{thm}\label{thm:Gg2}\cite{EP2}
There exists a natural bijection between ${G}_{NC}$, the set of characters, and $g_{NC}$, the set of infinitesimal characters on $\bar T(\mathcal{NC}(J))$. More precisely, for $\Phi \in G_{NC}, \exists ! \kappa  \in g_{NC}$ such that $\Phi=\varepsilon+\kappa\prec \Phi ,$ and conversely, for $\kappa\in g_{NC}(J)$
$$
	\Phi:=\varepsilon+\kappa +\kappa\prec \kappa +\kappa\prec (\kappa\prec \kappa)+
	\kappa\prec (\kappa\prec (\kappa\prec \kappa))+\cdots=:
	\exp^{\prec}( \kappa)
$$
is a character.
\end{thm}

Recall that the analogous Theorem holds for the group $G_{\mathbb C}$ and Lie algebra $g_{\mathbb C}$. 

In view of the importance of the Theorem for our forthcoming developments, we include a sketch of the proof. Since we are dealing with graded structures, the implicit equation $\Phi= \varepsilon+\kappa\prec \Phi$ can be shown recursively to have a unique solution $\kappa$ in $\bar T^\ast(\mathcal{NC}(J))$. Let us consider the infinitesimal character $\mu$ equal to $\kappa$ on the component $\mathcal{NC}(J)$ of $T(\mathcal{NC}(J))$, and let us show that $\mu$ also solves $\Phi=e+\mu\prec \Phi$; the first part of the Theorem will follow.

Indeed, for an arbitrary $w=w_1|\cdots|w_n \in T(\mathcal{NC}(J))$, we have: 
$$
	\Phi(w_1)	=(e+\kappa\prec \Phi)(w_1)
			=\kappa(w_1^{1,\prec})\Phi(w_1^{2,\prec})
			=\mu(w_1^{1,\prec})\Phi(w_1^{2,\prec}),
$$
so that for any $n > 1$
$$
	\Phi (w_1|\cdots|w_n)	=\Phi(w_1)\Phi(w_2|\cdots|w_n)
					=\mu(w_1^{1,\prec})\Phi(w_1^{2,\prec}|w_2|\cdots|w_n))
					=(e+\mu\prec \Phi)(w_1|\cdots|w_n),
$$
from which the property follows.
Conversely, 
$$
	\exp^{\prec}(\kappa)(w_1|\cdots|w_n)=(e+\kappa\prec\exp^{\prec}(\kappa))(w_1|\cdots|w_n)
							=\kappa(w_1^{1,\prec})\exp^{\prec}(\kappa)(w_1^{2,\prec}|\cdots|w_n).
$$
Assuming (by induction on the total tensor degree of the expressions) that $\exp^\prec$ acts as a character on $(w_1^{2,\prec}|\cdots|w_n)$, we obtain:
\begin{eqnarray*}
	\exp^{\prec}(\kappa)(w_1|\cdots|w_n)
				&=&(\kappa(w_1^{1,\prec})\exp^{\prec}(\kappa)(w_1^{2,\prec}))
								\exp^{\prec}(\kappa)(w_2)\cdots\exp^{\prec}(\kappa)(w_n)\\					
				&=&\exp^{\prec}(\kappa)(w_1)\exp^{\prec}(\kappa)(w_2) \cdots \exp^{\prec}(\kappa)(w_n).
\end{eqnarray*}

%%%%%%%%%%%%%%%%%%%%%%%%%%%%%%%%%%%%%%%

\subsection{From full to non-crossing Green's functions}
\label{subsect:fulltonc}

We turn now to the relationship between planar full and connected Green's functions viewed through the prism of non-crossing partitions.

\begin{defn}
The splitting map $Sp$ is the map from $T(J)$ to $\mathcal{NC}(J)$ defined by:
$$
	Sp(j_1 \cdots j_n):=\sum\limits_{L \in {NC}_n}L(j_1 \cdots j_n).
$$
It is extended multiplicatively to a unital map $Sp$ from $\bar T(T(J))$ to $\bar T(\mathcal{NC}(J))$, i.e., for $x_1,\ldots ,x_k\in T(J)$,
$$
	Sp(x_1|\cdots |x_k):=Sp(x_1)|\cdots |Sp(x_k).
$$
\end{defn} 

The name ``splitting map" is chosen because, on dual spaces it permits to ``split'' the value of a linear form, $\phi$, on $T(J)$ (typically the full Green's function), into a sum of terms indexed by non-crossing partitions (typically, the non-crossing Green's functions).

\begin{thm} \label{thm:splitHom}\cite{EP2}
The map $Sp$ from $\bar T(T(J))$ to $\bar T(\mathcal{NC}(J))$ is an unshuffle bialgebra morphism.
\end{thm}

The previous constructions dualize, that is, the linear dual of an unshuffle coalgebra is a shuffle algebra, and a morphism $f$ between two unshuffle coalgebras induces a morphism of shuffle algebras written $f^\ast$ between the linear duals. 
In particular:

\begin{lem} \label{thm:splitHomdual}
The map $Sp$ from $T(T(J))$ to $T(\mathcal{NC}(J))$ induces a morphism of shuffle algebras with units, $Sp^\ast$ from $\bar T^\ast(J)$ to $\bar T^\ast(T(J))$.
\end{lem}

\begin{lem}\label{lemMap}The map $Sp^\ast$ restricts to maps from $G_{NC}$ to $G_{\mathbb C}$ and from $g_{NC}$ to $g_{\mathbb C}$.
\end{lem}

Indeed, the map $Sp$ from $T(T(J))$ to $T({\mathcal NC}(J))$ is induced multiplicatively by a map from $T(J)$ to ${\mathcal NC}(J)$. With the same notation as at the begining of this section, we find that $Sp(x_1|\cdots|x_n)=Sp(x_1)|\cdots|Sp(x_n)$. Therefore, for $\Phi\in  G_{NC}$ we have that
\begin{eqnarray*}
	Sp^\ast(\Phi)(x_1|\cdots|x_n) &=& \Phi\circ Sp(x_1|\cdots|x_n)\\
						    &=& \Phi(Sp(x_1)|\cdots|Sp(x_n))\\
						    &=& \Phi\circ Sp(x_1)\cdots\Phi\circ Sp(x_n)\\
						    &=& Sp^\ast(\Phi)(x_1)\cdots Sp^\ast(\Phi)(x_n),
\end{eqnarray*}
and $Sp^\ast(\Phi)\in G_{\mathbb C}$. A similar argument holds for  infinitesimal characters.

\ \par

 Notice that elements in  $G_{NC}$ and $g_{NC}$ are entirely characterized by their restrictions to $\mathcal{NC}(J)$; similarly elements in $G_{\mathbb C}$ and $g_{\mathbb C}$ are characterized by their restrictions to $T(J)$. It follows that any section $\sigma$ of the map $Sp^\ast$ from $Lin(\mathcal{NC}(J),\mathbb{K})$ to $Lin(T(J),\mathbb{K})$ induces a  right inverse to the map $Sp^\ast$ from  ${G}_{NC}$ (resp. $g_{NC}$) to $G_{\mathbb C}$ (resp. $g_{\mathbb C}$). The existence of such sections, and the surjectivity of $Sp^\ast$ follow by direct inspection. 

\begin{thm}\label{Thm:commut}
The isomorphism of Theorem \ref{thm:Gg2} commutes with $Sp^\ast$, in the sense that, given $\Phi$ and $\kappa$
$$
	Sp^\ast(\Phi)=\varepsilon +Sp^\ast(\kappa)\prec Sp^\ast(\Phi) = \exp^{\prec}( Sp^\ast(\kappa)).
$$
\end{thm}

The Theorem follows from Lemma \ref{thm:splitHomdual} and Lemma \ref{lemMap}.

\ \par

In view of Theorems \ref{Thm:commut}, \ref{thm:Gg2}, \ref{thm:Gg}, and of the previous results in this section, one can therefore use any section $\sigma$ of the map $Sp^\ast$ from $Lin(\mathcal{NC}(J),\mathbb{K})$ to $Lin(T(J),\mathbb{K})$ to lift the equation $\tau_\mathrm{Z}=\varepsilon+\tau_\mathrm{W}\prec \tau_\mathrm{Z}$ relating full and connected Green's functions to non-crossing partitions, that is, to $Lin(\mathcal{NC}(J),\mathbb{K})$. 

We introduce for that reason a so-called ``standard section''. In free probabilities, this leads to a new presentation of the classical M\"obius-inversion type relations between free moments and free cumulants \cite{EP2}.  Other choices of sections are theoretically possible, that would lead to other lifts to non-crossing partitions of the relations between planar full and connected Green's functions. However, in planar QFT, the choice of a section is governed by physics and seems to be essentially unique from this point of view: other sections than the ``standard section'' to be introduced now would be probably pointless. 

\begin{defn}
Let $\kappa$ be a unital map from $T(J)$ to $\mathbb{K}$. We call the linear form $sd(\kappa)$  on $\mathcal{NC}(J)$ the standard section of $\kappa$, defined by
$$
	sd(\kappa)(L (j_{i_1}\cdots j_{i_n})):=\kappa(j_{i_1}\cdots j_{i_n}),
$$
if $L$ is the trivial non-crossing partition ($L =[n]$), and zero else.
\end{defn}

For an arbitrary non-crossing partition $L=\{L_1,\ldots ,L_k\}$ of $[n]$, we write $\kappa^L(j_{i_1}\cdots j_{i_n}):=\prod_{i=1}^k\kappa(j_{L_i})$. We write $\phi$ for the solution to $\phi=\varepsilon +\kappa\prec\phi$ with the notation of Theorem \ref{thm:Gg}.

\begin{prop}\label{keyrell}
The solution $\Psi$ of the equation $\Psi=\varepsilon +sd(\kappa)\prec\Psi$ in $Lin(\bar T(\mathcal{NC}(J)),\mathbb{K})$ satisfies the identity:
\begin{equation}
\label{principEq}
	\Psi(Sp(j_{i_1}\cdots j_{i_n}))=\sum\limits_{L\in NC_n}\kappa^L(j_{i_1}\cdots j_{i_n})=\phi(j_{i_1}\cdots j_{i_n}),
\end{equation}
More precisely, for an arbitrary non-crossing partition $L\in NC_n$ we have that
\begin{equation}
\label{fundEqn}
\Psi(L(j_{i_1}\cdots j_{i_n}))=\kappa^L(j_{i_1}\cdots j_{i_n}).
\end{equation}
\end{prop}

The identity $\Psi(Sp(j_{i_1}\cdots j_{i_n}))=\phi(j_{i_1}\cdots j_{i_n})$ follows from Theorem \ref{Thm:commut} since 
$$
	\Psi\circ Sp=Sp^\ast(\Psi)=\varepsilon +Sp^\ast\circ sd(\kappa)\prec Sp^\ast(\Psi)=\varepsilon +\kappa\prec Sp^\ast(\Psi),
$$
from which we obtain (by unicity of the solution) that $Sp^\ast(\Psi)=\phi$ on $T(L)$.

The second identity, from which $\Psi(Sp(j_{i_1}\cdots j_{i_n}))=\sum\limits_{L\in NC_n}\kappa^L(j_{i_1}\cdots j_{i_n})$ is deduced, follows by induction on $[n]$. Let us assume that identity (\ref{fundEqn}) holds for non-crossing partitions of $[p]$, $p<n$. We get:
$$
	\Psi(L(j_{i_1}\cdots j_{i_n}))
	    = \sum\limits_{Q\!\coprod\limits_{\rm{adm}}\!P=L}
	    sd(\kappa)
	    \big(st(Q)(j_Q)\big)\Psi\big(st(J_1^{L,Q})(j_{J_1^{L,Q}})|\cdots |st(J_{k(L,Q}^{L,Q})(j_{J_{k(L,Q)}^{L,Q}})\big).
$$
However, since $sd(\kappa)$ vanishes on all non-crossing partitions except the trivial ones, terms on the right hand-side vanish except when $Q$ is the component of $L$ containing 1, and the expression reduces finally to
\begin{eqnarray*}
	\Psi(L(j_{i_1}\cdots j_{i_n})
		&=& sd(\kappa)(st(L_1)(j_{L_1})\Psi(st(L_2)(j_{L_2})| \cdots | st(L_k)(j_{L_k}))\\
		&=& \kappa( j_{L_1})\Psi(st(L_2)(j_{L_2}| \cdots | st(L_k)(j_{L_k})),
\end{eqnarray*}
where $L=\{L_1,\ldots ,L_k\}$. From the induction hypothesis and the multiplicativity of $\Psi$, we get the expected identity
$$
	\Psi(L(j_{i_1}\cdots j_{i_n})=\prod\limits_{i=1}^k\kappa(j_{L_i}).
$$

\begin{cor}
Let $\kappa$ be the infinitesimal character on $T(J)$ defined by the connected Green's function
$$
	\kappa(j_{i_1} \cdots j_{i_n}):=\mathrm{W}^{(n)}_{j_{i_1} \cdots j_{i_n}}.
$$
We have: $sd(\kappa)(L(j_{i_1}\cdots j_{i_n}))=\mathrm{W}^{(n)}_{j_{i_1} \cdots j_{i_n}}$ if $L=[n]$ and is equal to $0$ else. Moreover,
the solution to $\Psi=\varepsilon +sd(\kappa)\prec\Psi$ is the non-crossing Green's function, that is
$$
	\Psi(L(j_{i_1}\cdots j_{i_n}))=\mathrm{L}^{(n)}_{j_{i_1}\cdots j_{i_n}}.
$$
Besides, we recover from the properties of the splitting map $Sp$ the relation linking full and non-crossing Green's functions
$$
	\mathrm{Z}^{(n)}_{j_{i_1} \cdots j_{i_n}}
				=\Psi(Sp(j_{i_1}\cdots j_{i_n}))
				=\sum\limits_{L\in NC_n}\mathrm{L}^{(n)}_{j_{i_1} \cdots  j_{i_n}}.
$$
\end{cor}

%%%%%%%%%%%%%%%%%%%%%%%%%%%%%%%%%%%%%%%

\subsection{From full to connected (non-planar)}
\label{non-planar-case}

Recall that in the non-planar case, the classical relation between full and connected Green's functions reads
\begin{equation}
\label{lpfc}
	Z(j)=\exp\big(W(j)\big).
 \end{equation}
By functional derivation we therefore obtain
$$
	\frac\partial{\partial j_1}Z(j)=(\frac\partial{\partial j_1}W(j))Z(j).
$$
Since the underlying algebra is commutative, the usual Leibniz rule of differential calculus applies and we get
\begin{align*} 
	Z^{(n)}_{j_{i_1} \cdots j_{i_n}} 	&= \left. \frac{\partial^n}{\partial_{j_{i_n}} \cdots \partial_{j_{i_1}}}\right |_{j=0} \exp\big( W(j) \big) \\
							&= \left. \frac{\partial^{n-1}}{\partial_{j_{i_n}} \cdots \partial_{j_{i_{2}}}} \Big( (\partial_{j_{i_1}} W(j)) Z[j]  \Big)\right|_{j=0}\\
							&=\sum_{S=\{1,s_2\ldots,s_k\}\subset [n]}
								W^{(k)}_{j_{1}j_{i_{s_2}}\cdots j_{i_{s_k}}}
								Z^{(n - k)}_{j_{i_{2}} \cdots j_{i_{s_2 - 1}} j_{i_{s_2 + 1}} 
								\cdots  j_{i_{s_k - 1}} j_{i_{s_k + 1}} \cdots j_{i_{n}}}\\
							&= \sum\limits_{S=\{1=s_1,\ldots,s_k\}\subset [n]}
								W^{(k)}_{j_{i_{s_1}}\cdots j_{i_{s_k}}}Z^{(n-k)}_{j_{i_{l_1}}\cdots j_{i_{l_{n-k}}}},
\end{align*}
where $L=\{l_1,\ldots,l_{n-k}\}$ and $L\coprod S=[n]$. 

We will now show that the relation between non-planar full and connected Green's functions is naturally described in the context of the cocommutative unshuffle bialgebra $\bar T(J)$ with coproduct \eqref{deshufsplit}. Indeed, cocommutativity $\Delta^{\!\!\sh}_{\succ} = \tau\circ \Delta^{\!\!\sh}_{\prec}$ implies that, for arbitrary $\alpha, \beta \in T^*(J)$, we have
$$
	\alpha \prec \beta = \alpha \succ \beta,
$$
so that $\bar T^*(J)=Lin(\bar T(J),{\mathbb C})$ is a unital commutative shuffle algebra for the left and right half-convolution products $\prec$ and $\succ$ defined in \eqref{deshufleft} respectively \eqref{deshufright}. Let  now $\phi : \bar T(J) \to \mathbb{C}$ be a unital map in $\bar T^*(J)$, and consider the linear fixed point equation
\begin{equation}
\label{nonplanarGF}
	\phi = \iota + \tau \prec \phi.  
\end{equation}
Here, the map $\iota : \bar T(J) \to \mathbb{C}$ is the identity on ${\mathcal J}^{\otimes 0}$, and the null map on higher tensor powers of $\mathcal J$. Let $w=j_i \in J \hookrightarrow T(J)$ be a single letter different from the empty word. Then $\Delta^{\!\!\sh}_{\prec}(j_i)=j_i \otimes \un$. With $\phi(\un)=1$ we find
$$
	\phi(j_1) = \tau (j_1).
$$
Next we look at the word $w=j_1j_2 \in T_2(J)$. Recall that the left half-unshuffle $\Delta^{\!\!\sh}_{\prec}(j_1j_2)=j_1j_2 \otimes \un + j_1 \otimes j_2$, such that
$$
	\phi(j_1j_2) = \tau(j_1j_2) + \tau(j_1)\tau(j_2).
$$
For the order three word $w=j_1j_2j_3 \in T_3(J)$ the left half-unshuffle yields
$$
	\phi(j_1j_2j_3) = \tau(j_1j_2j_3) 
						+ \tau(j_1)\tau(j_2j_3)
						+ \tau(j_1j_3)\tau(j_2)  
						+ \tau(j_1j_2)\tau(j_3) 
						+ \tau(j_1)\tau(j_2) \tau(j_3).
$$
And for the order four word $w=j_1j_2j_3j_4 \in T_4(J)$, the left half-unshuffle gives 
\begin{align*}
	\phi(j_1j_2j_3j_4) &=  \tau( j_{1} j_{2} j_{3}j_{4}) 
					+ \tau(j_{1} j_{2}) \phi(j_{3}j_{4})
					+ \tau(j_{1} j_{3}) \phi(j_{2}j_{4})
					+ \tau(j_{1} j_{4}) \phi(j_{2}j_{3})\\
				   &\quad + \tau(j_1)\phi(j_{2}j_{3}j_4)
				   	+ \tau(j_{1} j_{2}j_3) \phi(j_{4})
					+ \tau(j_{1} j_{2}j_4) \phi(j_{3})
					+ \tau(j_{1} j_{3}j_4) \phi(j_{2})\\
				   &= \tau( j_{1} j_{2} j_{3}j_{4}) 
					+ \tau(j_{1} j_{2}) \tau(j_{3}j_{4}) 
					+ \tau(j_{1} j_{4}) \tau(j_{2}j_{3}) 
					+ \tau(j_{1} j_{3}) \tau(j_{2}j_{4}) \\ 
				   &\quad   + \tau(j_{1} j_{3}) \tau(j_{2})\tau(j_{4})
				   	+ \tau(j_{1} j_{2}) \tau(j_{3})\tau(j_{4}) 
					+ \tau(j_{1} j_{4}) \tau(j_{2})\tau(j_{3}) \\
				   &\quad + \tau(j_3j_4)\tau(j_1)\tau(j_2)
						+ \tau(j_2j_3)\tau(j_1)\tau(j_4)  
						+\tau(j_2j_4) \tau(j_1)\tau(j_3) 
						+ \tau(j_1)\tau(j_2)\tau(j_3) \tau(j_4)\\
				  &\quad	
				  	+ \tau(j_1)\tau(j_{2}j_{3}j_4)
					+ \tau(j_{1} j_{2}j_3) \tau(j_{4})
					+ \tau(j_{1} j_{2}j_4) \tau(j_{3})
					+ \tau(j_{1} j_{3}j_4) \tau(j_{2}).
\end{align*}

Next, observe that by defining the linear map $\tau(w):=\tau_W(w):=W^{(|w|)}_w$ for $w=j_{i_1}\cdots j_{i_n} \in T_n(J)$ and $\tau(\un):=0$, we obtain that the above relations equal the relations \eqref{nonplanar1} between full and connected non-planar Green's functions up to order four, i.e., $\phi(w)=\tau_Z(w):=Z^{(|w|)}_W$. 

Hence, analog to the planar case, we obtain that the relation between non-planar full and connected Green's functions, $\tau_{Z}$ respectively $\tau_{W}$, is encoded by the linear fixed point equation.
\begin{equation}
\label{npFullConGF}
	\tau_Z = \iota + \tau_W \prec \tau_Z,
\end{equation}
in $\bar{T}^*(J)$. Indeed, by the very definition of the left half-shuffle product $\prec$, this equation gives for any word  $w=j_{i_1}\cdots j_{i_n} \in T_n(J)$
\begin{align}
		Z^{(n)}_{j_{i_1}\ldots j_{i_n}}	=\tau_{Z}(w)
								&=(\tau_{W} \prec \tau_{Z}) (w) \nonumber\\
								&=\sum_{S=\{1=s_1,\ldots,s_k\}\subset [n]}
	W^{(k)}_{j_{i_{s_1}}j_{i_{s_2}}\cdots j_{i_{s_k}}}
		Z^{(n - k)}_{j_{i_{2}} \cdots j_{i_{s_2 - 1}} j_{i_{s_2 + 1}} \cdots  j_{i_{s_k - 1}} j_{i_{s_k + 1}} \cdots j_{i_{n}}}. \label{unshuffleRelation2}		
\end{align}

The solution to this fixed point equation is given by exponentials.

\begin{thm}\label{thm:expGF}
Expanding the linear fixed point equation \eqref{npFullConGF} yields
\begin{equation}
\label{Solution1}
	\tau_Z=\exp^{\prec}\big(\tau_W\big). 
\end{equation}
A closed solution is given in terms of the proper exponential map defined with respect to the commutative shuffle product $\shuffle : \bar  T^*(J) \otimes \bar  T^*(J) \to \bar  T^*(J)$
\begin{equation}
\label{Solution2}
	\tau_Z=\exp^{\shuffle}\big(\tau_W\big),
\end{equation}
which is the convolution product in $\bar  T^*(J)$, defined in terms of the left and right half-shuffles defined in \eqref{deshufleft} respectively \eqref{deshufright}, i.e., $\shuffle = \prec + \succ$.  
\end{thm}

\begin{proof}
The first statement is obvious and follows from comparing the expansions of \eqref{Solution1} with the iteration of \eqref{npFullConGF}. The second statement follows from the fact that the unshuffle coproduct \eqref{unshuffle} is cocommutative. In the resulting commutative shuffle algebra Theorem \ref{thm:ExpSolution} says that $\Omega'(\tau_W)=\tau_W$. Eventually, it is easy to verify that the shuffle algebra identity $a^{\shuffle n} = n!\ a^{\prec{n}}$ holds. 
\end{proof}	

Note that the actual statement is the interesting identity of exponentials
$$
	Z^{(n)}_{j_{i_1} \cdots j_{i_n}} 	
	= \left. \frac{\partial^n}{\partial_{j_{i_n}} \cdots \partial_{j_{i_1}}}\right |_{j=0} \exp\big( W(j) \big) 
	= \exp^{\shuffle}(\tau_W)(j_{i_1} \cdots j_{i_n}),
$$
which follows immediately from the Leibniz rule.

%%%%%%%%%%%%%%%%%%%%%%%%%%%%%%%%%%%%%%%

\subsection{Partition Green's function (non-planar)}
\label{sub:partition}

We saw that the relation between full and connected Green's functions can be refined in the planar case by introducing non-crossing Green's functions. The relations between full and connected Green's functions can be refined similarly in the classical, i.e., non-planar, case. The point of view we develop here seems to be new. We sketch the arguments that follow the same lines as for the planar case.

A Feynman graph $\Gamma$ with external sources $j_1,\ldots, j_n$ in a given QFT splits into $n\geq k\geq 1$ connected components. Grouping external sources according to their belonging to a common connected component defines a partition $L=\{L_1,\ldots,L_k\}$ of $[n]$. We assume from now on that the $L_i$ are ordered according to their minimal elements, so that $1\in L_1$. Summing over all partitions of $[n]$ (the set of partitions of $[n]$ is written $P_n$), we get:
$$
	Z^{(n)}_{j_1 \cdots j_n}
	=\sum\limits_{L\in P_n}\prod\limits_{i=1}^kW^{|L_i|}_{j_{l_1^i}\cdots j_{l_{|L_i|}^i}},
$$
where $L_i=\{l_1^i, \ldots , l_{|L_i|}^i\}$. We abbreviate $W^{|L_i|}_{j_{l_1^i}\cdots j_{l_{|L_i|}^i}}$ to $W^{L_i}$ from now on and define the partitioned Green's functions by
$$
	L^{(n)}_{j_1 \cdots j_n}:=\prod\limits_{i=1}^kW^{|L_i|}_{j_{l_1^i}\cdots j_{l_{|L_i|}^i}},
$$ 
so that
$$
	Z^{(n)}_{j_1 \cdots j_n}=\sum\limits_{L\in P_n}L^{(n)}_{j_1 \cdots j_n}.
$$
Notice that we use deliberately a similar notation as for the planar case since it will always be clear from the context that one works with a planar or non-planar theory.

The corresponding bialgebra of partitions is defined as follows. We write $P$ (resp. $\bar{P}$) for the linear span of the $P_n$, $n\in\Nb^\ast$ (resp. in $\Nb$). The product is the shifted concatenation of partitions: for $L=\{L_1,\ldots,L_k\}$ a partition of $[n]$ and $K=\{K_1,\ldots,K_l\}$ a partition of $[m]$, $L\cdot K$ is the partition $\{L_1,\ldots,L_k,K_1+n,\ldots,K_l+n\}$ of $[n+m]$ where, for a subset $S=\{s_1,\ldots,s_p\}$ of the integers, $S+n$ stands for $\{s_1+n,\ldots,s_p+n\}$.
The coproduct acting on partitions is simply the (standardized) unshuffling coproduct:
$$
	\Delta(L)=\Delta(\{L_1,\ldots,L_k\})=\sum\limits_{I\coprod J=[n]}st(L_I)\otimes st(L_J),
$$
where $L_I=\{L_{i_1},\ldots,L_{i_k}\}$ for $I=\{i_1,\ldots,i_k\}$. This coproduct is a map of algebras from $P$ to $P\otimes P$, and it is coassociative:
$$
	(\Delta\otimes \id)\otimes\Delta (L)=
	\sum\limits_{I\coprod J\coprod K=[n]}st(L_I)\otimes st(L_J)\otimes st(L_K)=(\id \otimes \Delta)\circ\Delta (L).
$$
It splits naturally into two half-coproducts
$$
	\Delta_\prec (L):=\sum\limits_{I\coprod J=[n] \atop 1\in I}st(L_I)\otimes st(L_J)
$$ 
and $\Delta_\succ:= \Delta - \Delta_\prec$.

Decorated partitions are defined exactly as in the case of non-crossing partitions, and the previous constructions hold in the presence of decorations as well. We finally get:

\begin{thm}
The algebra ${P}$ and its decorated version ${P}(J)$ equipped with the concatenation product and the coproduct $\Delta = \Delta_\prec + \Delta_\succ$ define an unshuffle bialgebras.
\end{thm}

Using \it mutatis mutandis \rm the same notations as for the planar case, the splitting map $Sp$ is now defined as a map from $T(J)$ to ${P}(J)$:
$$
	Sp(j_{i_1}\cdots j_{i_n}):=\sum\limits_{L\in P_n}L(j_{i_1}\cdots j_{i_n}),
$$
and we get  the following proposition (the proof of which is left to the reader):

\begin{prop}
The splitting map $Sp$ is a map of unshuffle bialgebras from $\bar T(J)$ to $\bar {P}(J)$.
\end{prop}

The standard section of a unital map $\kappa$ from $T(J)$ to the scalars is defined by 
$$
	st(\kappa)(L(j_{i_1}\cdots j_{i_k})):=\kappa(j_{i_1}\cdots j_{i_k})
$$
if $L$ is the trivial partition $(L=[n]$), and zero else. For an arbitrary partition $L=\{L_1,\ldots,L_k\}$ of $[n]$ we write
$$
	\kappa^L(j_{i_1}\cdots j_{i_k}):=\prod\limits_{i=1}^k\kappa(j_{L_i}).
$$
Finally, we arrive at the next result.

\begin{prop}
Let $\kappa$ be the linear form on $T(J)$ defined by the connected Green's function
$$
	\kappa(j_{i_1}\cdots j_{i_n}):=W^{(n)}_{j_{i_1}\cdots j_{i_n}}.
$$
Then we have that: $sd(\kappa)(L(j_{i_1}\cdots j_{i_n}))=W^{(n)}_{j_{i_1}\cdots j_{i_n}}$ if $L=[n]$, and zero else. 
Moreover, the solution to linear fixed point equation 
$$	
	\Psi=\iota +sd(\kappa)\prec\Psi
$$ 
gives the partitioned Green's function, that is:
$$
	\Psi(L(j_{i_1}\cdots j_{i_n}))=L^{(n)}_{j_{i_1}\cdots j_{i_n}}.
$$
From the properties of the splitting map $Sp$ we recover the relation linking full and partitioned Green's functions:
$$
	Z^{(n)}_{j_{i_1}\cdots j_{i_n}}=\Psi(Sp(j_{i_1}\cdots j_{i_n}))
	=\sum\limits_{L\in P_n}L^{(n)}_{j_{i_1}\cdots j_{i_n}}.
$$
\end{prop}

%%%%%%%%%%%%%%%%%%%%%%%%%%%%%%%%%%%%%%%%

\section{Planar Connected and 1PI Green's functions}
\label{sect:1PIplGreen}

We consider now the relation between planar connected and 1PI Green's functions from a Hopf algebraic point of view. It is well-known that these Green's functions are related through planar trees, that is, the planar connected $n$-point Green's function is given as an -- infinite -- sum over planar trees with $n$ external legs and 1PI vertex insertions. In \cite{cvitanovic1,cvitanovic2} the corresponding diagrammatical and functional calculus is explained. The latter is rather involved due to the non-commutative nature of the external sources, which renders the notion of functional derivatives and the corresponding chain and Leibniz rules rather non-standard. The diagrammatical description can be summarized by stating that non-planar trees must be avoided. However, the precise link between both the diagrammatical and functional calculi is non-trivial -- compared, for instance, to the non-planar case.         

As outlined above, the Hopf algebraic approach puts the linear fixed point equation \eqref{pFullConGF} in the center, which provides a rather natural way to capture the relations between planar (as well as non-planar) full and connected Green's functions. The functional calculus is replaced by looking at generating functionals as linear maps from the Hopf algebra $\bar {T}(T(J))$ (in the non-planar case $\bar T(J)$) into the base field $\mathbb{K}$. The aim of this section is to show how this approach also provides an efficient way to describe the relation between planar connected and 1PI Green's functions. Our work is inspired by \cite{bp,mestreoeckl}, which deal with the non-planar case. We emphasise, however, that our approach provides a new understanding of the links between planar connected and 1PI Green's functions. Regarding the non-planar case, our approach is different from the one developed in the aforementioned articles. However, from now on we refrain from developing in parallel the non-planar case, and leave the task of writing down details to the interested reader.

As mentioned, the starting point for our approach is the linear fixed point equation  \eqref{dendeq}, or \eqref{pFullConGF},  in the convolution algebra $\bar {T}^*(T(J))$
$$
	\phi = \varepsilon + \kappa \prec \phi
$$
and its straightforward exponential solution 
$$
	\phi = \exp^\prec(\kappa).
$$
In the case of planar full and connected Green's functions, the map $\phi$ becomes the character $\tau_{\mathrm{Z}}$, and $\kappa$ is the infinitesimal character $\tau_{\mathrm{W}}$.      

The relation between planar connected and 1PI Green's functions, both seen as linear functions on $\bar{T}(T(J))$ is described by modifying  \eqref{dendeq} as follows. First, recall that $\delta_{\prec}^{(0)}:=\id$ and $\delta_{\prec}^{(n)}:=(\id \otimes \delta_{\prec}^{(n-1)})\circ \delta_{\prec}$ -- note that $\delta_{\prec}$ is not coassociative. Then, it is clear that in general 
\begin{align*}
	\phi 	&= \varepsilon + \kappa + \kappa \prec \kappa +  \kappa \prec (\kappa \prec \kappa) 
				+  \kappa \prec (\kappa \prec (\kappa \prec \kappa)) + \cdots \\  
		&= \varepsilon + \kappa + (\kappa \otimes \kappa) \circ \delta_{\prec} 
					 + (\kappa \otimes \kappa \otimes \kappa) \circ (\id \otimes \delta_{\prec})\circ\delta_{\prec} \\
		& \qquad   + (\kappa \otimes \kappa \otimes \kappa \otimes \kappa) \circ (\id \otimes (\id \otimes \delta_{\prec})\circ
					 \delta_{\prec} )\circ\delta_{\prec} + \cdots\\
		&= \varepsilon + \sum_{n>0} \kappa^{\otimes n} \circ \delta_{\prec}^{(n-1)}\\
		&= \varepsilon + \kappa + \sum_{n>1} \kappa^{\otimes n} \circ \prod_{j=2}^{n} (\id^{(\otimes n-j)} \otimes \delta_{\prec}),
\end{align*}
where $\id^{(0)}=1$ and $1 \otimes \delta_{\prec} = \delta_{\prec}$. Now, recall \eqref{HAleft2} in Remark \ref{Cor:singleInterval}. Assuming that $\kappa$ is an infinitesimal character permits to reduce the calculation of the left half-unshuffle to formula \eqref{HAleft2}. Hence, for instance, the solution of \eqref{pFullConGF} may actually be written
$$
	\tau_{\mathrm{Z}} = \varepsilon + \sum_{n>0} \tau_{\mathrm{W}}^{\otimes n} \circ \hat\delta_{\prec}^{(n-1)}.
$$
The goal is now to relate the linear map $\tau_{\mathrm{W}}$ to the 1PI $n$-point Green's functions written  $\mathrm{I}^{(n)}_{j_{i_1} \cdots j_{i_n}}$ using a fixed point-type relation. For the sake of simplicity, we assume that tadpoles have been removed and that $\mathrm{I}^{(1)}_{j_{i}}:=0$. Since $\tau_{\mathrm{W}}$ is an infinitesimal character, the same kind of property should hold when dealing with planar 1PI Green's functions. In particular, the relation with the connected ones should follow from a proper iteration of the $\delta_\prec$ map. This is indeed the case.

%for the linear function % $\tau_{\mathrm{I}}$ generating functional
%representing the generating functional  of planar 1PI Green's functions $I^{(n)}_{j_{i_1}\cdots j_{i_n}}$:
%$\tau_{\mathrm{I}}$ is the infinitesimal character on $\bar T(T(J))$ whose restriction to $T(J)$ is given by:
%$$\tau_{\mathrm{I}}(j_{i_1}\cdots j_{i_n}):=I^{(n)}_{j_{i_1}\cdots j_{i_n}}.$$

For notational convenience and later use, we also recall the definition of amputated Green's functions, and introduce partially amputated planar 1PI Green's functions (recall that the Einstein convention is in place and one should integrate or sum over repeated indices). The amputated 1PI Green's function $\mathrm{AI}^{(n)}$  is defined by:
$$
	\mathrm{AI}^{(n)}_{y_1\ldots y_n} \prod\limits_{k=1}^n \mathrm{G}_F(j_{i_k},y_k) = I^{(n)}_{j_{i_1}\cdots j_{i_n}},
$$
where $\mathrm{G}_F(x,y)$ is the usual Feynman propagator \cite{itzyksonzuber}. To define the notion of partially amputated Green's functions, we have to distinguish between two sets of parameters: $\{j_{i_1},\ldots ,j_{i_k}\}$ stand for external parameters, whereas $\{y_1,\ldots ,y_{n-k}\}$ are associated with internal ones. We let $a_1\cdots a_n$ be a word containing each of these letters once, and define a map $\beta$ from the set $[k]$ to the set $[n]$ by $y_i := a_{\beta(i)}$.
Partially amputated $n$-point Green's functions $\mathrm{PI}^{(n)}$ are then given by:
$$
	\mathrm{PI}^{(n)}_{a_1\cdots a_n}:=\mathrm{AI}^{(n)}_{a_1' \cdots a_n'}
	\prod\limits_{l=1}^k \mathrm{G}_F(y_l,a'_{\beta(l)}).
$$
where $a_i = a_i'$ if $i$ is not in the image of the map $\beta$ (else it is a new parameter over which one has to integrate or sum).

The goal of linking 1PI and connected Green's functions is now achieved by doubling the alphabet, i.e., we introduce another infinite set of letters $X:=\{x_1,x_1',x_2,x_2',x_3,x_3',\ldots\}$. Let $X_n:=\{x_1,x_1',x_2,x_2',x_3,x_3',\ldots, x_n,x_n'\}$, where $X_0:=\emptyset$. Next we define the maps $R^{(i)} : T(J \coprod X_{i-1}) \to T(J \coprod X_i)^{\otimes 2}$, for $i \in \mathbb{N}^+$
\begin{equation} 
\label{def:Ri}
	R^{(i)}({a_{1}\cdots a_{n}}) :=  \sum_{I_1 \coprod I_2 \coprod I_3 = [n] \atop 1 \in I_1, I_2 \neq \emptyset} 
	a_{I_1} x_i a_{I_3} \otimes  x_i' a_{I_2} \in T(J \coprod X_i)^{\otimes 2},
\end{equation} 	
where $a_1,\ldots,a_n\in J \coprod X_{i-1}$.
Recall that $I_1,I_2,I_3$ are three disjoint intervals ordered by their minimal elements $\min(I_1)=1 < \min(I_2) < \min(I_3)$, and set partitioning $[n]$, i.e., $I_1 \coprod I_2 \coprod I_3 = [n]$.  Hence, the map $R^{(i)}$ is essentially defined in terms of $\hat\delta_{\prec}$ by inserting the letter $x_i \in X$ in the position where the connected component $I_2$ has been extracted from $[n]$, and concatenating $a_{I_2}$ from the left by the letter $x_i'$.
%The infinitesimal character $\tau_{\mathrm{I}}$ is extended to the $\bar T(T(J \coprod X_i))$ in the obvious way (the $x_i,x_i'$ are treated as new sources in the computation of the 1PI Green's functions).

With the maps $R^{(i)}$ at hand, we define the following equation which will  be key in relating the infinitesimal characters $\tau_{\mathrm{W}}$ and $\tau_{\mathrm{I}}$. 
\begin{equation} 
\label{def:key1}
	\Phi (j_{i_1}\cdots j_{i_n}):= \sum_{n>0} F_\mathrm{I}^{(n)} \circ \mathcal{R}^{(n-1)}(j_{i_1}\cdots j_{i_n}).
\end{equation} 		
In fact, we will show that $\Phi$ is nothing but the restriction of $\tau_W$ to $T(J)$: it describes the relation between two planar Green's functions. But, before we get to this point, we have to understand how its ingredients work. 

\begin{enumerate}[i)]

\item The maps $\mathcal{R}^{(n)}$ will serve to generate all the planar trees that will appear in the expansion of connected Green's functions as sums over planar trees with 1PI insertions:
$\mathcal{R}^{(0)} = \id$, and	
\begin{equation} 
\label{def:Rn}
	\mathcal{R}^{(n)} := (R^{(n)} \otimes \id^{(\otimes n-1)}) \circ 
					(R^{(n-1)} \otimes \id^{(\otimes n-2)} ) \circ 
						\cdots \circ R^{(1)}. 
\end{equation}
Note that  $\id^{(\otimes 0)} = 1$, such that $\mathcal{R}^{(1)} = R^{(1)}$. 
	
\item
The map $F_\mathrm{I}^{(n)}$ should be thought of as a Feynman-type rule. Indeed, it tells what amplitude is going to be associated to a planar tree constructed using the map $\mathcal{R}^{(n)} $ when 1PI insertions are taken into account. The rule reads as follows. Notice first, that for any tensor product $O = O_1\otimes \cdots \otimes O_n$ constructed from $\mathcal{R}^{(n)}(j_{i_1} \cdots j_{i_k})$, the $X_i$ are non-commutative monomials in the letters $j_{i_1},\ldots, j_{i_k}$ and $x_1,x_1',\ldots , x_n,x_n'$. Moreover each letter appears only once in the tensor product. Then: 
\begin{itemize}

\item
to each monomial $O_i$ (for example $j_3x_2j_5j_6x_1$), we associate the corresponding 1PI partially amputated Green's function (for example $\mathrm{PI}^{(5)}_{j_3x_2j_5j_6x_1}$). The $j_i$ are treated as external parameters as they give rise to Feynman propagators in the expansion of $\mathrm{PI}$ in terms of Feynman propagators and fully amputated Greens functions. The $x_i$ are treated as internal parameters.

\item to each pair $x_i,x_i'$ we associate the corresponding Feynman propagator $\mathrm{G}_F(x_ix_i')$.

\item at last, as usual, one has to define the symmetry factor $\gamma_X$ associated to the tensor product. Its calculation will be explained later on.

\end{itemize}
The action of the map $F_\mathrm{I}^{(n)}$ on $O$ is then obtained by multiplying the so-obtained 1PI amputated Green's function and Feynman propagators, dividing by the symmetry factor, and integrating over repeated indices, i.e., the $x_i$ and the $x_i'$.

Let us look at some examples at low orders to get acquainted with this construction.
\begin{align}
	F_\mathrm{I}^{(2)}\circ R^{(1)} (j_1 j_2) 
		&= F_\mathrm{I}^{(2)} (j_1 x_1 \otimes x'_1 j_2) \nonumber \\
		&= \mathrm{PI}^{(2)}_{j_1 x_1 } {\mathrm{G}}_F(x_1x'_1) \mathrm{PI}^{(2)}_{ x'_1j_2 }. \label{diagram1}
\end{align}
 Diagrammatically, the 1PI $2$-point function $\mathrm{I}^{(2)}_{j_1 j_2}$ is represented by
\begin{equation}
\label{1PI2point}
	\aFPxxyyFPb
\end{equation}
whereas its partially amputated $\mathrm{PI}^{(2)}_{j_1 x_2}$ version is represented by the same diagram where the right hand side half-line has been reduced 
\begin{equation}
\label{1PI2pointAMPU}
	\aFPxxyyFPbAMPU
\end{equation}
 A Feynman propagator is represented as usual by an edge, and finally
 \eqref{diagram1} is represented by 
\begin{equation}
\label{connected2point2}
\aFPxxyyFPzzuuFPb
\end{equation}
%Note that in \eqref{diagram1} we have omitted the Feynman propagators connected to the external sources.

\noindent Another example is in order to understand \eqref{def:key1}. 
\begin{align*}
	F_\mathrm{I}^{(3)}  \circ \mathcal{R}^{(2)} (j_1 j_2) 
		&= F_\mathrm{I}^{(3)} 
										\circ (R^{(2)} \otimes \id)\circ R^{(1)} (j_1 j_2)\\
		&:=  F_\mathrm{I}^{(3)} 
										\circ R^{(2)}(j_1 x_1) \otimes x'_1 j_2\\
		&:= F_\mathrm{I}^{(3)} ( 
										 j_1 x_2 \otimes x'_2 x_1 \otimes x'_1 j_2)\\
		&:=  	\mathrm{PI}^{(2)}_{j_1 x_2}{\mathrm{G}}_F(x_2x'_2)
			\mathrm{AI}^{(2)}_{x'_2 x_1} {\mathrm{G}}_F(x_1x'_1)
			\mathrm{PI}^{(2)}_{x'_1j_2}.
\end{align*}

\noindent Diagrammatically, this is represented by 

\begin{equation}
\label{connected2point3}
\gTgTgTg
\end{equation}

\noindent Going one order higher we find three terms for the expression
 \begin{align}
	F_\mathrm{I}^{(2)} \circ R^{(1)} (j_1 j_2j_3) 
		&=F_\mathrm{I}^{(2)} ( j_1 x_1 \otimes x'_1 j_2j_3 
													+ j_1 j_2 x_1 \otimes x'_1 j_3
													+ j_1  x_1 j_3 \otimes x'_1 j_2)\nonumber\\
		&= \mathrm{PI}^{(2)}_{j_1 x_1 } {\mathrm{G}}_F(x_1x'_1) \mathrm{PI}^{(3)}_{x'_1 j_2j_3 }
			+  \mathrm{PI}^{(3)}_{j_1 j_2x_1 } {\mathrm{G}}_F(x_1x'_1) \mathrm{PI}^{(2)}_{x'_1 j_3 } \nonumber\\
		&\quad	+  	\mathrm{PI}^{(3)}_{j_1 x_1 j_3 } {\mathrm{G}}_F(x_1x'_1) 
					\mathrm{PI}^{(2)}_{x'_1 j_2}. \label{connected3point-formula}
\end{align}
Diagrammatically, this is represented by a sum of three planar tree diagrams. Each summand corresponds to the respective term in the above sum

\vspace{-0.5cm}

\begin{equation}
\label{connected3point1}
\aFPxxyyFPzzuvuFPbvFPc +\
	 \aFPxbFPyxyzzFPuuvvFPc +\
		 \aFPxcFPzxyzyFPuuvvFPc
\end{equation}

\vspace{2cm}

\noindent 
Until now, no multiplicities (symmetries) showed up, that is, each diagram appeared only once. This changes in the next term in \eqref{def:key1}. Let us make the computation without taking into account symmetry factors properly (we write $F_{nn}$ for the corresponding ``non-normalized'' Feynman rule). We observe seven terms 
\allowdisplaybreaks{
 \begin{align}
	F_{\mathrm{I}nn}^{(3)}
		\circ \mathcal{R}^{(2)} (j_1 j_2j_3) 
		&=F_{\mathrm{I}nn}^{(3)}	\circ (R^{(2)} \otimes \id)( j_1 x_1 \otimes x'_1 j_2j_3 
													+ j_1 j_2 x_1 \otimes x'_1 j_3
													+ j_1  x_1 j_3 \otimes x'_1 j_2)\nonumber\\
		&=F_{\mathrm{I}nn}^{(3)}\circ \big(R^{(2)}(j_1 x_1)  \otimes x'_1 j_2j_3 
													+ R^{(2)}(j_1 j_2 x_1) \otimes x'_1 j_3
													+ R^{(2)}(j_1  x_1 j_3) \otimes x'_1 j_2\big)\nonumber\\	
		&=F_{\mathrm{I}nn}^{(3)}\circ \big(
													j_1 x_2  \otimes x'_2 x_1 \otimes x'_1 j_2j_3				 
													+ j_1 x_2  \otimes x'_2 j_2 x_1 \otimes x'_1 j_3\nonumber\\
								& \quad	
													+ j_1 j_2 x_2  \otimes x'_2  x_1 \otimes x'_1 j_3
													+ j_1 x_2 x_1 \otimes x'_2  j_2 \otimes x'_1 j_3\nonumber\\
								& \quad					 
													+ j_1 x_2 \otimes x'_2 x_1 j_3 \otimes x'_1 j_2
													+ j_1  x_1 x_2 \otimes x'_2 j_3 \otimes x'_1 j_2	
													+ j_1  x_2 j_3 \otimes x'_2 x_1 \otimes x'_1 j_2 \big)\nonumber\\
		&= 	\mathrm{PI}^{(2)}_{j_1 x_2}{\mathrm{G}}_F(x_2x'_2)
			\mathrm{AI}^{(2)}_{x'_2 x_1}{\mathrm{G}}_F(x_1x'_1) 
			\mathrm{PI}^{(3)}_{x'_1 j_2j_3} \label{bigSum}\\
		&\qquad\quad	
				+ 	\mathrm{PI}^{(2)}_{j_1 x_2}{\mathrm{G}}_F(x_2x'_2)
					\mathrm{PI}^{(3)}_{x'_2 j_2 x_1}{\mathrm{G}}_F(x_1x'_1)
					\mathrm{PI}^{(2)}_{x'_1 j_3} \nonumber\\
		&\quad	
				+ 	\mathrm{PI}^{(3)} _{j_1 j_2 x_2}{\mathrm{G}}_F(x_2x'_2)
					\mathrm{AI}^{(2)} _{x'_2 x_1}{\mathrm{G}}_F(x_1x'_1)
					\mathrm{PI}^{(2)} _{x'_1 j_3}\nonumber\\
		&\qquad\quad	
				+	 \mathrm{PI}^{(3)}_{j_1  x_2 x_1}{\mathrm{G}}_F(x_2x'_2)
				 	 \mathrm{PI}^{(2)} _{x'_2 j_2}{\mathrm{G}}_F(x_1x'_1)
					 \mathrm{PI}^{(2)}_{x'_1 j_3} \label{symterm1}\\ 				
		&\quad	
				+ 	\mathrm{PI}^{(2)} _{j_1 x_2}{\mathrm{G}}_F(x_2x'_2) 
				 	\mathrm{PI}^{(3)} _{x'_2 x_1 j_3}{\mathrm{G}}_F(x_1x'_1)
					\mathrm{PI}^{(2)} _{x'_1 j_2}\nonumber\\
		&\qquad\quad	
				+ 	\mathrm{PI}^{(3)} _{j_1 x_1x_2}{\mathrm{G}}_F(x_2x'_2)   
					\mathrm{PI}^{(2)} _{x'_2 j_3}{\mathrm{G}}_F(x_1x'_1)
					\mathrm{PI}^{(2)} _{x'_1 j_2} \label{symterm2}\\
		&\quad	
				+ 	\mathrm{PI}^{(3)} _{j_1 x_2 j_3}{\mathrm{G}}_F(x_2x'_2)   
					\mathrm{AI}^{(2)} _{x'_2 x_1}{\mathrm{G}}_F(x_1x'_1)
					\mathrm{PI}^{(2)} _{x'_1 j_2} \nonumber
		\end{align}}
		
Diagrammatically, this is represented by 

\newpage
%\vspace{-0.7cm}

\begin{equation}
\label{connected3point2}
\kk 
+\ \nn 
+\ \ll
\end{equation}

\vspace{2.5cm}

\begin{equation}
\label{connected3point3}
+ \pp
+\oo
\end{equation}

\vspace{1cm}

\begin{equation}
\label{connected3point4}
+ \pp
+ \mm
\end{equation}
\vspace{1.5cm}

\noindent where each tree graph corresponds respectively to the term in the sum \eqref{bigSum}. 
\end{enumerate}	

Two important remarks are in order. First, note that the fourth and sixth tree graph are identical. In fact, the only difference in the corresponding tensor products is manifested in the different orders of the inner indices $x_1,x'_1,x_2,x'_2$. It is clear that the corresponding expressions in \eqref{symterm1} and \eqref{symterm2} are identical since the ``internal" points $x_1,x_1'$ and $x_2,x_2'$ are summed/integrated over, that is
\begin{align}
	&\mathrm{PI}^{(3)} _{j_1 x_1x_2}{\mathrm{G}}_F(x_2x'_2)   
		\mathrm{PI}^{(2)} _{x'_2 j_3}{\mathrm{G}}_F(x_1x'_1)
		\mathrm{PI}^{(2)} _{x'_1 j_2} \nonumber\\
	&\qquad 
	=  \mathrm{PI}^{(3)} _{j_1  x_2 x_1}{\mathrm{G}}_F(x_2x'_2)
		\mathrm{PI}^{(2)} _{x'_2 j_2}{\mathrm{G}}_F(x_1x'_1)
		\mathrm{PI}^{(2)} _{x'_1 j_3}\label{symmetry}.
\end{align}
The calculation of the symmetry factor associated to these two expressions follows; it is two, and the correct expression for $F_{\mathrm{I}}^{(3)} \circ \mathcal{R}^{(2)} (j_1 j_2j_3) $ is obtained from the one of $F_{\mathrm{I}nn}^{(3)} \circ \mathcal{R}^{(2)} (j_1 j_2j_3) $ by multiplying these two terms by $\frac{1}{2}$. In general, the calculation of symmetry factors can be performed graphically (identifying diagrams as we have just done), or by a combinatorial argument allowing to deduce symmetry factors directly from the tensor products. We will return to this further below. 

A crucial observation is that the ``external" decorations, corresponding to the external -- non-commutative -- sources, or letters, $j_1,j_2,j_3$ are ordered clockwise, and that the tree graphs are strictly planar.    

The latter makes it rather natural to consider those tree graphs inscribed into circles with vertices ordered clockwise on the circumference. Below we have listed those so-called chord-type diagrams corresponding to the above tree graphs. The first one corresponds to the simple Feynman propagator $\mathrm{G}_F(j_1j_2)$. The second corresponds to \eqref{1PI2point}, the 1PI $2$-point function $\mathrm{I}^{(2)} _{j_1 j_2}$. The third corresponds to \eqref{connected2point2}, the connected $2$-point function consisting of two 1PI $2$-point functions linked via the Feynman propagator $\mathrm{G}_F(x_1x'_1)$, etc. In general, the internal decorations correspond to the two arguments of the Feynman propagators connecting the 1PI $2$- or $n$-point functions. However, note that we only denote one of the two arguments. The external decorations correspond to the external sources, that is, letters in $J$.

%\vspace{-2cm}

\begin{equation}
\label{circleGraphs1}
	\abI
	\abII
	\abIII
	\abIV
\end{equation}

\noindent The last chord-type diagram corresponds to \eqref{connected2point3}. The next sum of chord-type diagrams corresponds to \eqref{connected3point1}, and displays diagrammatically the sum of terms defining $F^{(2)}_\mathrm{I}\circ R^{(1)} (j_1 j_2j_3)$.

\vspace{-0.5cm}

\begin{equation}
\label{circleGraphs2}
	\abci +
	\abcii + 
	\abciii
\end{equation}

\vspace{1cm}

Recall that the decoration $x_1 \in X=\{x_1,x'_1,x_2,x'_2,\ldots \}$ encodes the Feynman propagator $\mathrm{G}_F(x_1x'_1)$ linking 1PI functions. %We have omitted the propagators linking to the sources. 
The following seven chord-type diagrams

\begin{equation}
\label{circleGraphs3}
	\abcI
	\abcIV
	\abcII
\end{equation}

\begin{equation}
\label{circleGraphs4}
	\abcVI
	\abcV
	\abcVII
	\abcIII
\end{equation}	

\noindent correspond to the seven tree graphs in \eqref{connected3point2} -- \eqref{connected3point4}. Note that the first and third graph in \eqref{circleGraphs4} represent the two identical tree graphs in \eqref{connected3point3} and \eqref{connected3point4}, respectively. They differ only in the decorations of the two internal edges. 
%This points to the problem, i.e., how to describe the relation between connected and 1PI Green's functions correctly from a Hopf algebraic point of view. To answer this question, we need to take care of those symmetries reflected in the above example, where the same circle graphs only differ by the decorations of internal edges. Otherwise, our approach may lead to an over counting of terms.      

To understand better symmetry factors, we shall move now from chord-type diagrams with inscribed decorated trees to actual decorated planar rooted trees. The key step is to associate to each chord-type diagram, and hence to each tensor product in \eqref{def:key1}, a so-called decreasing planar rooted tree \cite{berflasal,drmota}. It turns out that this is the most natural picture for our Hopf algebraic approach. In fact, this will be further enhanced when we present a graphical approach to the planar functional calculus used in references \cite{cvitanovic1,cvitanovic2}. 		

In the following we consider therefore decorated planar rooted trees with all edges oriented away from the root. The root is at the bottom and decorated by $1$. The $n-1$ leaves are labelled from left to right in strictly increasing order. The left-most leave by $2$ and the right-most leave by $n$. The $m$ inner vertices are decorated by the unprimed elements from the set $X_m:=\{x_1,x'_1,x_2,x'_2, \ldots,x_m,x'_m\}$ in such a way, that for any path from the root to any of the leaves, the decorations of the inner vertices on this path are strictly decreasing. In the following we outline how planar rooted trees correspond to chord-type diagrams. Each edge of a chord-type diagram with $n$ external sources, is marked by a vertex. Those vertices corresponding to edges connected to the external sources are decorated by the set $[n]$ according to the source decoration. Internal vertices are automatically decorated by the corresponding decorations of the edges by -- unprimed -- elements from the set $X_m$. The vertex decorated by $1$ is the root. This root vertex is connected by edges to all those vertices corresponding to edges of the chord-type diagram that share a vertex with the edge associated to the root. Each of these vertices connected to the original root is considered as the root of a sub-corollary, by connecting it to those vertices corresponding to edges of the chord-type diagram that share a vertex with this sub-root. This process stops at the leaves, i.e., the vertices of the edges connected to the external sources. The following rooted trees correspond respectively to the chord-type diagrams in \eqref{circleGraphs1} and \eqref{circleGraphs2}.      
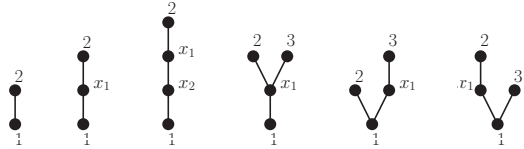
$$
\treeB
\treeC
\treeD
\treeDa
\treeDb
\treeDc
$$
The  following seven rooted trees correspond respectively to the chord-type diagrams in \eqref{circleGraphs3} and \eqref{circleGraphs4}. Note the decreasing decorations of internal vertices by -- unprimed -- elements from $X_2=\{x_1,x'_1,x_2,x'_2\}$. The fourth and fifth tree are symmetric except for the decorations of internal vertices. They represent the left- and right-hand side of \eqref{symmetry}, respectively.
$$
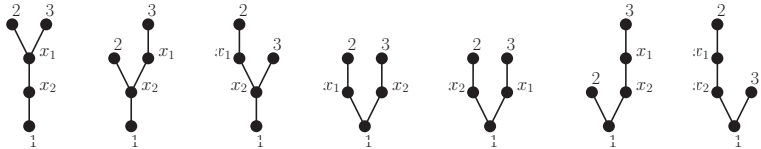
\treeEa
\treeEb
\treeEc
\treeEdi
\treeEdii
\treeEe
\treeEf
$$		

The decorations of the inner vertices of a decreasing planar rooted tree do no contribute to the calculation of the corresponding amplitudes (they are dummy variables to be integrated/summed over). We therefore may identify those trees that differ only by decorations of the internal vertices. For instance the following trees contribute the same term in the calculation of the corresponding amplitude  

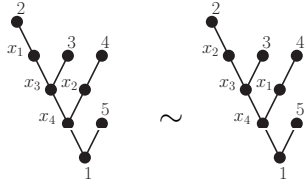
\vspace{-0.5cm}

$$
\treeAsym \sim \treeBsym
$$

\vspace{1cm}

The rule for the computation of symmetry factors is as follows. First we note that each tensor product resulting from expanding ${\mathcal R}^{(n)}$ generates in turn a decreasing rooted tree. Conversely, each decreasing rooted tree can be associated to a tensor product appearing in the expansion of ${\mathcal R}^{(n)}$. This follows by recursion: the decoration $x_1$ corresponds to the first operation resulting from ${\mathcal R}^{(n)}$ applied to the word $j_1 \cdots j_n$ in the construction of the tensor product. This can be continued recursively. For example, by looking at the graph on the left hand side just above, we see that it is obtained by the following sequence of operations:
$$
	j_1 \cdots j_5 \to j_1x_1j_3j_4j_5 \otimes x'_1j_2 \to j_1x_1j_3x_2j_5 \otimes x'_2j_4 \otimes x'_1j_2  \to \cdots
$$
and so on.

Two decreasing trees are equivalent if they are equal when internal decorations have been erased. Let us write $u(T)$ for the undecorated planar rooted tree associated to a decreasing rooted tree $T$. The symmetry factor $\gamma(T)$ associated to a given decreasing rooted tree $T$  is therefore simply the number of decorations of the internal vertices of $u(T)$ that make it a decreasing tree. We will return to this further below.

With this definition, we get as expected

\begin{thm} 
\label{thm:Convs1PI}
\begin{equation} 
\label{def:key2}
	\tau_{\mathrm{W}}(j_1 \cdots j_k) = \sum_{n>0} F_{\mathrm{I}}^{(n)}  \circ \mathcal{R}^{(n-1)}(j_1 \cdots j_k).
\end{equation}
\end{thm}

Note that the right hand side of \eqref{def:key2} can be abstracted by writing it as a sum over all decreasing planar rooted trees with a fixed number of leaves (in the above case $k-1$). The argument in the sum is than replaced by a functional that associates with each decreasing tree the corresponding amplitude multiplied by a proper symmetry factor.

%%%%%%%%%%%%%%%%%%%%%%%%%%%%%%%%%%%%%%%%

\section{Graphical description of planar functional calculus}
\label{sect:graphicalcalc}	

We return to references \cite{cvitanovic1,cvitanovic2}, and try to see how the diagrammatic representation of the relations between planar connected and 1PI Green's functions in terms of decreasing trees can be adapted to the planar, i.e., non-commutative functional calculus proposed by Cvitanovic et al. Recall that the diagrammatic representation of the tensor products in \eqref{def:key2} describing the relation between planar connected and 1PI Green's functions derives from a Hopf algebra point of view. Further below we will see that planar functional derivations translate into a certain growth operation on planar rooted trees.

%%%%%%%%%%%%%%%%%%%%%%%%%%%%%%%%%%%%%%%%

\subsection{Planar functional calculus}
\label{ssect:planarCalc}	

First recall that the generating functional for planar full Green's functions is given by
$$
	\mathrm{Z}[j]:= 1 + \sum_{k>0} \mathrm{Z}^{(k)}_{j_{i_1} \cdots j_{i_k}} j_{i_1} \cdots j_{i_k} 
$$
and that 
$$
	\mathrm{Z}^{(n)}_{j_{i_1} \cdots j_{i_n}} 	= \left. \frac{\partial^n}{\partial_{j_{i_n}} \cdots \partial_{j_{i_1}}}\right |_{j=0}\mathrm{Z}[j].
$$
The rules stated in \cite{cvitanovic1,cvitanovic2} for functional derivations with respect to the strictly non-commutative sources are
$$
	\frac{\partial}{\partial_{j_i}} (u j_{i_1} \cdots j_{i_k} ) = u \delta_{ii_1} j_{i_2} \cdots j_{i_k}, 
$$ 
and for the so-called $c$-number $u \in \mathbb{K}$ we have $\frac{\partial}{\partial_{j_i}} u = 0$. The key ingredients in the planar functional calculus developed in \cite{cvitanovic1,cvitanovic2} are the Leibniz and chain rules. Let $\mathrm{A}[j] = a + \sum_{k>0} \mathrm{A}^{(k)}_{j_{i_1} \cdots j_{i_k}} j_{i_1} \cdots j_{i_k}$ and $\mathrm{B}[j] = b + \sum_{k>0} \mathrm{B}^{(k)}_{j_{i_1 }\cdots j_{i_k}} j_{i_1} \cdots j_{i_k}$ be generating series in the non-commutative sources $j$. Then the functional derivative of the product $\mathrm{A}[j]\mathrm{B}[j]$ is
$$
	\frac{\partial}{\partial_{j_i}} (\mathrm{A}[j]\mathrm{B}[j]) 
	= \frac{\partial}{\partial_{j_i}} (\mathrm{A}[j])\mathrm{B}[j] + \mathrm{A}[0]\frac{\partial}{\partial_{j_i}} \mathrm{B}[j].  
$$
Here $\mathrm{A}[0] = a \in \mathbb{K}$. Now assume that for generating functionals $\mathrm{A},\mathrm{B}_1,\ldots,\mathrm{B}_n,\ldots$, the composition $\mathrm{A}[\mathrm{B}]$ is defined by the substitution rule (composition of formal power series)
$$
	\mathrm{A}[\mathrm{B}] = \sum_{k>0} \mathrm{A}^{(k)}_{j_{i_1} \cdots j_{i_k}} \mathrm{B}_{i_1} \cdots \mathrm{B}_{i_k},
$$
where $\mathrm{B}_i[j] = \sum_{k>0} \mathrm{B}^{(k)}_{i;j_{i_1} \cdots j_{i_k}} j_{i_1} \cdots j_{i_k}$. Note that none of the expansions have $c$-numbers. Then 
$$
	\frac{\partial}{\partial_{j_i}} \mathrm{A}[\mathrm{B}] = \frac{\partial \mathrm{B}_m[j] }{\partial_{j_i}} 
								\frac{\partial}{\partial_{\mathrm{B}_m}}\mathrm{A}[\mathrm{B}], 
$$
with an implicit summation over $m$. The planar 1PI Green's  functions are denoted $\mathrm{\Gamma}[\Phi]$ in \cite{cvitanovic1,cvitanovic2}. They are defined in terms of the fields $\Phi_i := \frac{\partial}{\partial_{j_i}}  \mathrm{W}[j]$. From this it follows that $\frac{\partial}{\partial_{j_i}} =  \mathrm{W}[j]_{il} \frac{\partial}{\partial_{\Phi_l}}$. We refrain from giving a more detailed account of the relations between planar full, $\mathrm{Z}[j]$, connected, $\mathrm{W}[j]$, and 1PI, $\mathrm{\Gamma}[\Phi]$, Green's  functions using planar functional calculus. The reader is referred to Cvitanovic et al.~\cite{cvitanovic1,cvitanovic2}.

The corresponding Legendre transform together with the planar chain and Leibniz rules lead to the relations between planar connected and 1PI Green's functions. Following the non-trivial calculations in \cite{cvitanovic1,cvitanovic2} one arrives, for instance, at 
\begin{align}
\label{con3point}
	 \frac{\partial^3}{\partial_{j_{i_3}} \partial_{j_{i_2}}  \partial_{j_{i_1}}}  \mathrm{W}[j] 
	= \mathrm{W}[j]_{j_{i_1} x_1}\mathrm{W}[0]_{j_{i_2} x_2}\Gamma[\Phi]_{x_1x_2x_3}\mathrm{W}[j]_{x_3j_{i_3}},	
\end{align}
where $\mathrm{W}[j]_{j_{i}j_{l}} = \frac{\partial^2}{\partial_{j_{l}}\partial_{j_{i}}}\mathrm{W}[j]$. From this follows the expression for the planar connected $3$-point Green's function in terms of the planar 1PI $3$-point Green's function, with dressed external legs. The factor $\mathrm{W}[0]_{j_{i_2} x_2}$ is a $c$-number and denotes the full propagator. The calculation of this expression is rather non-trivial. At order four the extraction of four legs from a planar connected generating  functional $\mathrm{W}[j] $ gives
\begin{align}
\label{con4point} 
     \frac{\partial^4}{\partial_{j_{i_4}} \cdots \partial_{j_{i_1}}} \mathrm{W}[j] 
    &= \mathrm{W}[j]_{j_{i_1}j_{i_2} x_1}\mathrm{W}[0]_{j_{i_3} x_2}\Gamma[\Phi]_{x_1x_2x_3}\mathrm{W}[j]_{x_3j_{i_4}}\\
  &\quad  + \mathrm{W}[0]_{j_{i_2} x_2}\mathrm{W}[0]_{j_{i_3} x_3}\mathrm{W}[j]_{j_{i_1}x_1}\Gamma[\Phi]_{x_1x_2x_3x_4}\mathrm{W}[j]_{x_4j_{i_4}} \nonumber \\
   &\quad +  \mathrm{W}[0]_{j_{i_2} x_1}\mathrm{W}[0]_{j_{i_3} x_2}\Gamma[0]_{x_1x_2x_3}\mathrm{W}[j]_{j_{i_1}x_3j_{i_4}}.  \nonumber 
\end{align}     
In these calculations it is rather challenging to keep track of the planar derivatives as well as the appearance of the $c$-number Green's functions, such as $\mathrm{W}[0]_{j_{i_2} x_2}$ or $\Gamma[0]_{x_1x_2x_3}$. The diagrammatical representation used in \cite{cvitanovic1,cvitanovic2} is intuitive, but may hide some underlying simplicity as we would like to indicate further below.

%%%%%%%%%%%%%%%%%%%%%%%%%%%%%%%%%%%%%%%%

\subsection{Resummation of planar rooted trees}
\label{ssect:planarRootedTrees}	
   		
We always assume that planar rooted trees with $n-1$ leaves carry, from now on, the decoration $1$ on the root. The leaves are decorated, as before, from left to right in increasing order by the integers $2, \ldots, n$. Recall that this corresponds to terms in the expansion of $\mathrm{W}^{(n)}_{j_1 \cdots j_n}$, with the conventions used in an earlier section. The notion of {\it{(branch) reduced tree}} is introduced. This is a planar rooted tree without vertices that have only one incoming and one outgoing edge. The $p$ internal vertices of a reduced tree are decorated by -- unprimed -- elements from $X_p:=\{x_1,x'_1\cdots,x_p,x'_p\}$ from top to bottom and right to left, which we call the standard decoration. The following example shows a tree with six leaves and three internal vertices decorated by elements from $X_p:=\{x_1,x'_1,x_2,x'_2,x_3,x'_3\}$
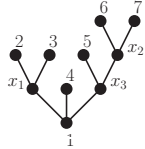
$$
	\ExpTreeA
$$
Reduced rooted trees correspond to amplitudes calculated according to the Feynman rule $F_\mathrm{I}^{(n)}$, that have no $2$-point 1PI Green's function terms $\mathrm{PI}^{(2)}_{j_lx}$, $\mathrm{PI}^{(2)}_{xj_l}$ or $\mathrm{AI}^{(2)}_{xx'}$.

\begin{defn} 
The reduction operator $\rm{red}$ maps a planar rooted tree to the corresponding reduced rooted tree by erasing so-called branches, that is, vertices with one incoming and one outgoing edge.
\end{defn}

For example
$$
	\mathrm{red}\Big(\hspace{-0.2cm}\begin{array}{c}\\[-1.5cm]\ExpTreeBa\end{array}\hspace{-0.4cm}\Big) 
	= \begin{array}{c} \\[-1.5cm]\ExpTreeBb\end{array}
$$

\vspace{1cm}

One of the usual processes in QFT, consisting of replacing Feynman propagators by connected $2$-point Green's functions (through a suitable resummation of Feynman graphs), is then obtained by defining a new Feynman rule $F_R$ for reduced trees $T \in T^{red}$ by
$$
	F_R(T)= \sum_{\mathrm{red}(T')=T} F_\mathrm{I}(T'),
$$
where the sum is over planar rooted trees $T'$ such that $\mathrm{red}(T')=T$ and 
$$
	F_\mathrm{I}(T') := F_\mathrm{I}^{(m)}(T'_d),
$$
where $m$ is the number of internal vertices of $T'$, and $T'_d$ stands for $T'$ equipped with an arbitrary decreasing decoration. Theorem \ref{thm:Convs1PI} then yields
\begin{equation}
\label{expansion-reduced-trees}
	\mathrm{W}^{(n)}_{j_1 \cdots j_n} = \sum_{T \in \mathcal{T}^{red}_{n-1}} F_{R}(T),
\end{equation}
where the sum runs over the set $\mathcal{T}^{red}_{n-1}$ of all reduced planar rooted trees with a fixed number of $n-1$ leaves. 

\begin{rmk}\label{schroeder}{\rm{Note that the sum in \eqref{expansion-reduced-trees} over reduced trees with a fixed number of leaves is finite. The number of terms is given by the Schr\"oder--Hipparchus or super-Catalan numbers, which count the number of planar trees with a given number of leaves: 1, 1, 3, 11, 45, 197, 903, 4279, 20793, 103049, ... \footnote{Sequence A001003 in OEIS.} See \cite{Stanley} for details.
}}
\end{rmk}

Let us describe the new Feynman rule $F_R$ using Cvitanovich's notation. Hence, we write $\Gamma[0]_{x_1'x_2'x_3'}$ for the amputated 1PI $3$-point function $\mathrm{AI}_{x_1'x_2'x_3'}$, and similarly for the other terms. The rule for a reduced planar rooted tree with standard decoration then reads
\begin{enumerate}

\item
Create an amputated 1PI function for the root and all other internal vertices.

\item
Create a connected $2$-point Green's function for all edges (and another for the source $j_1$ corresponding to the root).

\item
Take the product of these terms and integrate/sum over repeated indices.

\end{enumerate}  		
These rules are best understood through an example. The following decorated planar rooted tree
\begin{equation}
\label{exampleTree1}
	\ExpTreeC
\end{equation}
yields 
\begin{align*}
	&\mathrm{W}[0]_{x_1 x'_1}
	\Gamma[0]_{x'_1a'_2a'_3}
	\mathrm{W}[0]_{a'_2 j_2} 
	\mathrm{W}[0]_{a'_3 j_3}\\
	&\mathrm{W}[0]_{x_2 x'_2}
	\Gamma[0]_{x'_2a'_5a'_6}
	\mathrm{W}[0]_{a'_5 j_5} 
	\mathrm{W}[0]_{a'_6 j_6}\\
	&\mathrm{W}[0]_{j_1 a'_1}
	\Gamma[0]_{a'_1x_1a'_4x_2}
	\mathrm{W}[0]_{a'_4j_4} 
\end{align*}

%%%%%%%%%%%%%%%%%%%%%%%%%%%%%%%%%%%%%%%%

\subsection{Functional derivations revisited}
\label{ssect:TreeGrowth}	

Another notation is needed to encode functional derivations. We will write  $\Gamma[\Phi]_{a_1 \cdots a_n}$ for the generating series
$$
	\Gamma[\Phi]_{a_1  \cdots a_n} = \sum_{i_1, \ldots, i_k} \mathrm{AI}^{(n+k)}_{a_1  \cdots a_n j_{i_1} \cdots j_{i_k}}
	 \mathrm{W}[j]_{i_1} \cdots \mathrm{W}[j]_{i_k}.
$$
Recall that $\mathrm{W}[j]_{i_k} = \frac{\partial}{\partial_{j_{i_k}}}\mathrm{W}[j] = \Phi_{i_k}$.

With this notation, the following ansatz defines a new Feynman rule $F_{\gamma}$ for reduced graphs equipped with the standard decoration.
\begin{enumerate}

\item \label{item:1}
In the expansion of $F_R(T)$, replace $\mathrm{W}[0]$ (respectively $\Gamma[0]$) by $\mathrm{W}[j]$ (respectively $\Gamma[\Phi]$), whenever the corresponding term is associated to the path going from the root to the right most leaf of the tree $T$.  

\item  \label{item:2}
Order these (non-commuting) terms from left to right according to the order obtained by following the path from the right most leaf to the root.

\end{enumerate}
		
Regarding \eqref{item:1} the rule gives with respect to the foregoing example tree \eqref{exampleTree1} for each term in the path from the root to the rightmost leaf:
$\mathrm{W}[0]_{j_1 a'_1}$;
	$\Gamma[0]_{a'_1x_1a'_4x_2}$;\ 
	$\mathrm{W}[0]_{x_2 x'_2}$;\ 
	$\Gamma[0]_{x'_2a'_5a'_6}$;\  
	$\mathrm{W}[0]_{a'_6 j_6}$
the expressions
\begin{equation}
\label{exampleTerms1}
	\mathrm{W}[j]_{j_1 a'_1};\
	\Gamma[\Phi]_{a'_1x_1a'_4x_2};\
	\mathrm{W}[j]_{x_2 x'_2};\
	\Gamma[\Phi]_{x'_2a'_5a'_6};\
	\mathrm{W}[j]_{a'_6 j_6}.
\end{equation}

Following rule \eqref{item:2}, that is, taking the order into account for the non-commutative terms in \eqref{exampleTerms1}, we obtain the expression 
\begin{align*}
	\Gamma[0]_{x'_1a'_2a'_3}
	\mathrm{W}[0]_{x_1 x'_1}
	\mathrm{W}[0]_{a'_2 j_2} 
	\mathrm{W}[0]_{a'_3 j_3}
	\mathrm{W}[0]_{a'_4j_4} 
	\mathrm{W}[0]_{a'_5 j_5}\\
	\mathrm{W}[j]_{a'_6 j_6}
	\Gamma[\Phi]_{x'_2a'_5a'_6}
	\mathrm{W}[j]_{x_2 x'_2}
	\Gamma[\Phi]_{a'_1x_1a'_4x_2}
	\mathrm{W}[j]_{j_1 a'_1}
\end{align*}

Let us look at another simple example. The term corresponding to the rooted tree
$$
	\ExampleTreeA,
$$
that follows from the rules  \eqref{item:1} and  \eqref{item:2} is
$$
	\mathrm{W}[0]_{a'_2 j_2} \mathrm{W}[j]_{a'_3 j_3} \Gamma[\Phi]_{a'_1a'_2a'_3}\mathrm{W}[j]_{j_1 a'_1},
$$
and we recognise $\mathrm{W}[j]_{j_1 j_2 j_3}$. Graphically we denote this expression by distinguishing the rightmost brach, i.e., the path form the root to the rightmost leaf by thickening the vertices and edges along this path. For the above example this yields

$$
	\ExampleTreeAa
$$

For the linear combination of rooted trees
$$
	\ExampleTreeB + \ExampleTreeC + \ExampleTreeD
$$
we obtain 
\begin{align*}
	&\mathrm{W}[0]_{a'_2 j_2}
	\mathrm{W}[0]_{a'_3 j_3}
	\mathrm{W}[j]_{a'_4j_4}
	\Gamma[\Phi]_{a'_1a'_2a'_3a'_4}
	\mathrm{W}[j]_{j_1 a'_1} \\
	&+
	\mathrm{W}[0]_{a'_2 j_2}
	\mathrm{W}[0]_{a'_3 j_3}
	\mathrm{W}[j]_{a'_4j_4}
	\Gamma[\Phi]_{x'_1a'_3a'_4}
	\mathrm{W}[j]_{x_1 x'_1}
	\Gamma[\Phi]_{a'_1a'_2x_1}
	\mathrm{W}[j]_{j_1 a'_1}\\
	&+
	\mathrm{W}[0]_{a'_2 j_2}
	\mathrm{W}[0]_{a'_3 j_3}
	\Gamma[0]_{x'_1a'_3a'_4}
	\mathrm{W}[0]_{x_1 x'_1}
	\mathrm{W}[j]_{a'_4j_4}
	\Gamma[\Phi]_{a'_1x_1a'_4}
	\mathrm{W}[j]_{j_1 a'_1}
\end{align*}
which (after translating notational conventions) is just $\mathrm{W}[j]_{j_1 j_2 j_3j_4}$. Diagrammatically this expression is denoted
$$
	\ExampleTreeBb + \ExampleTreeCc + \ExampleTreeDd
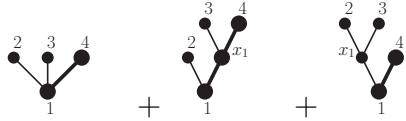
$$

\begin{rmk}\label{rmk:shadow}{\rm{
The following diagrammatical interpretation can be given to rules \eqref{item:1} and \eqref{item:2}. For any tree $T$ the rightmost branch, i.e., the unique path form the root of $T$ to the rightmost leaf of $T$ is depicted by thickened edges and vertices. The corresponding terms along this paths, i.e., $\mathrm{W}[j]_{\cdots}$ and $\mathrm{\Gamma}[\Phi]_{\cdots}$, depend on $j$ and $\Phi$, respectively. All edges and vertices in the {\it{shadow}} of the rightmost branch, i.e., those that are strictly to the left of the unique path form the root of $T$ to the rightmost leaf of $T$, are depicted normal. The corresponding terms $\mathrm{W}[0]_{\cdots}$ and $\mathrm{\Gamma}[c]_{\cdots}$, consist of $c$-numbers, and are therefore mapped to zero by functional derivations.}}
\end{rmk}

\smallskip
   		
These results hold in general.

\begin{thm}
The Feynman rule $F_\gamma$ computes the functional derivative of $\mathrm{W}[j]$. That is 
$$
	\mathrm{W}[j]_{j_1 \cdots j_k} = \sum_{T \in \mathcal{T}^{red}_{k-1}} F_\gamma(T)
$$
where the sum runs over all reduced trees with $k-1$ leaves. 
\end{thm}		

To avoid the introduction of cumbersome notation, we will give a graphical proof on a generic example using induction on $k$. It is enough to prove that
\begin{align*}
	\frac{\partial}{\partial_{j_{k+1}}} \mathrm{W}[j]_{j_1 \cdots j_k}  
	&= \sum_{T} \frac{\partial}{\partial_{j_{k+1}}} F_\gamma(T)\\
	&= \sum_{\tilde{T} \in \mathcal{T}^{red}_{k}} F_\gamma(\tilde{T}), 
\end{align*}
where the last sum runs over reduced trees $\tilde{T}$ with $k$ leaves. Let us consider the following example of the tree
$$
	\ExpTreeCa
$$
which under rules \eqref{item:1} and \eqref{item:2} is mapped to the tree 
$$
	\ExpTreeCaa
$$
corresponding to the expression
$$
	\mathrm{W}[0]_{a'_2 j_2}  	
	\mathrm{W}[0]_{a'_3 j_3}
	\mathrm{W}[j]_{a'_4 j_4} 
	\Gamma[\Phi]_{x'_1a'_3a'_4}
	\mathrm{W}[j]_{x_1 x'_1}
	\Gamma[\Phi]_{a'_1a'_2x_1}
	\mathrm{W}[j]_{j_1a'_1}
$$
Let us list case by case, where the functional derivation $\frac{\partial}{\partial_{j_{s}}}$ acts on non-constant terms, i.e., those with a $j$- (or $\Phi$-) dependence. Using the rules for the non-commutative calculus in \cite{cvitanovic2}, we first consider the case, where the functional derivation acts on the non-constant $\mathrm{W}[j]_{a_4 j_4}$ term. This yields
$$
	\mathrm{W}[0]_{a'_2 j_2} 
	\mathrm{W}[0]_{a'_3 j_3}
	\mathrm{W}[0]_{a'_4 j_4} 
	\mathrm{W}[j]_{a'_5 j_5}
	\Gamma[\Phi]_{x'_1a'_4a'_5}
	\mathrm{W}[j]_{x_1 x'_1}
	\Gamma[\Phi]_{x'_2a'_3x_1}
	\mathrm{W}[j]_{x_2 x'_2}
	\Gamma[\Phi]_{a'_1a'_2x_2}
	\mathrm{W}[j]_{j_1a'_1} 
$$
This expression is represented by the planar rooted tree with thickened rightmost branch
$$
	\ExpTreeCda
$$
The next case we consider is, when the functional derivation $\frac{\partial}{\partial_{j_{s}}}$ acts on $\Gamma[\Phi]_{x'_1a'_3a'_4}$   (using that $\frac{\partial}{\partial_{j_i}} =  \mathrm{W}[j]_{il} \frac{\partial}{\partial_{\Phi_l}}$). This generates the product 
$$
	\mathrm{W}[0]_{a'_2 j_2}
	\mathrm{W}[0]_{a'_3 j_3}  
	\mathrm{W}[0]_{a'_4 j_4} 
	\mathrm{W}[j]_{a'_5 j_5} 
	\Gamma[\Phi]_{x'_1a'_3a'_4a'_5}
	\mathrm{W}[j]_{x_1 x'_1}
	\Gamma[\Phi]_{a'_1a'_2x_1}
	\mathrm{W}[j]_{j_1a'_1}
$$
The corresponding rooted tree is 
$$
	\ExpTreeCba
$$
Let us now consider the case, where the functional derivation $\frac{\partial}{\partial_{j_{s}}}$ acts on the non-constant term $\mathrm{W}[j]_{x_1 x'_1}$. This yields
$$
	\mathrm{W}[0]_{a'_2 j_2} 
	\mathrm{W}[0]_{a'_3 j_3} 
	\mathrm{W}[0]_{a'_4 j_4}
	\mathrm{W}[0]_{x_1 x'_1}
	\Gamma[0]_{x'_1a'_3a'_4}
	\mathrm{W}[j]_{a'_5 j_5} 
	\Gamma[\Phi]_{x'_2x_1a'_5}
	\mathrm{W}[j]_{x_2 x'_2}
	\Gamma[\Phi]_{a'_1a'_2x_2}
 	\mathrm{W}[j]_{j_1a'_1}
$$
The corresponding rooted tree is 
$$
	\ExpTreeCca
$$
The next case is, where the functional derivation $\frac{\partial}{\partial_{j_{s}}}$ acts on the non-constant $\Gamma[\Phi]_{a'_1a'_2x_2}$ term (again using that $\frac{\partial}{\partial_{j_i}} =  \mathrm{W}[j]_{il} \frac{\partial}{\partial_{\Phi_l}}$). This yields
$$
	\mathrm{W}[0]_{x_1 x'_1}
	\mathrm{W}[0]_{a'_2 j_2} 
	\mathrm{W}[0]_{a'_3 j_3} 
	\mathrm{W}[0]_{a'_4 j_4} 
	\Gamma[0]_{x'_1a'_3a'_4}
	\mathrm{W}[j]_{a'_5 j_5} 
	\Gamma[\Phi]_{a'_1a'_2x_1a'_5}
	\mathrm{W}[j]_{j_1a'_1}
$$
corresponding to the planar rooted tree with thickened rightmost branch
$$
	\ExpTreeCea
$$
The last case is, where the functional derivation $\frac{\partial}{\partial_{j_{s}}}$ acts on the non-constant term $\mathrm{W}[j]_{j_1 a'_1}$, which yields
$$
	\mathrm{W}[0]_{a'_3 j_3} 
	\mathrm{W}[0]_{a'_4 j_4} 
	\Gamma[0]_{x'_1a'_3a'_4}
	\mathrm{W}[0]_{x_1 x'_1}
	\mathrm{W}[0]_{a'_2 j_2} 
	\Gamma[0]_{x'_2a'_2x_1}
	\mathrm{W}[0]_{x_2 x'_2}
	\mathrm{W}[j]_{a'_5 j_5} 
	\Gamma[\Phi]_{a'_1x_2a'_5}
	\mathrm{W}[j]_{j_1a'_1}
$$
corresponding to the planar rooted tree with thickened rightmost branch
$$
	\ExpTreeCfa
$$

From this example we deduce a general rule: functional differentiation of $F_\gamma(T)$, where $T$ is a reduced tree with, say, $k-1$ leaves, creates all the $F_\gamma(T')$, where $T'$ are reduced trees with $k$ leaves, and from which the theorem follows.

\smallskip

Observe how the functional derivation $\frac{\partial}{\partial_{j_{s}}}$ adds always a new leaf to the original rooted tree. In fact, we can represent this functional derivation in terms of a grafting operation on planar rooted trees with thickened rightmost branch. First we define the set (vector space) of all -- decreasing -- planar rooted trees -- with labeled leaves and -- with thickened rightmost branch by $T^r$ ($\mathcal{T}^r$). On $\mathcal{T}^r$ we define the operator $\mathfrak{r}_\bullet$, which augments the number of leaves by one. It maps every tree to a finite sum of trees following the rules: 
\begin{enumerate}

\item
For any tree $T \in T^r$ with $m$ internal vertices (decorated in decreasing order along any path from the root to any leaf of $T$ by the -- unprimed -- elements from $X_m:=\{x_1,x'_1,\ldots,x_m,x'_m\}$) and $n-1$ leaves (labeled in increasing order from left to right by elements from $[2,\ldots,n]$) the operator $\mathfrak{r}_\bullet$ starts at the root of $T$ (which is indexed by $1$), and connects it to a new thickened root vertex. The new root is indexed by $1$ and the label of the original root is replaced by the decoration by the element $x_{m+1}$. The thickened edges and vertices along the rightmost branch of $T$ are replaced by usual edges and vertices. This new decreasing rooted tree $T'$ with thickened root is concatenated from the right by a thickened vertex labeled by $n+1$, which is then connected to the root of $T'$ by a thickened edge. The result is a decreasing planar rooted tree $\mathfrak{r}_\bullet(T) \in T^r$, with a thickened rightmost branch, $n$ leaves, and with $m+1$ inner vertices decorated by $X_{m+1}:=\{x_1,x'_1,\ldots,x_{m+1},x'_{m+1}\}$. For instance
$$
	\mathfrak{r}_\bullet\Big(\hspace{-0.3cm}\begin{array}{c}\ExampleTreeAa \end{array}\hspace{-0.6cm}\Big)
	= \begin{array}{c}\ExampleTreeDd \end{array} + \cdots \qquad\
	\mathfrak{r}_\bullet\Big(\hspace{-0.3cm}\begin{array}{c}\ExpTreeCaa \end{array}\hspace{-0.6cm}\Big)
	= \begin{array}{c}\ExpTreeCfa \end{array} + \cdots 
$$

\item
Then the operator $\mathfrak{r}_\bullet$ is applied to the edges and vertices along the rightmost branch except for its leaf in the following way
\begin{itemize}

\item[-] When $\mathfrak{r}_\bullet$ is applied to a vertex $v$ (decorated by, say, $x_l \in X_m$), which is different from the root and leaf of the rightmost -- thickened -- branch of the tree $T \in T^r$, then it grafts a new leaf to this vertex, i.e., the tree $T$ is concatenated on its right by a thickened vertex which is then connected to the vertex $v$ by a thickened edge. The decoration of $v$ is unchanged. Hence this new vertex is a new leaf labeled by $n+1$. The edges and vertices of the part of the thickened rightmost branch of $T$ starting at $v$ are replaced by ordinary edges and vertices.  
$$
	\mathfrak{r}_\bullet\Big(\hspace{-0.3cm}\begin{array}{c}\ExampleTreeAa \end{array}\hspace{-0.6cm}\Big)
	= \begin{array}{c}\ExampleTreeBb \end{array} + \cdots \quad
	\mathfrak{r}_\bullet\Big(\hspace{-0.3cm}\begin{array}{c}\ExpTreeCaa \end{array}\hspace{-0.6cm}\Big)
	= \begin{array}{c}\ExpTreeCba \end{array}+ \cdots 
$$
$$
	\mathfrak{r}_\bullet\Big(\hspace{-0.3cm}\begin{array}{c}\ExpTreeCaa \end{array}\hspace{-0.6cm}\Big)
	= \begin{array}{c} \ExpTreeCea\end{array}+ \cdots 
$$

\item[-] When $\mathfrak{r}_\bullet$ is applied to an edge of the rightmost -- thickened -- branch of the tree $T \in T^r$, then it splits this edge by putting a thickened vertex on it, and it grafts a new leaf, which is labeled by $n+1$, to this new vertex. The thickened edges and vertices to the left of this new leaf are replaced by ordinary ones. Regarding the decreasing decoration, the following shift is applied. Note that the edge on which $\mathfrak{r}_\bullet$ just acted starts in a vertex decorated by $x_s \in X_m$, and it ends either in a leaf or it ends in a vertex $v$ of $T$, which is decorated by, say, $x_l \in X_m$, $l<s$. First the set of decorations $X_m$ is replaced by $X_{m+1}$, and the decorations $x_{s}, \ldots , x_m$ in $T$ are all shifted by one, i.e., $x_{s} \to x_{s+1}, \ldots , x_{m} \to x_{m+1}$. The new vertex put on the edge coming out of $v$ becomes a child of $v$, and is decorated by $x_{s}$. In case the decorating set was empty, the new vertex is decorated by $x_1 \in X_1$.

$$
	\mathfrak{r}_\bullet\Big(\hspace{-0.3cm}\begin{array}{c}\ExampleTreeAa \end{array}\hspace{-0.6cm}\Big)
	= \begin{array}{c} \ExampleTreeCc \end{array}+ \cdots \qquad\
	\mathfrak{r}_\bullet\Big(\hspace{-0.3cm}\begin{array}{c}\ExpTreeCaa  \end{array}\hspace{-0.6cm}\Big)
	= \begin{array}{c}\ExpTreeCca \end{array}+ \cdots 
$$

\end{itemize}

Hence, the functional derivation $\frac{\partial}{\partial_{j_{s}}}$ can be interpreted as a certain right-growth operation on decreasing planar rooted trees defined by grafting a new leaf on the right. The special nature of the action of $\mathfrak{r}_\bullet$ on the root is due to the fact that we prefer to work with rooted trees instead of planted planar rooted trees. The single extra edge going out of the root of the latter would explicitly account for the term $\mathrm{W}[j]_{j_1a'_1}$ on which the functional derivation $\frac{\partial}{\partial_{j_{s}}}$ acts. Going back to Remark \ref{rmk:shadow} we understand now that the part of the tree that lies in the shadow of the rightmost branch is the part that is not accessible to the right-growth operation, and hence it corresponds to terms $\mathrm{W}[0]_{\cdots}$ and $\mathrm{\Gamma}[0]_{\cdots}$, consisting of $c$-numbers. 

\end{enumerate}

%%%%%%%%%%%%%%%%%%%%%%%%%%%%%%%%%%%%%%%
%%%%%%%%%%%%%%%%%%%%%%%%%%%%%%%%%%%%%%%
%%%%%%%%%%%%%%%%%%%%%%%%%%%%%%%%%%%%%%%
%%%%%%%%%%%%%%%%%%%%%%%%%%%%%%%%%%%%%%%

%%%%%%%%%%%%%%%%%%%%%%%%%%%%%%%%%%%%%%%

%%%%%%%%%%%%%%%%%%%%%%%%%%%%%%%%%%%%%%%%%%%%%

\begin{thebibliography}{99999999999}

\bibitem{beissinger}
	J.S.~Beissinger,
	{\textit{The Enumeration of Irreducible Combinatorial Objects}},
	J.~Comb.~Theory, Ser.~A {\bf{38}}, (1985) 143--169.

\bibitem{belinschi}
	S.~Belinschi, M.~Bozejko, F.~Lehner, R.~Speicher,
	{\it{The normal distribution is $\boxplus$-infinitely divisible}}
	Adv.~Math. {\bf{226}}, (2011) 3677--3698.

\bibitem{berflasal}
	F.~Bergeron, P.~Flajolet, B.~Salvy,
	{\it{Varieties of increasing trees}}
	in CAAP '92, Lecture Notes in Computer Science {\bf{581}}, (1992) 24--48.

\bibitem{biane} 
	P.~Biane,
	{\it{Free probability and combinatorics}},
	Proceedings of the International Congress of Mathematicians, 
	Vol.~II (Beijing, 2002), 765--774. Higher Ed. Press, Beijing, 2002. 

\bibitem{brouder1}
	C.~Brouder,
	{\it{Quantum field theory meets Hopf algebra}},
	Mathematische Nachrichten {\bf{282}}, (2009) 1664--1690.

\bibitem{bp} 
	C.~Brouder, F.~Patras, 
	{\it{Decomposition into one-particle irreducible Green functions in many-body physics}},
	Combinatorics and Physics, Eds.~Ebrahimi-Fard, Marcolli, van Suijlekom, 
	Contemporary Mathematics {\bf{539}}, (2011) 1--25.

\bibitem{cartier1}
        P.~Cartier,  
        {\it{A primer of Hopf algebras}}, 
        In ``Frontiers in Number Theory, Physics, and Geometry II'',
        Springer Berlin Heidelberg, (2007) 537--615.
	
\bibitem{cartier2} 			
	P.~Cartier,
	{\textit{Vinberg algebras, Lie groups and combinatorics}}, 
	Clay Mathematical Proceedings {\bf{11}}, (2011) 107--126.
	
\bibitem{caswellkennedy}
 	W.~E.~Caswell, A.~D.~Kennedy,
   	{\textit{A simple approach to renormalisation theory}},
    	Phys.~Rev.~D~{\bf{25}}, (1982) 392--408.

\bibitem{chappat}  
	F.~Chapoton, F.~Patras, 
	{\emph{Enveloping algebras of preLie algebras, Solomon idempotents and the Magnus formula}},
	International Journal of Algebra and Computation {\bf{23}}, No. 4 (2013) 853--861.
	
\bibitem{collins}
	J.~Collins,
	Renormalization,
	Cambridge monographs in mathematical physics, Cambridge (1984).

\bibitem{ck0}
	A.~Connes, D.~Kreimer,
	{\textit{Hopf Algebras, Renormalization and Noncommutative Geometry}},
	Commun.~Math.~Phys. {\bf{199}}, (1998) 203--242.

\bibitem{ck1}
	A.~Connes, D.~Kreimer,
	{\textit{Renormalization in quantum field theory and the Riemann--Hilbert problem I: 
	The Hopf algebra structure of graphs and the main theorem}},
	Commun.~Math.~Phys. {\bf210}, (2000) 249--273.

\bibitem{ck2}
	A.~Connes, D.~Kreimer,
	{\textit{Renormalization in quantum field theory and the Riemann--Hilbert problem II: 
	The $\beta$-Function, Diffeomorphisms and the Renormalization Group}},
	Commun.~Math.~Phys. {\bf{216}}, (2001) 215--241

\bibitem{cvitanovic1}
	P.~Cvitanovic,
	{\textit{Planar perturbation expansion}},
	Phys. Lett. B {\bf{99}}, (1981) 49--52.

\bibitem{cvitanovic2}
	P.~Cvitanovic, P.G.~Lauwers, P.N.~Scharbach,
	{\textit{The planar sector of field theories}}, 
	Nucl. Phys. B {\bf{203}}, (1982) 385--412.

\bibitem{douglas}
	M.~Douglas,
	{\textit{Stochastic Master Fields}},
	Phys.~Lett.~B {\bf{344}}, (1995) 117--126

\bibitem{drmota}
	M.~Drmota,
	{\textit{Random Trees, An interplay between Combinatorics and Probability}},
	Springer, Wien, 2009.

\bibitem{EGP}
	K.~Ebrahimi-Fard, J. M.~Gracia-Bond{\'\i}a, F.~Patras, 
	{\it{A Lie theoretic approach to renormalization}}, 
	Commun.~Math.~Phys. {\bf{276}}, (2007) 519--549.

\bibitem{EGP2}
	K.~Ebrahimi-Fard, J.~M.~Gracia-Bond\'\i a, F.~Patras, 
	{\it{Rota--Baxter algebras and new combinatorial identities}}, 
	Lett.~Math.~Physics {\bf{81}}, (2007) 61--75. 

\bibitem{EM1} 
	K.~Ebrahimi-Fard, D.~Manchon, 
	{\emph{Dendriform equations}}, 
	Journal of Algebra {\bf{322}}(11), (2009) 4053--4079.

\bibitem{EM2} 
	K.~Ebrahimi-Fard, D.~Manchon
	{\emph{A Magnus- and Fer-type formula in dendriform algebras}}, 
	Foundations of Computational Mathematics {\bf{9}}, (2009) 295--316.
	
\bibitem{EMP}
	K.~Ebrahimi-Fard, D.~Manchon, F.~Patras, 
	{\it{A noncommutative Bohnenblust--Spitzer identity for Rota--Baxter algebras solves Bogolioubov's recursion}}, 
	Journal of Noncommutative Geometry {\bf{3}}, Issue 2 (2009), 181--222.

\bibitem{EP1}
	K.~Ebrahimi-Fard, F.~Patras,
	{\emph{Cumulants, free cumulants and half-shuffles}}, 
	Proc. R. Soc. A {\bf{471}} Issue: 2176, January 2015.

\bibitem{EP2}
	K.~Ebrahimi-Fard, F.~Patras,
	{\emph{The splitting process in free probability}},
	accepted for publication in IMRN, June 2015, arXiv:1502.02748.  				

\bibitem{ELundMan}
	K.~Ebrahimi-Fard, A.~Lundervold, D.~Manchon,
	{\textit{Noncommutative Bell polynomials, quasideterminants and incidence Hopf algebras}}
	International Journal of Algebra and Computation {\bf{24}}, 5 (2014) 671--705.
        
\bibitem{foissy}
	L.~Foissy,
	{\emph{Bidendriform bialgebras, trees, and free quasi-symmetric functions}}, 
	J.~Pure Appl.~Algebra {\bf{209}}, no.~2  (2007) 439--459.
	
\bibitem{foipat}	
	L.~Foissy, F.~Patras, 
	{\emph{Natural endomorphisms of shuffle algebras}}, 
	Int. J. Algebra and Computation {\bf{23}}, no.~4 (2013) 989--1009. 

\bibitem{gopagross}		
	R.~Gopakumar, D.J.~Gross,
	{\emph{Mastering the Master Field}},
	Nucl. Phys. B {\bf{451}}, (1995) 379--415.

\bibitem{GBFV}
	J.~M.~Gracia-Bond\'{i}a, J.~C.~V\'arilly, and H.~Figueroa,
	Elements of Noncommutative Geometry,
	Birkh\"auser, Boston, 2001.
		
\bibitem{itzyksonzuber}
  	C.~Itzykson and J.-B.~Zuber,
  	Quantum Field Theory,
  	McGraw-Hill (1980).
			
\bibitem{manchon1}
	D.~Manchon,
	{\textit{A short survey on pre-Lie algebras}}, 
	E.~Schr\"odinger Institut Lectures in Math. Phys., ``Noncommutative Geometry and Physics: 
	Renormalisation, Motives, Index Theory", Eur. Math. Soc, A.~Carey Ed. (2011).

\bibitem{manchon2}
	D.~Manchon,
	{\textsl{Hopf algebras and renormalisation}},
	Handbook of algebra {\bf{5}} (M. Hazewinkel ed.)  (2008) 365--427.

\bibitem{mastnaknica}
	M.~Mastnak, A.~Nica, 
	{\it{ Hopf algebras and the logarithm of the $S$-transform in free probability}}, 
	Transactions of the American Mathematical Society {\bf{362}} (7), (2010) 3705--3743.

\bibitem{mestreoeckl}
	\^A.~Mestre, R.~Oeckl, 
	{\textit{Combinatorics of n-point functions via Hopf algebra in quantum field theory}}, 
	Journal of Mathematical Physics {\bf{47}}, (2006) 052301. 

\bibitem{neuspeicher}
	P.~Neu, R.~Speicher,
	{\it{A self-consistent master equation and a new kind of cumulants}},
	Z.~Phys.~B {\bf{92}}, (1993) 399--407.

\bibitem{nicaspeicher}
	A.~Nica, R.~Speicher, 
	Lectures on the combinatorics of free probability,
	London Mathematical Society Lecture Note Series {\bf{335}}, 
	Cambridge University Press, 2006. 

\bibitem{novaksniady}
	J.~Novak, P.~Sniady
	{\it{What is a Free Cumulant?}},
	Notices of the American Mathematical Society {\bf{58}} (2), (2011) 300--301.

\bibitem{reutenauer} 
	C.~Reutenauer, 
	Free Lie algebras, 
	Oxford University Press, 1993.

\bibitem{speed}
	T.~P.~Speed, 
	{\it Cumulants and Partition Lattices}, 
	Austral.~J.~Statist. {\bf 25}, 2, (1983) 378--388.

\bibitem{speicher} 
	R.~Speicher, 
	{\it{Free probability theory and non-crossing partitions}}, 
	S\'em.~Lothar.~Combin.~{\bf{39}}, (1997) 38.

\bibitem{sweedler} 
	M.~E.~Sweedler,
       	Hopf algebras,
       	Benjamin, New-York (1969).

\bibitem{Stanley}
	R.~Stanley,
	{\emph{Hipparchus, Plutarch, Schr\"oder, and Hough}}, 
	American Mathematical Monthly {\bf{104}}(4)  (1997) 344.%--350,

\bibitem{voiculescu1} 
	D.~Voiculescu, K.~Dykema, A.~Nica, 
	Free random variables, CRM Monograph Series {\bf{1}}, AMS, Providence, RI, (1992).
                  
\bibitem{voiculescu2} 
	D.~Voiculescu, 
	{\it{Free Probability Theory: Random Matrices and von Neumann Algebras}}, 
	Proceedings of the International Congress of Mathematicians, 
	Z\"urich, Switzerland 1994. Birkh\"auser Verlag, Basel, Switzerland (1995).
         
\bibitem{zavialov}
	O.~I.~Zavialov,
	Renormalized Quantum Field Theory,
	Kluwer Acad.~Publ.~(1990).
        
\end{thebibliography}
\end{document}